\documentclass[review]{elsarticle}
\usepackage{hyperref}
\usepackage{float}
\usepackage{verbatim} 
\usepackage{apalike}
\usepackage{multicol}
\usepackage{graphicx}
\usepackage{subcaption}
\usepackage{booktabs}
\usepackage{xfrac}
\usepackage[dvipsnames]{xcolor}

\usepackage{amsthm}
\usepackage{multirow}
\usepackage{tikz}
\usepackage{tabu}
\usepackage[compat = 0.6]{yquant}
\usepackage{geometry}
\usepackage{physics}
\usepackage{amsmath}
\usepackage{amsfonts}
\DeclareMathOperator*{\argmax}{\arg\!\max}

\makeatletter
\renewcommand*\env@matrix[1][\arraystretch]{%
  \edef\arraystretch{#1}%
  \hskip -\arraycolsep
  \let\@ifnextchar\new@ifnextchar
  \array{*\c@MaxMatrixCols c}}
\makeatother

\newtheoremstyle{tight}
  {0.5em}   
  {-1em}   
  {\itshape}  
  {}        
  {\bfseries} 
  {.}       
  { }       
  {}        

\theoremstyle{tight}

\newtheorem{theorem}{Theorem}[section]
\newtheorem{lemma}[theorem]{Lemma}
\newtheorem{definition}[theorem]{Definition}

\newtheorem{assertion}[theorem]{Assertion}
\newtheorem{corollary}[theorem]{Corollary}

\usepackage{tikz}
\usepackage{circuitikz}
\usepackage[normalem]{ulem}
\usetikzlibrary{positioning}
\usetikzlibrary{positioning,shapes,arrows}
\usepackage{wrapfig}
\restylefloat{figure}
\restylefloat{table}

\journal{Expert Systems with Applications}

\biboptions{authoryear}

\begin{document}
\begin{frontmatter}


\title{Quantum Bayesian Networks Can Speed up Reinforcement Learning in Partially Observable Environments}
\author[b,c]{Gilberto Cunha}

\author[a,b,c]{Alexandra Ramôa\corref{cor1}}

\author[a,b,c]{André Sequeira}

\author[a,c]{Michael de Oliveira}

\author[a,b,c]{Luís Barbosa}

\cortext[cor1]{Corresponding author: alexandra.ramoa@inl.int.}

\address[a]{High-Assurance Software Laboratory (HASLab), INESC TEC, Portugal}
\address[b]{Department of Computer Science, University of Minho, Portugal}
\address[c]{International Iberian Nanotechnology Laboratory (INL), Portugal}

\begin{abstract}
Reinforcement learning (RL) provides a principled framework for decision-making in partially observable environments, which can be modeled as Markov decision processes and compactly represented through dynamic decision Bayesian networks. Recent advances demonstrate that inference on sparse Bayesian networks can be accelerated using quantum rejection sampling combined with amplitude amplification, leading to a computational speedup in estimating acceptance probabilities.\\
Building on this result, we introduce Quantum Bayesian Reinforcement Learning (QBRL), a hybrid quantum-classical look-ahead algorithm for model-based RL in partially observable environments. We present a rigorous, oracle-free time complexity analysis under fault-tolerant assumptions for the quantum device. Unlike standard treatments that assume a black-box oracle, we explicitly specify the inference process, allowing our bounds to more accurately reflect the true computational cost. We show that, for environments whose dynamics form a sparse Bayesian network, horizon-based near-optimal planning can be achieved sub-quadratically faster through quantum-enhanced belief updates. On the other hand, we show that there is no quantum speed-up for environments that are either fully observable, or characterized by Bayesian networks whose maximum in-degree is not small.\\
Furthermore, we present numerical experiments benchmarking QBRL against its classical counterpart on simple yet illustrative decision-making tasks. Our results offer a detailed analysis of how the quantum computational advantage translates into decision-making performance, highlighting that the magnitude of the advantage can vary significantly across different deployment settings.
\end{abstract}


\end{frontmatter}
\newpage

\section{Introduction}
\label{sec:intro}

Reinforcement learning (RL) is a foundational paradigm for sequential decision-making under uncertainty, enabling agents to learn optimal strategies by interacting with and adapting to dynamic environments. It formalizes decision-making as a reward optimization problem, and plays a central role in artificial intelligence (AI), especially in scenarios where explicit supervision is unavailable or partial. RL has been pivotal to some of the most significant advances in modern AI, including AlphaGo’s mastery of the game of Go \cite{silver2016mastering} and more recently, the outstanding performance improvement of large language models \cite{wangReinforcementLearningEnhanced2025}. RL has also shown impact in domains such as finance, robotics, natural language processing, and metrology \cite{fiderer2021,Nian_2020,Bai_2024}. These successes highlight RL’s versatility, but also its computational demands—particularly in high-dimensional, partially observable, or data-scarce settings—where improvements in sample efficiency and computational scalability are critical.
 
Quantum advantages for RL have been investigated in both tabular and approximate settings \cite{meyer2024}. In the tabular regime, \cite{dunjkoExponentialImprovementsQuantumaccessible2018a} suggests that Quantum Reinforcement Learning (QRL) may yield polynomial to exponential improvements under various conditions. This result has been further substantiated by recent contributions: \cite{Sannia_2023} employed quantum amplitude amplification to accelerate Q-value estimation in Q-learning; \cite{cherratQuantumReinforcementLearning2023} introduced a QRL framework based on policy iteration, showing that quantum algorithms for policy evaluation and improvement can achieve meaningful speedups; and \cite{gangulyQuantumSpeedupsRegret2024} demonstrated that logarithmic regret is achievable in model-based episodic and average-reward Markov Decision Processes (MDPs), suggesting theoretical benefits for quantum agents in structured environments.
 
In the approximate setting, quantum-enhanced policies based on parameterized quantum circuits (PQCs) have also shown promise. Empirical studies indicate that PQC-based models can improve agent convergence in various scenarios \cite{jerbiParametrizedQuantumPolicies2021,skolikQuantumAgentsGym2022a}. \cite{jerbiQuantumPolicyGradient2022} established that PQC-based policies can be trained quadratically faster when granted full access to quantum environments. Most recently, \cite{xuAcceleratingQuantumReinforcement2025} demonstrated that PQC-based model-free agents for infinite-horizon MDPs can attain a computational speedup that exceeds known classical lower bounds—underscoring the concrete benefits of fully quantum reinforcement learning approaches.

However, none of these advances explicitly addresses partially observable environments, which are significantly more complex than fully observable ones (PSPACE-complete \cite{papadimitriouComplexityMarkovDecision1987}), and thus more computationally challenging. Tackling this problem is paramount, as partial observability arises in a multitude of real-world scenarios, due to noisy sensors or otherwise limited information \cite{muskardin2024}. In practice, sensors are often inaccurate, or operate under conditions that constrain their ability to perceive the environment, such as limited visibility. 

In such environments, RL problems are typically modeled as POMDPs and can be compactly encoded as dynamic decision Bayesian networks (DDNs). A central challenge lies in efficiently performing belief updates and inference over these probabilistic graphical models.
 
Quantum computing offers promising tools to tackle these difficulties. In particular, quantum rejection sampling combined with amplitude amplification \cite{grover1996fast,brassard2002,low2014quantum} provides a quadratic speedup in the success probability of sampling for sparse Bayesian networks. Additionally, \cite{de2021quantum} introduced an algorithm that leverages quantum amplitude amplification to obtain provable quantum advantages in static decision-making tasks—without requiring explicit evaluation of conditional probabilities, while \cite{gao2022enhancing} demonstrated that quantum correlations enhance the expressive power of generative models such as Bayesian networks. Their findings suggest that minimal quantum extensions may improve sample complexity and convergence of quantum reinforcement learning agents given partially observable dynamics.
 
Building on this result, we introduce a hybrid quantum-classical look-ahead algorithm for model-based RL in partially observable environments. We present a rigorous, oracle-free time complexity analysis under fault-tolerant assumptions for the quantum device. Unlike standard treatments that assume a black-box oracle, we explicitly specify the inference process, allowing our bounds to more accurately reflect the true computational cost. We show that, for environments whose dynamics form a sparse Bayesian network, horizon-based near-optimal planning can be achieved sub-quadratically faster through quantum-enhanced belief updates. We additionally present numerical experiments on simple yet illustrative decision-making tasks, showing how the quantum algorithm outperforms its classical counterpart for a fixed cost, or conversely reduces the runtime for the same performance. Furthermore, we demonstrate and discuss the dependence of this advantage on the practical task at hand. Our results highlight that the magnitude of the advantage can vary significantly across different deployment settings.

The quantum component of our algorithm handles probabilistic inference for Bayesian networks, which model the reinforcement learning environment. The task is to estimate the conditional distributions that are necessary for belief updates. In partially observable settings, some of these estimates rely on rejection sampling. The acceptance probabilities typically shrink exponentially with the number of variables, making this a bottleneck \cite{low2014quantum}. For sparse networks, quantum amplitude amplification can improve the efficiency of the rejection sampling sub-routine by boosting the acceptance probability \cite{low2014quantum}. This has the potential to speed-up reinforcement learning in some scenarios, as we rigorously analyze and numerically demonstrate. 

In short, our contribution is to leverage a known quantum speed-up in rejection sampling for Bayesian networks and investigate how it can translate into a quantum speed-up for reinforcement learning. This includes a complexity analysis, numerical experiments, and a discussion of advantages and limitations. 

We note that the term ``Bayesian reinforcement learning'' typically refers to algorithms where Bayesian statistics are used to address the exploration/exploitation trade-off \cite{Ghavamzadeh_2015}. These are not directly related to our work, where Bayesian networks are a tool to represent the environment. The relevant comparison is with the classical counterpart of our algorithm; accordingly, that is used as a reference throughout this work. We note the existence of other solutions to POMDPs \cite{silver2010pomcp, moss2023betazero, smith2012hsvi, pineau2003pbvi, spaan2005perseus}, some of which can be more efficient. It would be interesting to study whether quantum speed-ups apply to these alternatives. 

\emph{Paper structure}. The rest of the paper is organized as follows. Section \ref{sec:related_work} overviews related work, and clarifies where this work stands with respect to it. Section \ref{sec:background} introduces the foundational concepts necessary to understand our hybrid quantum-classical algorithm, which is described in section \ref{sec:algorithm}. Section \ref{sec:complexity-analysis} presents the computational complexity of the classical and hybrid approaches. The methodology employed for the numerical evaluation of the algorithms is detailed in section \ref{sec:methods}, with the corresponding results and discussion presented in section \ref{sec:experimental-results}. Lastly, section \ref{sec:conclusion} highlights the main findings and outlines directions for future research.
\section{Related work}
\label{sec:related_work}

This section situates the present work with respect to related research. Research at the interface of quantum computing and reinforcement learning has developed along several distinct directions; this section situates the present work with respect to these directions. For an overview, refer to \cite{meyer2024}.

One such direction studies reinforcement learning in settings where the agent has quantum access to the environment; this yields provable improvements for exploration, policy evaluation, policy iteration, and regret minimization \cite{dunjkoExponentialImprovementsQuantumaccessible2018a, wang2021quantum, cherratQuantumDeepHedging2023, gangulyQuantumSpeedupsRegret2024, jerbiQuantumPolicyGradient2022}. In contrast, our work concerns classical environments; quantum resources are used to accelerate a specific inference subroutine.

A second line of work, motivated by near-term devices, uses parameterized quantum circuits as function approximators inside otherwise classical learning loops \cite{jerbiParametrizedQuantumPolicies2021,skolikQuantumAgentsGym2022a}. Such approaches investigate expressive quantum policies and hybrid training schemes, and in some cases show empirical or sample-complexity advantages. However, they do not target the bottleneck that arises from belief-state updates in model-based planning under partial observability, which is the focal point of this work. Furthermore, we rely on non-variational, fault tolerant circuits.

Another related direction is classical reinforcement learning as a tool for quantum control. Representative examples include deep-RL methods for quantum state preparation and gate control, as well as approaches for open quantum systems \cite{Bukov_2018,Niu_2019,Ernst_2025}. This direction too is conceptually different from our setting:  reinforcement learning is the optimization tool used to solve a quantum-control problem, whereas here we consider the reverse scenario where quantum computation is a tool to accelerate reinforcement learning.  

There have also been proposals of quantum-inspired algorithms for reinforcement learning \cite{Chen_2008, kaldari2025, yu2025}. Unlike the present work, these are fully classical approaches; they are motivated by quantum mechanical principles, but do not effectively use quantum hardware.  

Our work is most directly connected to quantum probabilistic inference for Bayesian networks. It has been shown that quantum rejection sampling combined with amplitude amplification offers a quadratic improvement in the dependence on the evidence probability for sparse Bayesian networks \cite{low2014quantum}. Quantum-circuit constructions for Bayesian and dynamic Bayesian networks were subsequently developed in \cite{borujeni2021quantum}, while \cite{gao2022enhancing} showed that quantum correlations can enhance the expressive power of generative models related to Bayesian networks. We build on this inference perspective and ask how such an advantage propagates to sequential decision making when a partially observable environment is represented as a dynamic decision network.

Finally, our paper also relates to the vast literature on classical planning and learning in POMDPs. Many paths have been proposed for solving partially observable problems, such as point-based methods, online Monte Carlo planning methods, and methods leveraging learned approximations \cite{pineau2003pbvi,spaan2005perseus,smith2012hsvi,silver2010pomcp, moss2023betazero,muskardin2024}. These strategies address the same broad problem class, but the present work focuses on a different algorithmic question: we isolate the effect of replacing the classical rejection-sampling belief update inside a fixed look-ahead planner by its quantum counterpart. For that reason, the baseline throughout the paper is the classical version of the same algorithm.

Hence, our contribution is complementary to the existing literature on two topics: quantum reinforcement learning, and classical POMDPs. We do not propose a general-purpose quantum solution to all partially observable reinforcement-learning problems. Rather, we identify a specific regime in which quantum rejection sampling can reduce the cost of belief updates, and make explicit both the associated advantages and the limitations through a complexity analysis and numerical experiments.

\section{Background}
\label{sec:background}

This section assembles the theoretical ingredients that underpin the quantum-enhanced decision-making approach developed in the remainder of the article.  We first revisit POMDPs (\S\ref{subsec:pomdps}), establishing the notation that will carry through the paper and recalling how belief states encode an agent’s uncertainty about latent system dynamics.  The discussion then turns to DDNs (\S\ref{subsec:ddns}), a Bayesian-network (BN) realisation of POMDPs that exposes the conditional-independence structure required for efficient probabilistic inference.  Because our contribution leverages quantum resources, we next summarise the circuit model of quantum computation (\S\ref{subsec:qc}), and the notation used to represent quantum programs.  Finally, we outline the \textit{quantum rejection-sampling} algorithm of Low and Chuang, together with the amplitude-amplification techniques that yield a quadratic speed-up for inference in sparse BNs (\S\ref{subsec:qrs}).  Collectively, these four building blocks provide the conceptual and algorithmic scaffold for the quantum POMDP solver presented in Section \ref{sec:methods}.

\subsection{POMDPS}
\label{subsec:pomdps}

Partially observable Markov decision processes (POMDPs) are Markov decision processes (MDPs) where the state can't be directly observed by an agent. The agent, having to select an action $a_t \in \mathcal{A}$, relies instead in an observation $o_t \in \Omega$, obtained from the \textit{sensor model} $P \left( o_{t+1} \middle| s_{t+1}, a_t \right)$ of the environment, to describe the state $s_t \in \mathcal{S}$. Sets, random variables and their values are denoted by calligraphic, uppercase and lowercase letters respectively. The environment evolves over time given its \textit{transition dynamics} $P \left( s_{t+1} \middle| s_t, a_t \right)$ and the expected reward is obtained from the \textit{reward function} $\mathbb{E} \left( r_{t+1} \middle| s_t, a_t \right)$. These three quantities describe the environment's behaviour and are called the \textit{environment dynamics}.

A prior state distribution $b_0(s) = P(s)$ is used to define the initial belief state. As the agent selects actions and receives observations, it gathers a \textit{history} $h_t = \{a_0, o_1, \dots a_{t-1}, o_t\}$ of past information it can use to update its uncertainties about the environment. One way of achieving this is by defining a \textit{belief state} $b_t(s) = P \left( s \middle| h_t, b_0 \right)$, a probability distribution describing the uncertainties of the agent over the environment states, which it can't directly observe. It is also possible to sequentially update this belief state: given an action $a_t$, an observation $o_{t+1}$ and a previous belief state $b_t$, the updated belief state $b_{t+1}$ can be obtained using the \textit{belief updating rule}:
\begin{equation}
    b_{t+1}(s_{t+1}) = P \left( s_{t+1} \middle| b_t, a_t, o_{t+1} \right) \propto P \left( o_{t+1} \middle| s_{t+1}, a_t \right) \sum_{s_t \in \mathcal{S}} P \left( s_{t+1} \middle| s_t, a_t \right) b_t(s_t)
    \label{eq:belief-update}
\end{equation}

Adopting a more concise notation, denoting by $\mathcal Z $ the sensor model (conditional distribution of observations given states and actions), and dropping the subscripts in favor of primes to denote the following state, the belief update rule can be more succintly expressed as $\tau (b, a,o)(s') \equiv b'(s') \propto \mathcal Z^o_{s'a}\sum \mathcal P^a_{ss'}b(s)$.

The behaviour of an agent can be characterized by a \textit{policy} $\pi (b, a) = P \left( a \middle| b \right)$, a probability distribution over actions given belief states. The behaviour of an agent produces a sequence of rewards, which can be incorporated into a metric, called the return $G_t = \sum_{k=t}^{\infty} \gamma^{k-t} r_{k+1}$, that quantifies the desirability of that sequence. Given the stochastic nature of an environment, the goal of an agent is to maximize the \textit{expected return}, a quantity that is often used in the literature to define the value of both states and actions. The value of a belief state $b_t$ when following a policy $\pi$ is given by the \textit{state value function} $V^{\pi}(b)$, while the value of taking action $a$ when in belief state $b$ while following policy $\pi$ is given by the \textit{action value function} $Q^{\pi}(b, a)$. They are given by equations equation \ref{eq:v-value} and \ref{eq:q-value} respectively.
\vspace{1em}

\begin{minipage}{0.48\textwidth}
\begin{equation}
V^{\pi}(b) = \mathbb{E}_{\pi} \left( \sum_{k=t}^{\infty} \gamma^{k-t} r_{k+1} \; \middle| \; b \right)
\label{eq:v-value}
\end{equation}
\end{minipage}
\hfill
\begin{minipage}{0.48\textwidth}
\begin{equation}
Q^{\pi}(b, a) = \mathbb{E}_{\pi} \left( \sum_{k=t}^{\infty} \gamma^{k-t} r_{k+1} \; \middle| \; b, a \right)
\label{eq:q-value}
\end{equation}
\end{minipage}
\vspace{0.25em}

A policy $\pi^\star$ is said to be optimal if it yields value functions (also called optimal value functions) no smaller than any other policy's value functions: $V^{\star} (b) = \max_{\pi} V^{\pi}(b)$ and $Q^{\star} (b, a) = \max_{\pi} Q^{\pi}(b, a)$.

\subsection{Dynamic decision networks}
\label{subsec:ddns}
A Bayesian network (BN) is a directed, acyclic graph $\mathcal{G} = \{ \mathcal{V}, \mathcal{W} \}$ whose vertices $\mathcal{V} = \{X_1, X_2, \dots, X_N\}$ represent random variables (RVs) and the edges $\mathcal{W}$ represent dependencies between them. Each RV contains a conditional probability table (CPT) $P \left( X_i \middle| \text{par} \left( X_i \right) \right)$ describing the probability of occurrence for each value of the RV $X_i$, given the values of its parent RVs $\text{par} \left( X_i \right)$.

\begin{wrapfigure}{r}{0.4\textwidth}
\centering
\begin{tikzpicture}[
  node distance=0.8cm,
  mycircle/.style={draw,circle,text width=0.6cm,align=center},
  mydiamond/.style={draw,diamond,text width=0.4cm,align=center},
  mysquare/.style={text centered, draw,rectangle,text width=0.6cm, minimum height=0.85cm},
  myblank/.style={draw=none, fill=none}
]
\node[mycircle] (x0) {$S_t$};
\node[myblank, right=of x0] (h0) {};
\node[mycircle, right=of h0] (x1) {$S_{t+1}$};
\node[mycircle, above=of x1] (a0) {$A_t$};
\node[mycircle, above right=of x1] (e1) {$O_{t+1}$};
\node[mycircle, above left=of x1] (r1) {$R_{t+1}$};
\path 
    (x0) edge[-latex] (x1)
    (x0) edge[-latex] (r1)
    (a0) edge[-latex] (x1)
    (a0) edge[-latex] (r1)
    (a0) edge[-latex] (e1)
    (x1) edge[-latex] (e1);
\end{tikzpicture}
\caption{The simplest DDN, with one RV representing each state, action, observation, and reward.}
\label{fig:ddn}
\end{wrapfigure}
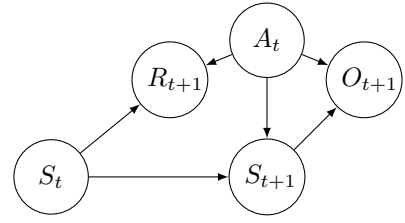

A dynamic decision network (DDN) is a BN that is constructed to capture the notion of time dynamics and agency, allowing it to model a PODMP. Its RVs describe the states $S_t$, actions $A_t$, observations $O_t$, and rewards $R_t$ for each time-step $t$ and their dependencies should mimic the environment dynamics of the underlying POMDP (see \autoref{fig:ddn}). Their CPTs represent the transition dynamics, the policy, the sensor model, and the reward function, respectively. Moreover, the CPT of the initial state $S_0$ encodes the prior belief state distribution $b_0$.

A probability distribution can be extracted from any BN (including DDNs) by using \textit{inference algorithms}, such as \textit{rejection sampling}. Therefore, DDNs can be used to model and extract probabilistic information about the dynamics of an environment.

\subsection{Look-ahead algorithm}

To approximately solve partially observable RL problems, we consider a lookahead algorithm which extracts a near-optimal action from a finite lookahead tree. The algorithm performs an exhaustive tree-based search of all possible futures of the agent within a given horizon $H$, enabling the selection of a rational action based on consideration of these future scenarios. 

The lookahead tree starts from the current belief state of the agent and expands through the enumeration of actions from the action space $\mathcal{A}$ and observations from the observation space $\Omega$. Each belief state is connected to reward nodes via an edge for each action, and each reward node is connected to a new belief node via an edge for each observation. Following the construction of the tree, the value of each node is iteratively calculated, starting with the leaves and proceeding upwards to the root node. Belief and reward node values are defined using state and action value functions, respectively.

\begin{wrapfigure}{r}{0.5\textwidth}
\centering
\scalebox{0.9}{\begin{tikzpicture}[
  node distance=1.25cm,
  bcircle/.style={draw,circle,fill=black,text width=0.2cm},
  wcircle/.style={draw,circle,text width=0.2cm},
  myblank/.style={draw=none,circle,fill=none,text width=0.2cm}
]
\node[bcircle] (b1) {};
\node[wcircle, below left=of b1] (o2) {};
\node[wcircle, left=of o2] (o1) {};
\node[myblank, below right=of b1] (h1) {$\dots$};
\node[wcircle, right=of h1] (o3) {};
\node[bcircle, below left=of o2] (b3) {};
\node[bcircle, left=of b3] (b2) {};
\node[myblank, below right=of o2] (h2) {$\dots$};
\node[bcircle, right=of h2] (b4) {};
\node[wcircle, below left=of b3] (o5) {};
\node[wcircle, left=of o5] (o4) {};
\node[myblank, below right=of b3] (h3) {$\dots$};
\node[wcircle, right=of h3] (o6) {};
\path 
    (b1) edge[-latex] node[pos=.8,above=.5mm] {$a_1$} (o1)
    (b1) edge[-latex] node[pos=.8,above=.5mm] {$a_2$} (o2)
    (b1) edge[-latex] (h1)
    (b1) edge[-latex] node[pos=.8,above=.5mm] {$a_{\left| \mathcal{A} \right|}$} (o3)
    (o2) edge[-latex] node[pos=.8,above=.5mm] {$o_1$} (b2)
    (o2) edge[-latex] node[pos=.8,above=.5mm] {$o_2$} (b3)
    (o2) edge[-latex] (h2)
    (o2) edge[-latex] node[pos=.8,above=.5mm] {$o_{\left| \Omega \right|}$} (b4)
    (b3) edge[-latex] node[pos=.8,above=.5mm] {$a_1$} (o4)
    (b3) edge[-latex] node[pos=.8,above=.5mm] {$a_2$} (o5)
    (b3) edge[-latex] (h3)
    (b3) edge[-latex] node[pos=.8,above=.5mm] {$a_{\left| \mathcal{A} \right|}$} (o6);
\end{tikzpicture}}
\caption{A two-step ($H=2$) lookahead tree. Belief and reward nodes are represented by black and white circles, respectively.}
\label{fig:lookahead-tree}
\end{wrapfigure}
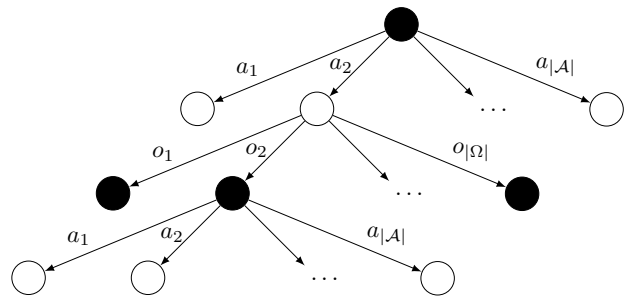

Consider $t$ the current time-step and an horizon $H$. The value of each leaf reward node is given by the expected reward given their parent belief state $b_{t+H}$ and the action $a_{t+H}$ on the edge to that leaf node:
\begin{equation}
    Q^{L}(b_{t+H}, a_{t+H}) = \mathbb{E} \left( r_{t+H+1} \middle| b_{t+H}, a_{t+H} \right)
    \label{eq:leaf-value}
\end{equation}

\noindent where we annotate the value functions with the superscript $L$ denoting the lookahead, which stops  the recursion. 

For non-leaf belief nodes, their value is determined by greedily selecting the action that leads to the highest valued reward node, which corresponds to a max operation:

\begin{equation}
    V^{L}(b_t) = \max_{a_t \in \mathcal{A}} Q^{L}(b_t, a_t)
    \label{eq:b-value}
\end{equation}

The value of a non-leaf reward node is calculated using the expected values of its child belief nodes, weighted by the probability of the observation that leads to each child. This accounts for the stochastic nature of observations in the environment, making it imperative to consider all observations in the calculation. Finally, the expected reward for that particular node is also added to the value:
\begin{equation}
    Q^{L}(b_t, a_t) = \mathbb{E} \left( r_{t+1} \middle| b_t, a_t \right) + \gamma \sum_{o_{t+1} \in \Omega} P \left( o_{t+1} \middle| b_t, a_t \right) V^{L} (b_{t+1})
    \label{eq:r-value}
\end{equation}

\noindent where $b_{t+1}$ is obtained via the belief update rule of (\ref{eq:belief-update}), given the belief state $b_t$, action $a_t$ and each of the observations $o_{t+1}$ in the summation of (\ref{eq:r-value}).

Following the calculation of values in the tree according to the aforementioned expressions, the action that should be taken is the one that maximizes value: $a^{L}_t = \argmax_{a_t \in \mathcal{A}} Q^{L}(b_t, a_t)$.

\subsection{Quantum computing with quantum circuits}
\label{subsec:qc}

A quantum circuit is a combination of different qubits, quantum gates, and measurements. The qubits encode quantum information that is transformed via the quantum gates to produce some desired result, which is converted to classical information via measurements. In $n$-qubit quantum systems, a quantum state $\ket{\psi} = (\alpha, \beta)^\top$ exist in a complex \textit{Hilbert space} $\mathcal{H} = \left( \mathbb{C}^2 \right)^{\otimes n}$ with a unit norm $\braket{\psi}{\psi} = 1$. It is common practice to use the Dirac notation, with $\bra{\psi}$ representing the complex conjugate of $\ket{\psi}$ and $\braket{\psi}{\phi}$ the inner-product of quantum states.

Quantum gates are represented by unitary operations $U: \mathcal{H} \mapsto \mathcal{H}$ transforming one quantum state into another. Some of the most important single-qubit gates are the \textit{Pauli} $X$, $Y$ and $Z$ matrices, and the \textit{phase-shift} gate $P(\theta)$, whose matrices are given below.
\begin{equation}
    X = \begin{pmatrix}[0.75]
        0 & 1\\
        1 & 0
    \end{pmatrix}, \;
    Y = \begin{pmatrix}[0.75]
        0 & -i\\
        i & 0
    \end{pmatrix}, \;
    Z = \begin{pmatrix}[0.75]
        1 & 0\\
        0 & -1
    \end{pmatrix}, \;
    P(\theta) = \begin{pmatrix}[0.75]
        1 & 0\\
        0 & e^{i\theta}
    \end{pmatrix}
    \label{eq:q-gates}
\end{equation}

Single–qubit rotations are generated by the Pauli operators
$\mathcal{P}=\{X,Y,Z\}$ through $R_U(\theta)\;=\;\exp\!\Bigl(-\tfrac{i\theta}{2}\,U\Bigr),$ $ U\in\mathcal{P}.$
More generally, any Pauli string 
$\sigma\in\{I,X,Y,Z\}^{\otimes n}$ 
defines an $n$-qubit rotation $R_\sigma(\theta)\;=\;\exp\!\Bigl(-\tfrac{i\theta}{2}\,\sigma\Bigr)$,
which can entangle the qubits whenever $\sigma$ acts non-trivially on more than one site. Besides such rotations, the controlled-NOT gate, $
\text{CNOT} =
\begin{pmatrix}
\mathbb{I}_2 & 0 \\
0 & X_2
\end{pmatrix}$, is the standard primitive for creating entanglement in quantum circuits.  Combined with single-qubit rotation gates, it results in a universal gate set for quantum computation.

\subsection{Quantum Bayesian networks}
\label{subsec:qbn}

A Quantum Bayesian Network (QBN) encodes the joint distribution
$P(X_1,\dots,X_N)$ of a classical BN into the
amplitudes of an $N$-qubit state.  Let
$\mathcal{B}$ denote a state-preparation circuit acting on
$\ket{0}^{\otimes N}$ such that
\begin{equation}
  \bigl|\braket{x_1\cdots x_N}{\Psi}\bigr|^{2}
  \;=\;
  P\!\left(X_1=x_1,\dots,X_N=x_N\right),\qquad
  \ket{\Psi}= \mathcal{B}\ket{0}^{\otimes N}.
  \label{eq:qbn-encoding-background}
\end{equation}
For binary random variables, $\mathcal{B}$ can be constructed with a
sequence of uniformly controlled $R_Y(\theta)$ rotations: each node
$X_i$ is mapped to a qubit $q_i$, and for every assignment of the
parent variables $\mathrm{parents}(X_i)$ a controlled rotation is applied to
$q_i$.  The angle
\begin{equation} 
    \theta = 2 \arctan \left( \sqrt{ \frac{ P\left( X_i = 1 \middle| \text{parents}\left(X_i\right) = x_p \right) }{ P\left( X_i = 0 \middle| \text{parents}\left(X_i\right) = x_p \right) } } \right)
    \label{eq:qbn-rotation-angle}
\end{equation}

\noindent implements the corresponding conditional-probability-table (CPT)
entry.  If $M=\max_i\ \mathrm{parents}(X_i)$ is the maximum parent count,
the computational complexity of the corresponding state preparation circuit scales as
$\mathcal{O}\!\left(N2^{M}\right)$\,\citep{low2014}.  Although this is
worse than the classical direct-sampling cost
$\mathcal{O}(NM)$, the quantum representation unlocks amplitude-based
speed-ups for subsequent inference tasks.

Importantly, the classical overhead associated with parameter calculations does not affect the quantum speed-up. The complexity of computing the angle in equation \ref{eq:qbn-rotation-angle} is $\mathcal O(1)$, since the conditional probabilities can be looked up in the CPT of the random variable $X_i$. Since the number of parent configurations is upper bounded by $2^M$, doing this for each of the $N$ qubits requires computing at most $N2^{M}$ such angles. Thus, these operations can be absorbed into the state preparation cost $\mathcal{O}\!\left(N2^{M}\right)$ for the complexity considerations.

\subsection{Quantum rejection sampling}
\label{subsec:qrs}

Given a QBN state~\eqref{eq:qbn-encoding-background}, one can extract
the conditional distribution
$P(\mathcal{Q}\mid\mathcal{E}=e)$ over query variables
$\mathcal{Q}$ provided evidence $e$, by quantum rejection sampling.  The joint state
is first partitioned as
\begin{equation}
  \ket{\Psi}
  \;=;
  \sqrt{P(e)}\,\ket{\mathcal{Q},e}
  \;+\;
  \sqrt{1-P(e)}\,\ket{\mathcal{Q},\overline{e}},
  \label{eq:qrs-decomposition}
\end{equation}
where $P(e)$ is the evidence probability.  An evidence phase-flip
operator $\mathcal{S}_e$ inverts the phase of the good subspace
$\ket{\mathcal{Q},e}$, while the Grover diffusion
operator $\mathcal{S}_0$ amplifies the states with evidence $e$.  Iterating the amplitude-amplification operator, $G =\mathcal{B}\,\mathcal{S}_0\,\mathcal{B}^{\dagger}\,\mathcal{S}_e$, $O\!\bigl(P(e)^{-1/2}\bigr)$ times boosts the success amplitude. 

For sparse networks, where each node has a small amount of parents $M$, the exponential factor $2^M$ in the quantum complexity can be approximated as $M$, and thus the quadratic quantum speedup is the main distinction between complexities. In these conditions, the total
complexity for the quantum algorithm can be written as
\(
  \mathcal{O}\!\left(
    N2^{M}\,P(e)^{-1/2}
  \right),
\)
providing a quadratic improvement in the inference over classical rejection sampling’s $\mathcal{O}\!\left(
  NM\,P(e)^{-1}
\right)$ cost \citealp{low2014}.  See Appendix \ref{app:q_rej_sampling} for a complete derivation of the quantum rejection sampling algorithm. The quantum advantage is powered by quantum amplitude amplification, where interference is used to decrease the rejection rate. An amplification operator can be defined using standard quantum operations, and repeatedly applied to boost the acceptance probability. The optimal degree of amplification depends on the probability of the evidence $P(e)$, i.e. the acceptance probability for the sampling. For $P(e) \geq 0.5$, the quantum and classical cases coincide, if using the standard definition of the amplitude amplification operator. However, this operator can be generalized to change this, as described in \cite{brassard2002}. The probability can then be raised to $100\%$ in all cases, provided that $P(e)$ is known. 

In practice, this does not hold. In  \citep{low2014}, an exponential progression proposed in \cite{brassard2002} is used to remedy this fact. While this method retains the quadratic advantage, the cost offset is not necessarily advantageous in terms of quantum resources. Instead, one may calculate the probability analytically using the Bayesian network (depending on its structure) or employ quantum amplitude estimation or even a rudimentary classical pre-estimation based on a few samples. The estimated value need not be exact; small deviations will manifest as non-optimal, but still high, acceptance rates, without jeopardizing correctness. 

Because the same construction applies when the BN is the
time-unrolled slice of a dynamic decision network, the QRS procedure
serves as the inference engine in the quantum POMDP framework
described in Section~\ref{sec:algorithm}.
\section{Quantum-classical POMDP lookahead}
\label{sec:algorithm}

In this section, we present quantum Bayesian reinforcement learning (QBRL): a classical model-based RL algorithm that incorporates the quantum rejection sampling algorithm for BNs as a subroutine. We propose quantum circuits for encoding each component of a POMDP, and show they can be used to perform a quantum belief update that is equivalent to the classical one. For proofs, as well as details on the construction of the quantum operators, the reader is referred to Appendix \ref{app:q_rej_sampling}.

Figure \ref{fig:qbn-belief-update} shows the quantum circuit that implements a belief update. Operators $\mathcal U (b)$ and $\mathcal U (a)$ encode the belief state and action, respectively. The operators $\mathcal U_1$, $\mathcal U_2$ and $\mathcal U_3$ encode the transition dynamics, the sensor model and the reward, respectively. $G^k(o)$ is the $k$-th power of the amplitude amplification operator for observation $o$, defined as in subsection \ref{subsec:qrs} from an evidence phase flip operator $\mathcal S_e$ and the Bayesian network encoding operator $\mathcal B$. Each of these operators is defined in Appendix \ref{app:q_rej_sampling}.

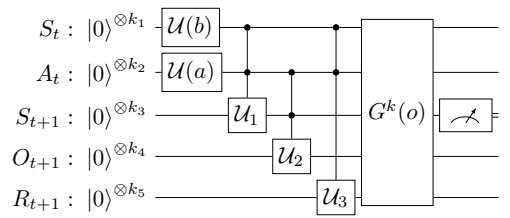
\begin{wrapfigure}{r}{0.4\textwidth}
    \centering
    \resizebox{ \linewidth}{!}{
    \begin{tikzpicture}
      \begin{yquant}
        qubit {$S_t: \; \ket{0}^{\otimes k_1}$} s1;
        qubit {$A_t: \; \ket{0}^{\otimes k_2}$} a1;
        qubit {$S_{t+1}: \; \ket{0}^{\otimes k_3}$} s2;
        qubit {$O_{t+1}: \; \ket{0}^{\otimes k_4}$} o2;
        qubit {$R_{t+1}: \; \ket{0}^{\otimes k_5}$} r2;
        box {$\mathcal{U}(b)$} s1;
        box {$\mathcal{U}(a)$} a1;
        box {$\mathcal{U}_1$} s2 | a1, s1;
        box {$\mathcal{U}_2$} o2 | s2, a1;
        box {$\mathcal{U}_3$} r2 | s1, a1;
        box {$G^k(o)$} (s1, a1, s2, o2, r2);
        measure s2;
      \end{yquant}
    \end{tikzpicture}}
    \caption{Quantum circuit for a belief update.}
    \label{fig:qbn-belief-update}
\end{wrapfigure} 

Lemma \ref{lemma:qbn-measurement-probs} gives the measurement probabilities of this circuit. From this follows theorem \ref{theorem:quantum-belief-update}, which determines that the quantum and classical belief updates coincide.

\begin{lemma}
    The probability of measuring $S_{t+1} = s^{\prime}$ in the quantum circuit of Figure \ref{fig:qbn-belief-update} is given by the following expression
    \begin{equation}
        \expval{\rho (b, a, o)}{s^{\prime}} = \frac{1}{\eta} P(o|s^{\prime}, a) \sum_{s \in \mathcal{S}} P(s^{\prime}|s, a) b(s)
        \label{eq:qbu-prob}
    \end{equation}
    \label{lemma:qbn-measurement-probs}
\end{lemma}
\noindent where $\eta$ is a normalizing constant (defined in appendix \ref{app:q_rej_sampling}).

\begin{theorem}
    The quantum and classical belief update rules are equivalent for DDNs using rejection sampling.
    \label{theorem:quantum-belief-update}
\end{theorem}

\vspace{1em}

 Thus, we can replace the classical belief update with a quantum one. Note that this would also apply if we measured the state encoding qubits without applying $G^k(o)$. However, this amplitude amplification step increases the acceptance probability in rejection sampling, yielding a quantum speed up as demonstrated in Appendix~\ref{app:complexity}. The quantum inference subroutine can be used to calculate the expected reward  of \autoref{eq:b-value} and \autoref{eq:r-value} and both the observation probabilities and belief updates of \autoref{eq:r-value}.
\section{Complexity analysis}
\label{sec:complexity-analysis}

In the section, we provide our main theoretical results: the computational complexity analysis for both the classical and the quantum-classical POMDP look ahead algorithms. In subsection  \ref{subsec:mdp-pomdp} we discuss the possible sources of quantum advantage, after which subsection \ref{subsec:time-comp} presents our key theorems. We derive a bound on the number of samples for the purely classical case, as well as for our quantum-enhanced version. We additionally derive simplified expressions for the case where the quantum algorithm stands out the most: the belief update dominates the cost of the algorithm. 

\subsection{MDP disadvantage and POMDP advantage}
\label{subsec:mdp-pomdp}

The quantum rejection sampling inference algorithm can be used to approximate three probability distributions for POMDPs: the reward distribution $P\left(r_{t+1}\middle|b_t, a_t\right)$; the observation distribution $P\left(o_{t+1}\middle|b_t,a_t\right)$; the belief update $P\left(s_{t+1}\middle|b_t,a_t,o_{t+1}\right)$. It is assumed that they require a different number of samples $n$, $m$ and $l$, respectively, to be approximated.

The evidence of these probability distributions involves the state RV $S_t$, the action RV $A_t$, and the observation RV $O_{t+1}$. Note that from these RVs, only the observation RV is not a root node of the DDN (see \autoref{fig:ddn}). As such, the values of the state and action RVs can be directly encoded in their corresponding CPTs before performing inference, ensuring that the evidence of those RVs always matches and never leads to a rejected sample. 

For the reward and observation distributions, encoding the evidence implies that a sample can be extracted, without ever being rejected, by simply performing random sampling at each RV. This algorithm is called \textit{direct sampling} and is a simpler and more efficient algorithm when compared to rejection sampling. However, \citep{low2014} shows that the quantum analog of direct sampling is \textit{exponentially slower} than the classical counterpart:  the classical and quantum complexities are $\mathcal{O}\left(NM\right)$ and $\mathcal{O}\left(N2^M\right)$ respectively, where $N$ is the number of RVs in the DDN and $M$ is the maximum number of parents of a RV. This makes the classical algorithm the preferred approach for inferring these probability distributions.

In an MDP, observation RVs do not exist; probability distributions depend only on states and actions. This implies that direct sampling is the best inference algorithm, and thus using its classical version is the most efficient approach. On the other hand, belief updating in POMDPs does need to be performed using rejection sampling, and therefore the quantum rejection sampling algorithm poses an advantage. Likewise, there may be a quantum speed-up for the reward calculation if it requires belief updating.

\subsection{Time complexity}
\label{subsec:time-comp}

We now state the main analytical results of this work, regarding the time complexity of the classical and quantum algorithms. Proofs are provided in Appendix \ref{app:complexity}. 

To introduce the notation for the upcoming theorems, there are two important factors that distinguish the complexity for the classical and quantum-classical algorithms. Denoting by $\Phi$ the set of all observation probabilities in the lookahead tree, these factors are $c_l = \sum_{p \in \Phi} p^{-1}$ and $q_l = \sum_{p \in \Phi} p^{-\frac{1}{2}}$. Let $\mathcal{S}, \mathcal{A}, \Omega,$ and $\mathcal{R}$ denote the sizes of the state, action, observation, and reward spaces, respectively. As previously defined, $H$ denotes the look-ahead size, $n$ and $m$ are the numbers of samples for the reward and observation sampling respectively, $N$ is the number of RVs in the BN, and $M$ is the largest number of parents among all RVs. The complexity of the classical algorithm is then given by Theorem \ref{theorem:classical-comp}.

\begin{theorem}
The computational complexity of the classical lookahead algorithm is given by:
\begin{equation*}
    \mathcal{O} \left( n N M \mathcal{A}^{H-1} \Omega^{H-1} \left( \mathcal{A} + \left( \left(\mathcal{R} + \Omega \right)\mathcal{S}^{H-1} \right)^2 c_l \right) \right)
\end{equation*}
\label{theorem:classical-comp}
\end{theorem}

In contrast, the complexity of the quantum-classical algorithm of section \ref{sec:algorithm} is  given by Theorem \ref{theorem:quantum-comp}.

\begin{theorem}
If $2^M \approx M$, the computational complexity of the quantum-classical lookahead algorithm is:
\begin{equation*}
    \mathcal{O} \left( n N M \mathcal{A}^{H-1} \Omega^{H-1} \left( \mathcal{A} + \left( \left(\mathcal{R} + \Omega \right)\mathcal{S}^{H-1} \right)^2 q_l \right) \right)
\end{equation*}
\label{theorem:quantum-comp}
\end{theorem}

Finally, corollary \ref{cor:simplified} provides simplified expressions under an ideal assumption for the quantum-classical lookahead algorithm: the contribution of the belief updating, the only part with a quantum speedup, is much more computationally expensive than the remaining inferences used in the algorithm.

\begin{corollary}
    If $\frac{1}{q_l} \ll \frac{\left( \left(\mathcal{R} + \Omega\right) \mathcal{S}^{H-1} \right)^2}{\mathcal{A}}$, the expressions of Theorems \ref{theorem:classical-comp} and \ref{theorem:quantum-comp} become respectively:
   
    \begin{minipage}{0.48\textwidth}
    \begin{equation*}
    \mathcal{O} \left( n N M \mathcal{A}^{H-1} \Omega^{H-1} \left( \left(\mathcal{R} + \Omega \right) \mathcal{S}^{H-1} \right)^2 c_l \right)
    \label{eq:classical-comp}
    \end{equation*}
    \end{minipage}
    \hfill
    \begin{minipage}{0.48\textwidth}
    \begin{equation*}
    \mathcal{O} \left( n N M \mathcal{A}^{H-1} \Omega^{H-1} \left( \left(\mathcal{R} + \Omega \right) \mathcal{S}^{H-1} \right)^2 q_l \right)
    \label{eq:quantum-comp}
    \end{equation*}
    \end{minipage}
    \label{cor:simplified}
\end{corollary}
\vspace{1em}

In the scenario of Corollary \ref{cor:simplified}, the quantum lookahead would witness a speed up of a factor of $\frac{c_l}{q_l} \in [1,q_l]$ (shown in appendix \ref{appsub:time_comp}). From this ratio, we can derive upper and lower bounds for the quantum advantage: 
\begin{equation}
\label{eq:advantage}
    c_l = q_l^{\beta}, \; \beta \in [1, 2]
\end{equation}

This means that the quantum-classical algorithm has a range of possible speedups: it has the same computational complexity as its classical counterpart in the worst case, and is quadratically faster in the best case. The numerical simulations in section \ref{sec:experimental-results} illustrate this variability, demonstrating how this ratio varies depending on the problem and the impact it has on performance. 

We stress that the result of equation \ref{eq:advantage} corresponds to the most advantageous case for the quantum algorithm, which involves two assumptions: the rejection sampling sub-routine dominates the cost, and the maximum in-degree of the Bayesian network is very small. These conditions are formalized as  $2^M \approx M$ and $\frac{1}{q_l} \ll \frac{\left( \left(\mathcal{R} + \Omega\right) \mathcal{S}^{H-1} \right)^2}{\mathcal{A}}$, respectively.  If they do not hold, the quantum advantage is instead given by: 
\begin{equation}
        \frac{ \mathcal{A} + \left( \left(\mathcal{R} + \Omega \right)\mathcal{S}^{H-1} \right)^2 c_l}{ \mathcal{A} + \left( \left(\mathcal{R} + \Omega \right)\mathcal{S}^{H-1} \right)^2 q_l2^M/M}  
\end{equation}

The second condition is especially restrictive, as the difference between exponential and linear scaling is substantial even for small values of $M$. However, there is potential for an advantage even if this difference is not negligible, as long it is small. More specifically, if we apply the first condition only, we derive the following condition for quantum advantage:
\begin{equation}
    c_l \geq \frac{2^M}{M}q_l
\end{equation}

In other words, the quantum advantage is penalized exponentially in $M$, with a break-even point at $c_l=q_l 2^M/M$. As such, there will be a quantum advantage only where the ratio $c_l/q_l$ is large enough to compensate the fast-growing penalty incurred by the quantum algorithm, which is plausible for problems where $M$ is small. 
\section{Methods}
\label{sec:methods}

We perform numerical simulations of reinforcement learning problems, with the goal of comparing the performances of our quantum-classical algorithm with the fully classical one.  The code is publicly available on GitHub  \cite{QBRLgit}. The objective is a practical assessment of what the $\frac{c_l}{q_l}$ speed-up means in practice. This ratio, and thus the magnitude of the quantum advantage in the rejection sampling sub-routine, depends on the problem at hand; thus, so does the effective benefit it entails.

The aim of our empirical study is not exhaustive benchmarking against the full POMDP literature, but a controlled validation of the mechanism predicted by the complexity analysis: when belief updating relies on rejection sampling, quantum rejection sampling can improve the cost-to-quality trade-off. Thus, we evaluate the algorithms under two complementary protocols: fixed query budget and fixed effective-sample budget. We consider two benchmarks, detailed below, which were chosen to expose different operating regimes: the first test case exhibits greater sensitivity to the number of samples, whereas the second has a higher quantum speed-up for the isolated sampling sub-routine. With this, we relate the empirical behavior directly to the key quantity $c_l/q_l$, and additionally discuss other practical considerations that determine the scenarios where the quantum advantage is most fruitful.


The numerical results we are interested in concern cost-to-benefit ratios. The units of cost are queries, defined as generated (but not accepted) samples\footnote{In the quantum case, the cost of applying the encoding operator once is equated to the cost of generating one sample, as is standard in the literature \cite{brassard2002}.}. The quantum speed-up can be used to improve performance for a given cost; or conversely, to decrease cost (and thus decision making time) for a given performance level.  For the purpose of demonstrating these improvements, we consider the time evolution of two quantities. The first one is the cumulative reward, that is, the rewards obtained by the agents, accumulated over time steps. In this case, the number of queries is the same for both algorithms. This experiment illustrates how the quantum algorithm can improve the decision making for a fixed time.

The second quantity is the cumulative cost, fixing the average performance requirements (represented by the number of accepted samples). Here, cost is again measured in terms of queries. This experiment illustrates how the quantum algorithm can speed up decision-making for a given performance level, highlighting the resource savings it achieves when replicating its classical counterpart exactly.

For a thorough assessment, we will consider in total 4 types of plots:

\begin{enumerate}
    \item \emph{Source of quantum advantage}. Cost per sample as a function of the baseline success probability, as in figure \ref{fig:quadratic_adv_both}. This visually demonstrates the quadratic quantum advantage in the rejection sampling sub-routine.
    \item \emph{Performance comparison under a fixed query budget.} Cumulative reward difference as a function of the time step, as in figures \ref{fig:vary_csample}, \ref{fig:reward_evol} and \ref{fig:reward}. That is, the vertical variable is $\sum_{t=0}^{t_i}r_q - \sum_{t=0}^{t_i}r_c$ This shows the gains of the quantum-classical algorithm as compared to the classical one, for a fixed decision making cost or time. In other words, it isolates the decision quality under equal cost.
    \item \emph{Cost comparison under a fixed effective-sample budget.} Cumulative cost difference as a function of the time step, as in figure \ref{fig:cost}. That is, the vertical variable is $\sum_{t=0}^{t_i}\text{queries}_q - \sum_{t=0}^{t_i}\text{queries}_c$. This shows the savings of the quantum-classical algorithm as compared to the classical one, for a fixed number of obtained samples and thus quality of the estimates. In other words, it isolates the cost under equal estimator quality.
    \item \emph{Cost comparison under a fixed performance level.} Cumulative cost as a function of the cumulative reward, as in figure \ref{fig:cost_vs_reward}. This is similar to the previous point, with the same vertical variable; but the plot additionally factors in how the difference in the number of samples/quality of the estimates affects the agent's performance. In other words, it isolates the cost under equal performance.
\end{enumerate}

We run each algorithm $10^2$ to $10^3$ to obtain averages, and additionally represent the standard deviations among runs as shaded areas where relevant. In the case of the fourth plot category, direct averages cannot be taken due to dispersion in the x-coordinate. To regularize results, we divide the domain into bins, and average the coordinates independently among each bin, to create visually intelligible plots. 

Our two test cases are the tiger problem of \cite{kaelbling1998planning}, and a robot exploration problem, illustrated in figures \ref{fig:tiger} and \ref{fig:robot}, respectively. 

\captionsetup[subfigure]{aboveskip=0pt}
\begin{figure}[!htb]
\centering
\captionsetup{aboveskip=0pt}
\begin{subfigure}[t]{.47\textwidth}
  \centering
  \includegraphics[width=\linewidth]{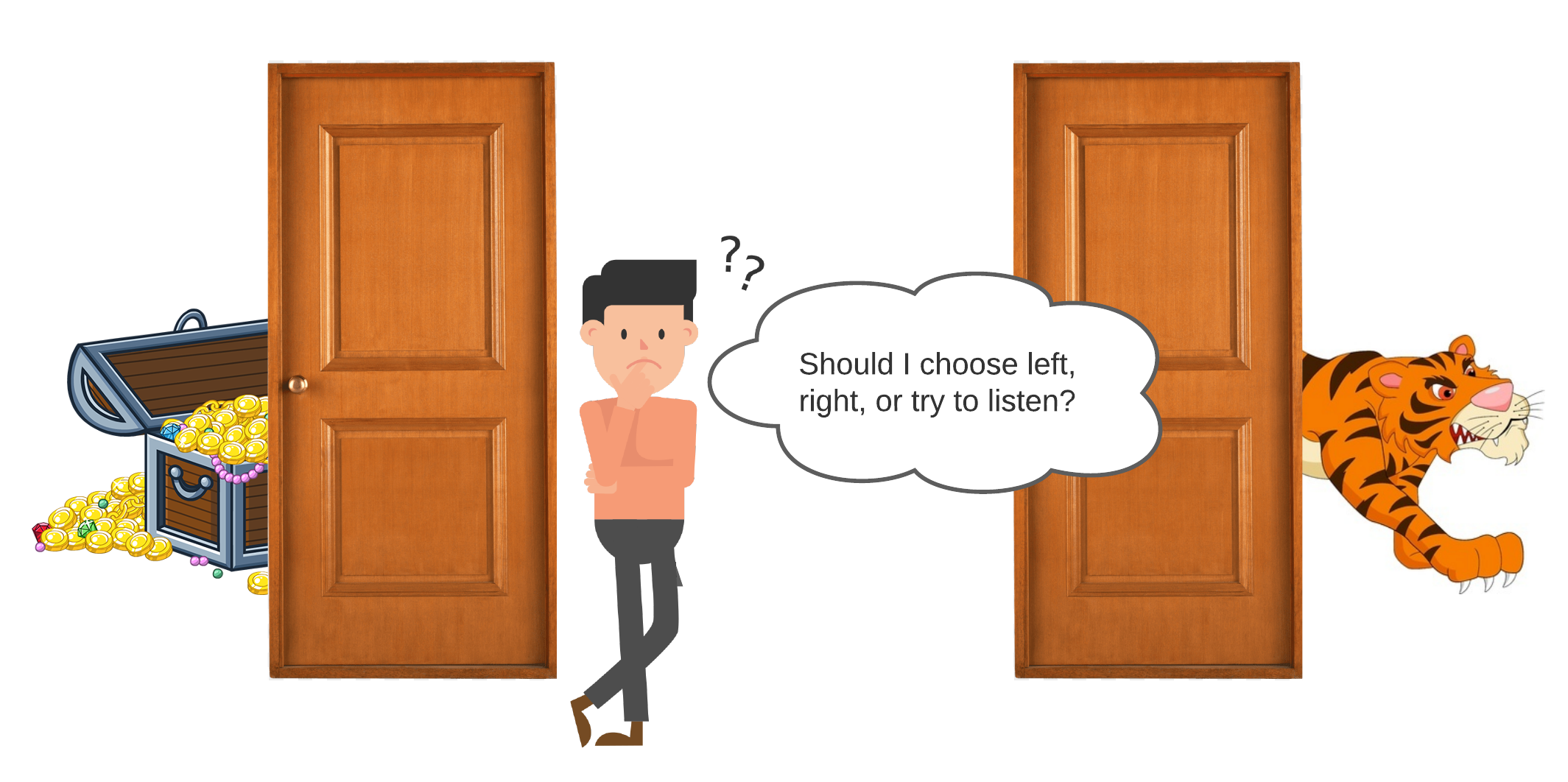}
  \caption{Tiger problem.}
  \label{fig:tiger}
\end{subfigure}\hfill
\begin{subfigure}[t]{.47\textwidth}
  \centering
  \includegraphics[width=0.7\linewidth]{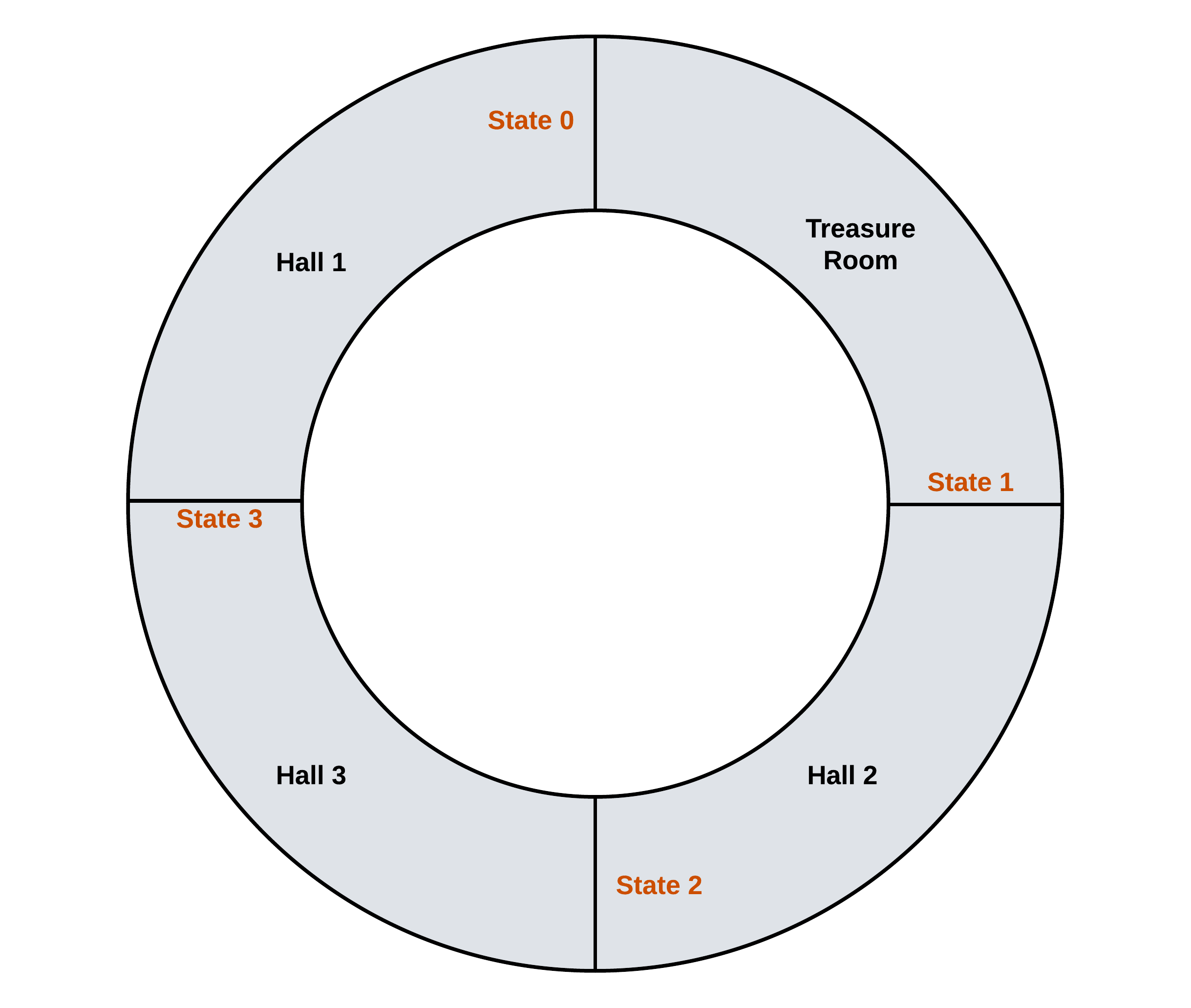}
  \caption{Robot problem.}
  \label{fig:robot}
\end{subfigure}
\caption{Illustration of the set ups for the considered test cases.}
\label{fig:illustrations}
\end{figure}

The tiger problem is as follows. There are two doors; one leads to a room with a tiger, the other to a room with a treasure. The possible actions are to open one door, the other, or try to listen through the doors. Finding the treasure gives a reward of $5$, whereas to find the tiger corresponds to a reward of $-10$. Listening has a reward of $-1$, and has $15\%$ failure probability. After each action except for listening, the tiger and treasure are randomly reassigned to the doors, so that the decision process can continue indefinitely.

For the robot problem, there are 4 contiguous rooms with a circular layout, one of which contains a treasure. The robot can identify whether it is in the treasure room up to $10\%$ error. At each timestep, the robot can move clockwise or counterclockwise, with a reward of $-1$. Alternatively, if it is in the treasure room, it can pull one of two levers. These levers have different chances of producing the treasure, which brings a reward of $10$: $70\%$ and $90\%$. As a caveat, they otherwise damage it to different degrees, yielding rewards of $-5$ and $-20$ for the lever with lower and higher success rates respectively. Pulling a lever in other rooms has no effect. After pulling a lever, the robot is placed back in the initial room.

\section{Experimental Results}
\label{sec:experimental-results}

In this section, we numerically test the performance of the QBRL algorithm applied to partially observable reinforcement learning problems, comparing it to the classical algorithm. 

The fundamental source of quantum advantage is depicted in figure \ref{fig:quadratic_adv_both}, which shows how the cost increases as the problem gets harder. The cost is higher for lower acceptance probabilities, since obtaining one effective sample requires more samples to be produced and rejected on average.  

\captionsetup[subfigure]{aboveskip=0pt}
\begin{figure}[!htb]
\centering
\captionsetup{aboveskip=0pt}
\begin{subfigure}[t]{.47\textwidth}
  \centering
  \includegraphics[width=\linewidth]{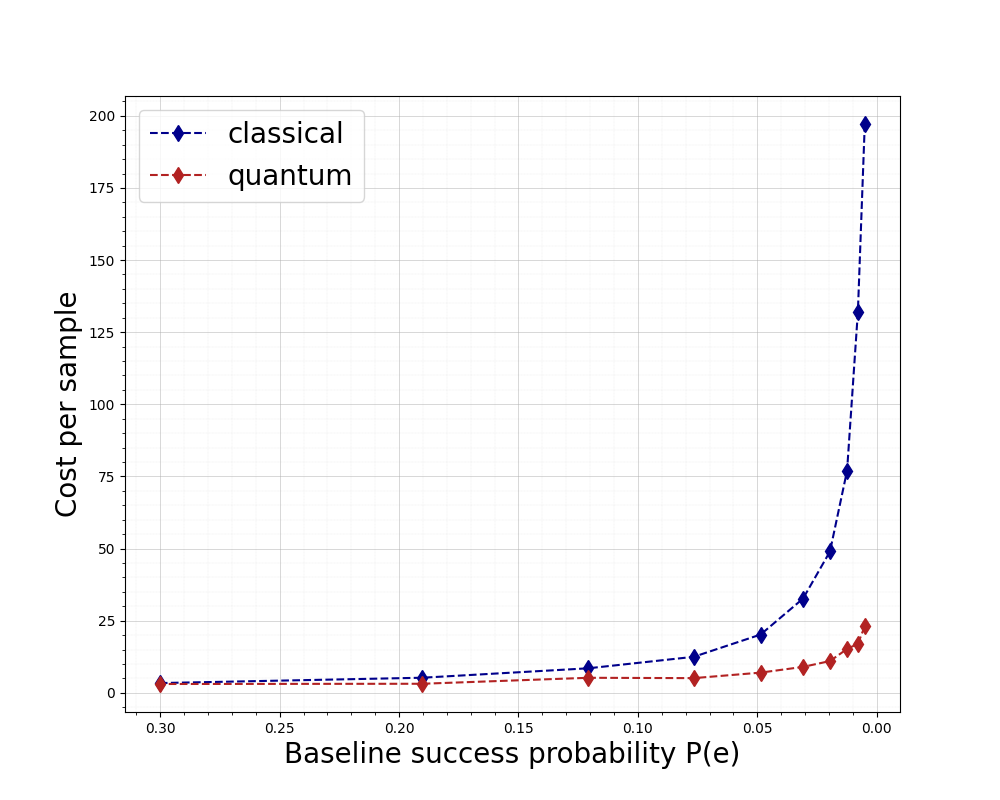}
  \caption{Evolution of the cost per sample with decreasing acceptance probability $P(e)$ for the classical and quantum rejection sampling routines.}
  \label{fig:quadratic_adv}
\end{subfigure}\hfill
\begin{subfigure}[t]{.47\textwidth}
  \centering
  \includegraphics[width=.88\linewidth]{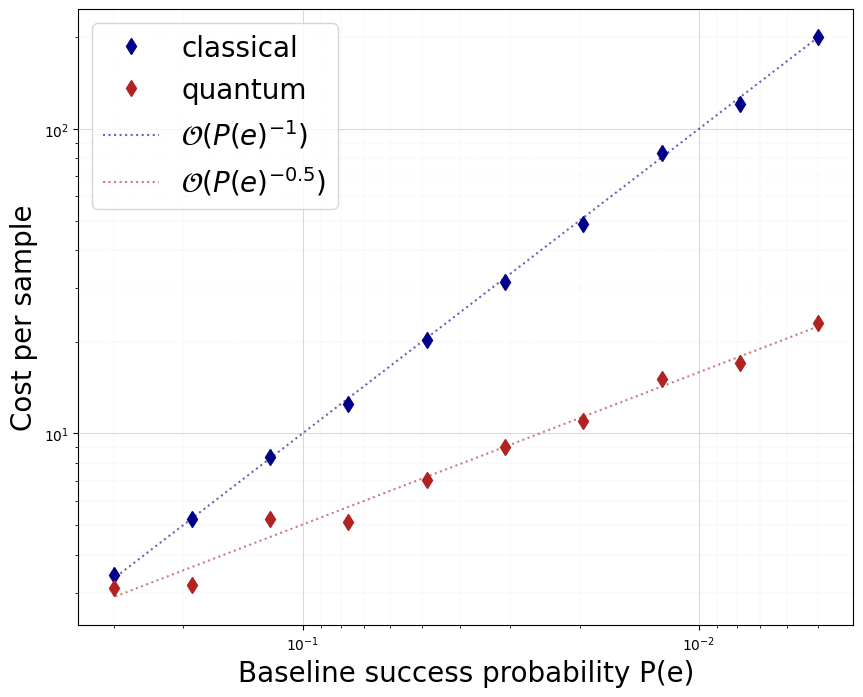}
  \caption{Evolution of the cost per sample with decreasing acceptance probability $P(e)$ for the classical and quantum rejection sampling routines, in a logarithmic scale and juxtaposed with reference lines for the complexity of each.}
  \label{fig:quadratic_adv_log}
\end{subfigure}
\caption{Quadratic quantum advantage in rejection sampling. The results are averaged over $10^3$ runs. The cost per sample corresponds to the average number of queries as defined in the text.}
\label{fig:quadratic_adv_both}
\end{figure}

For the classical case, the average cost per sample grows faster as the success probability decreases, as can be seen in figure \ref{fig:quadratic_adv}. Figure \ref{fig:quadratic_adv_log} clarifies the magnitude of this speed-up: the classical routine has complexity $\mathcal{O}(P(e)^{-1})$, whereas the quantum routine is quadratically faster, with complexity $\mathcal{O}(P(e)^{-0.5})$. This is due to the quantum amplitude amplification process.

In other words, if the acceptance probability becomes $100$ times smaller, the classical cost increases by a factor of $100$, and the quantum one by a factor of $10$. The magnitude of the quantum advantage for a specific instance of rejection sampling depends on the specific values of $P(e)$, which determine the $\frac{c_l}{q_l}$ ratio.  The operator generalization mentioned in section \ref{subsec:qc} would smooth out the irregularities in the plot \ref{fig:quadratic_adv}, which are due to the discrete nature of the amplification. 

The evolution of cumulative reward difference with the time step for the classical and quantum algorithms under a fixed query cost is presented in figure \ref{fig:vary_csample} for various numbers of classical samples ($\{5, 15, 50, 100\}$). The number of effective samples is higher for  the quantum algorithm due to the $\sfrac{c_l}{q_l}$ ratio, which enables a quantum advantage. This advantage depends on the number of classical samples, as well as on the problem. We note the large spread of the values due to the stochastic nature of the problem, which brings large inter-run variability.  

\captionsetup[subfigure]{aboveskip=0pt}
\begin{figure}[!htb]
\centering
\captionsetup{aboveskip=0pt}
\begin{subfigure}[t]{.47\textwidth}
  \centering
  \includegraphics[width=\linewidth]{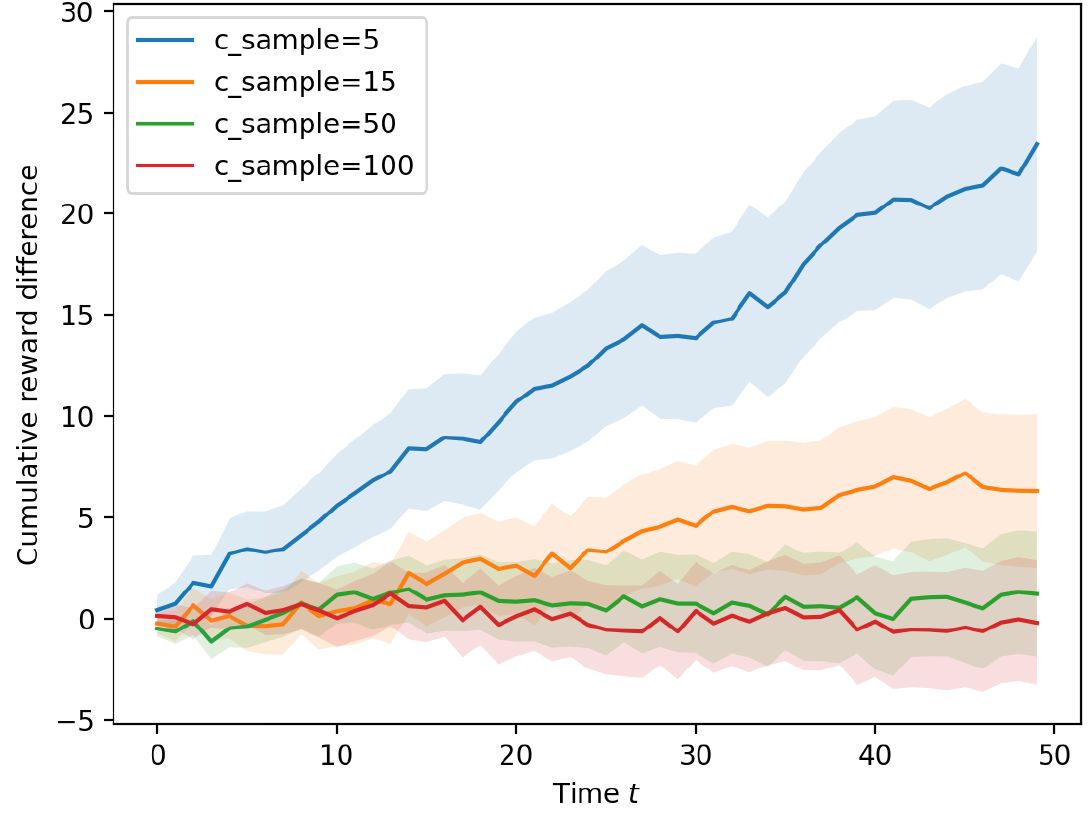}
  \caption{Tiger problem. }
  \label{fig:vary_csample_tiger}
\end{subfigure}\hfill
\begin{subfigure}[t]{.47\textwidth}
  \centering
  \includegraphics[width=\linewidth]{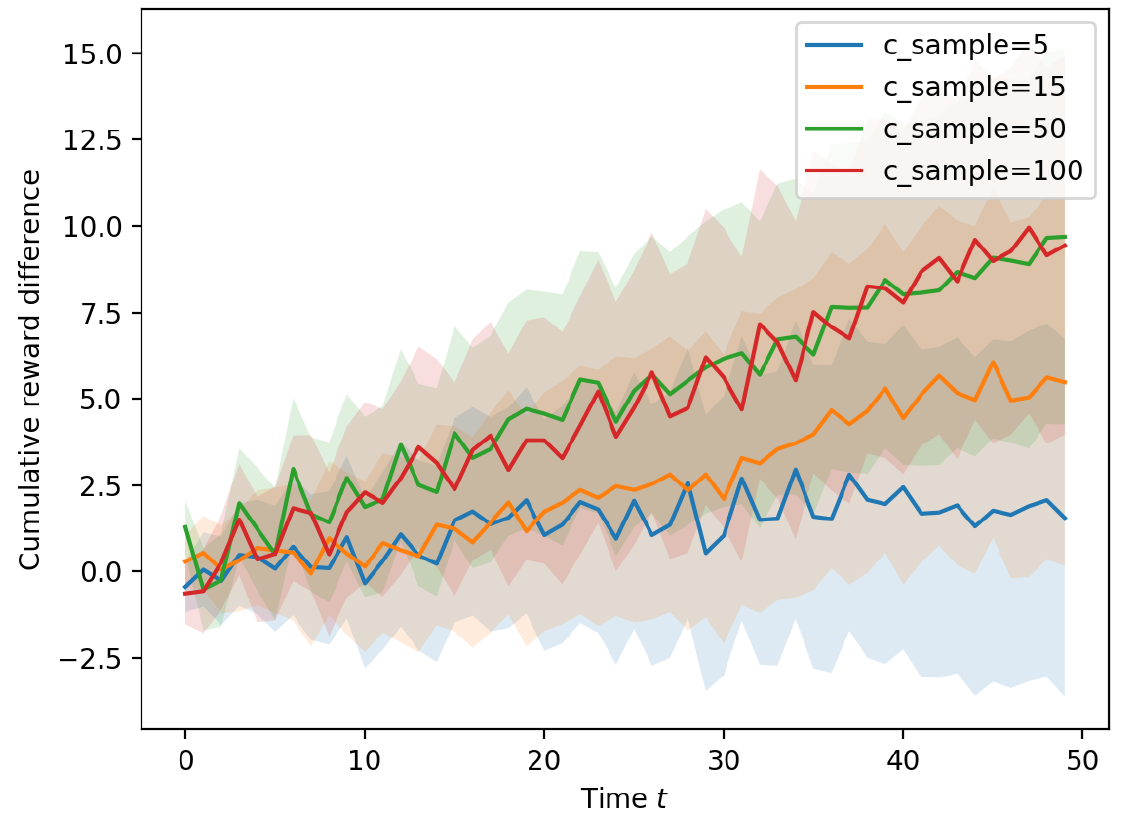}
  \caption{Robot problem. }
  \label{fig:vary_csample_robot}
\end{subfigure}
\caption{Evolution of the cumulative rewards for the quantum and classical algorithms with time, for the tiger (\ref{fig:vary_csample_tiger}) and robot (figure \ref{fig:vary_csample_robot}) problems. The cost is the same for the classical and quantum algorithms, but the latter may obtain higher rewards, depending on the number of samples.}
\label{fig:vary_csample}
\end{figure}
 
 In the tiger problem, the quantum reward advantage peaks for fewer classical samples: the decision-making is especially affected by low sample numbers for this problem, making increases in effective sample numbers particularly critical. This highlights the potential impact of the quantum speed-up in cases where minimal resource usage is pivotal, such as when decisions must be fast or lightweight.  For the robot problem, the advantage also stalls after a certain number of samples, as the benefit of extra samples is saturated; however, this number is higher. 

  We now focus on the sample counts that maximize the reward difference ($5$ and $50$ for the tiger and robot problems respectively), to highlight the most promising scenarios for quantum advantage. The evolution of the classical and quantum cumulative rewards is shown independently in figure  \ref{fig:reward_evol}, and quantitative results and specifications for each test case are given in table \ref{tab:performance}. The quantum algorithm achieved a $94\%$ and $8\%$ improvement in the cumulative reward after $50\%$ time-steps for the tiger and robot problems respectively. The smaller gains for the latter problem are related to the previous observation that the benefit of extra samples hits a plateau.
  
\captionsetup[subfigure]{aboveskip=0pt}
\begin{figure}[!htb]
\centering
\captionsetup{aboveskip=0pt}
\begin{subfigure}[t]{.47\textwidth}
  \centering
  \includegraphics[width=\linewidth]{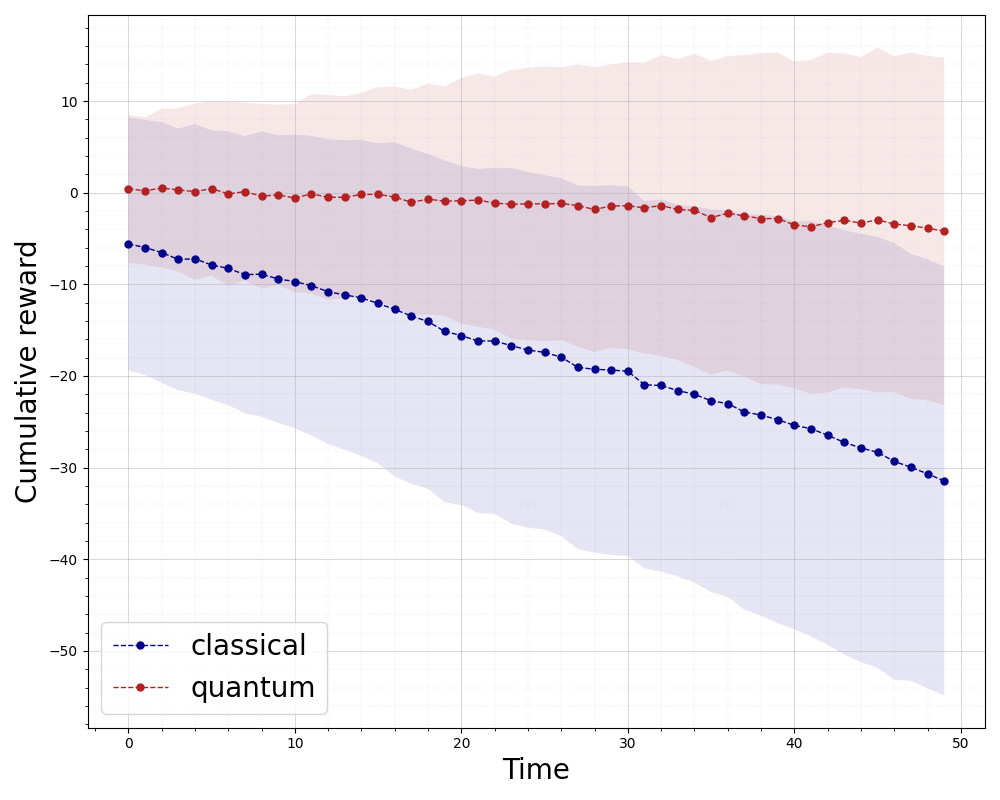}
  \caption{Tiger problem. }
  \label{fig:rewards_tiger}
\end{subfigure}\hfill
\begin{subfigure}[t]{.47\textwidth}
  \centering
  \includegraphics[width=\linewidth]{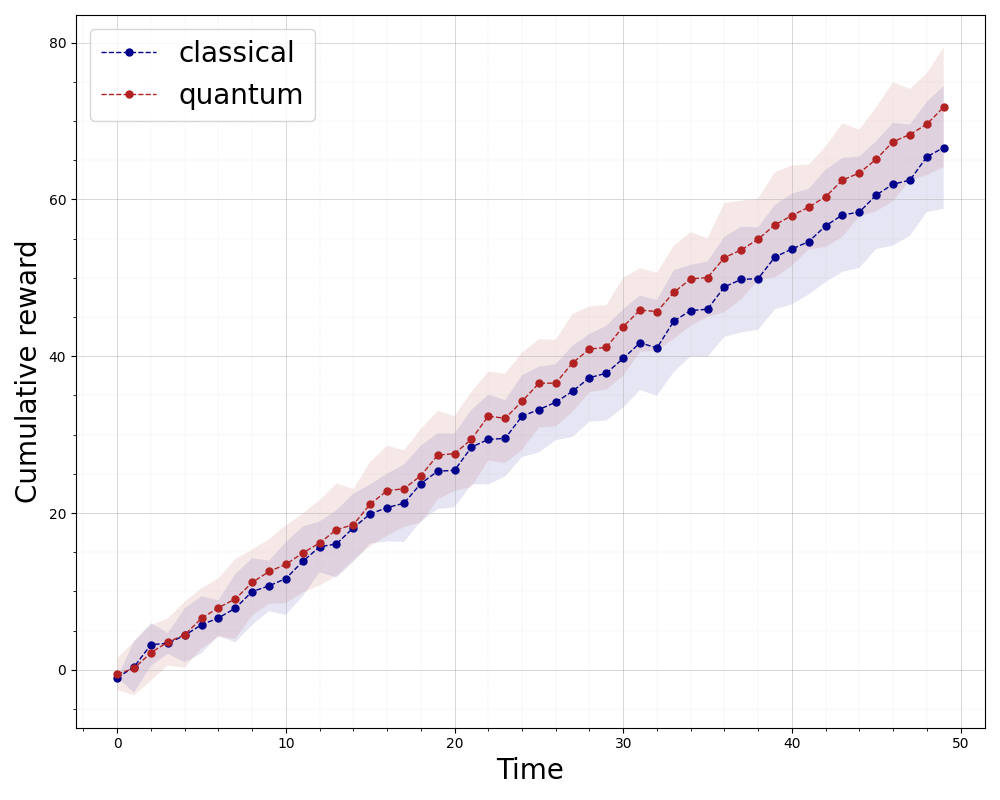}
  \caption{Robot problem.}
  \label{fig:rewards_robot}
\end{subfigure}
\caption{Evolution of the cumulative rewards for the quantum and classical algorithms with time, for the tiger (\ref{fig:rewards_tiger}) and robot (figure \ref{fig:rewards_robot}) problems. The cost is the same for the classical and quantum algorithms, but the quantum ones obtain higher rewards.}
\label{fig:reward_evol}
\end{figure}

 \begin{table}[htbp]
\centering
\begin{tabular}{lccccccccc}
\toprule
\multirow{2}{*}{\textbf{Case}} 
& \multirow{2}{*}{$|\mathcal{S}|$} 
& \multirow{2}{*}{$|\mathcal{A}|$} 
& \multirow{2}{*}{$|\Omega|$} 
& \multirow{2}{*}{$|\mathcal{R}|$} 
& \multirow{2}{*}{$H$} 
& \multirow{2}{*}{\textbf{Samples}} 
& \multirow{2}{*}{$\sfrac{c_l}{q_l}$} 
& \multicolumn{2}{c}{\textbf{Final reward}} \\
\cmidrule(lr){9-10}
& & & & & & & & Classical & Quantum \\
\midrule
\textbf{Tiger} & 2 & 3 & 2 & 3 & 2 & 5 & 1.5 & -31 & -4 \\
\midrule
\textbf{Robot} & 4 & 4 & 2 & 4 & 2 & 50 & 2.7 & 66 & 69 \\ 
\bottomrule
\end{tabular}
\caption{Specifications and final results for each of the test cases. The final reward was calculated after 50 time-steps. The quantum algorithm outperforms the classical one in both cases, with a greater margin in the tiger problem.}
\label{tab:performance}
\end{table}

To better interpret the results, further numerical results for these two problems are presented in figure \ref{fig:numresults}. These plots evidence the differences between the performances of the quantum and classical agents for each of the problems. Figure \ref{fig:reward} shows how the the cumulative reward difference between quantum and classical algorithms evolves over time, illustrating how the quantum-enhanced algorithm improves performance without increasing cost - which is again fixed. It is clear that the magnitude of this improvement is problem-dependent. While the quantum case presents a faster-than-classical increase in the rewards for both problems, the difference is especially marked for the tiger problem, in agreement with figure \ref{fig:reward_evol}. 

\captionsetup[subfigure]{aboveskip=0pt}
\begin{figure}[!htb]
\centering
\captionsetup{aboveskip=0pt}
\begin{subfigure}[t]{.47\textwidth}
  \centering
  \includegraphics[width=\linewidth]{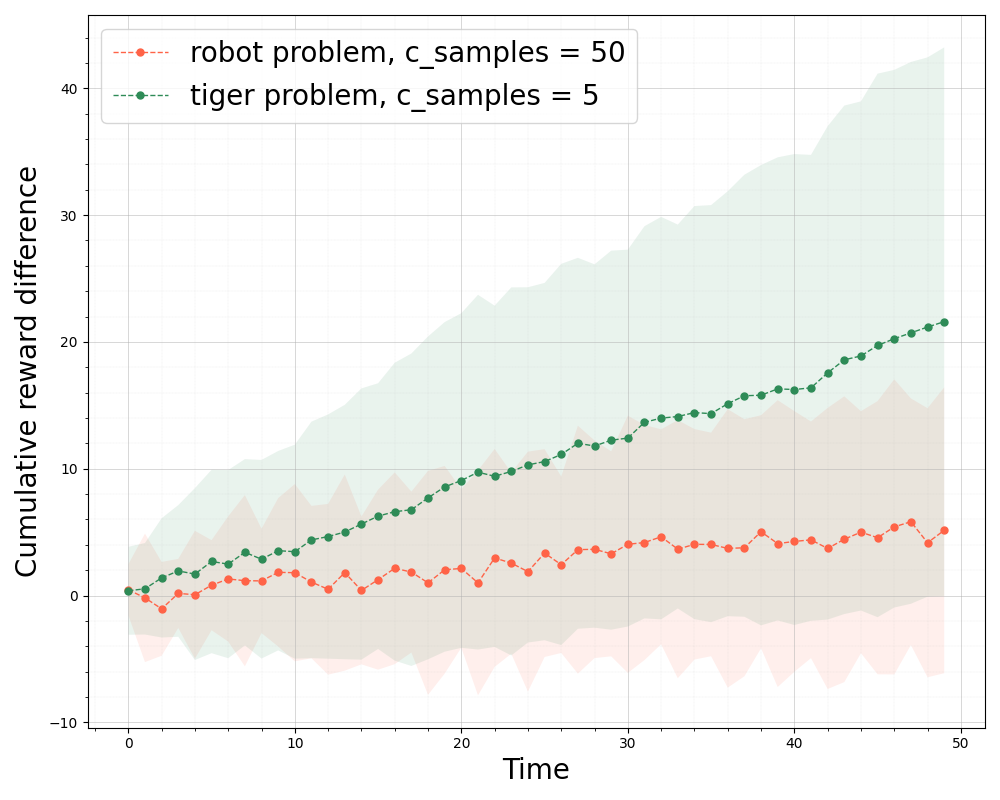}
  \caption{Evolution of the cumulative reward difference between the quantum and classical cases with time. The cost is the same for the classical and quantum algorithms. The quantum approaches become increasingly advantageous with the time. }
  \label{fig:reward}
\end{subfigure}\hfill
\begin{subfigure}[t]{.47\textwidth}
  \centering
  \includegraphics[width=\linewidth]{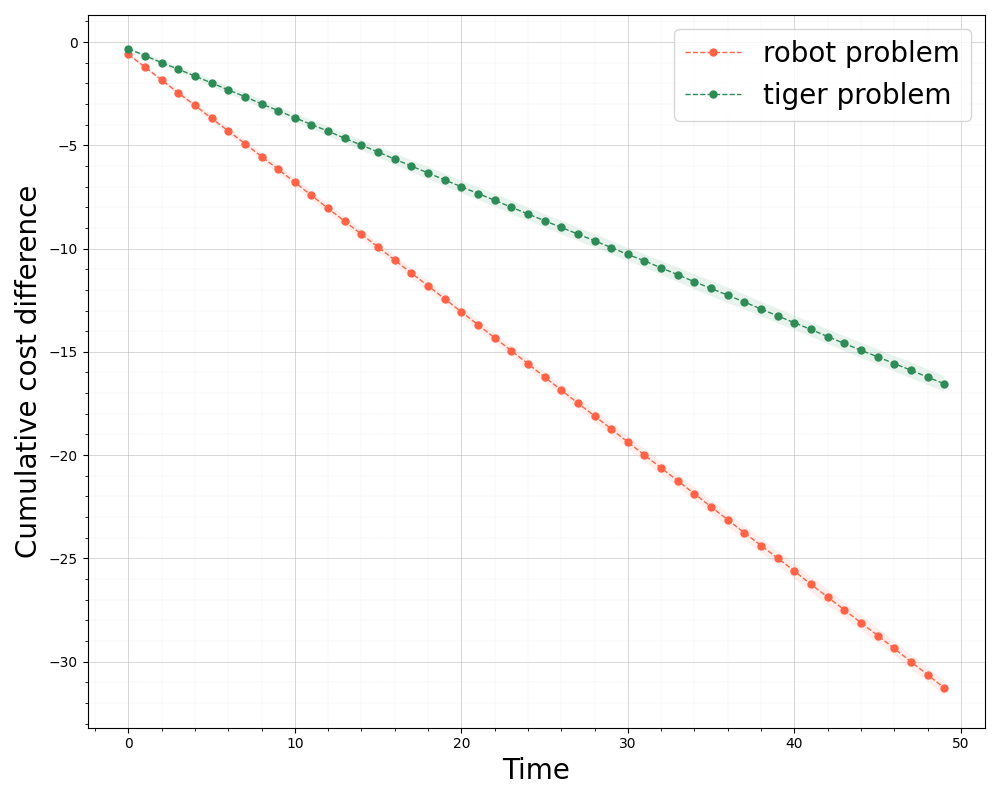}
  \caption{Evolution of the cumulative difference in cost between the quantum and classical cases with time. The expected performance is the same for the classical and quantum algorithms, but the cost is lower for the latter, and the difference increases with the time.}
  \label{fig:cost}
\end{subfigure}
\caption{Numerical assessment of the gains obtained by the quantum algorithm. It improves decision-making (figure \ref{fig:reward}, or equivalently, reduces the computational cost (figure \ref{fig:cost}). Mean results over $100$ runs are presented, with shaded areas marking one standard deviation. The look-ahead horizon is $H=2$ in both cases, and $250$ samples are used to approximate the reward distribution.}
\label{fig:numresults}
\end{figure}

Figure \ref{fig:cost} shows the cost difference over time. In this case, we disregard the number of samples, as it is a multiplicative factor; different values would hinder the visual analysis. This plot showcases how the quantum algorithm spends less resources for the same returns - the number of effective samples, and thus the quality of the decision making, being held constant. This quantity has a smoother evolution, owing to the fact that we are considering expected sample counts rather than effective outcomes. While the latter is related to the former, it is subject to noise, which is injected by the sampling process. In this case, neither of the axes is directly related to performance; as per figure \ref{fig:rewards_tiger}, an increase in time and cost may mean a decrease in rewards. Here, we are merely considering how the cost varies with the number of algorithm steps. 

For this metric, the robot problem shows a greater impact of the quantum routine, unlike before. This means that the $\frac{c_l}{q_l}$ ratio is larger. Clearly, this does not directly manifest as improved returns in practice. The reason is that this ratio quantifies the improvement in sampling efficiency, but not how this sampling affects performance. Some problems are more sensitive to sampling variations, similarly to how some numbers of samples show greater quantum advantage. In some cases, a small increase in the number of samples can bring significant returns, while others may require a larger difference to show a comparable effect. 

Figure \ref{fig:cost_vs_reward} shows how the average cost varies with the achieved rewards. This processing effectively injects noise back into figure \ref{fig:cost}, by considering the effective outcome rather than the expected performance.  

\captionsetup[subfigure]{aboveskip=0pt}
\begin{figure}[!htb]
\centering
\captionsetup{aboveskip=0pt}
\begin{subfigure}[t]{.47\textwidth}
  \centering
  \includegraphics[width=\linewidth]{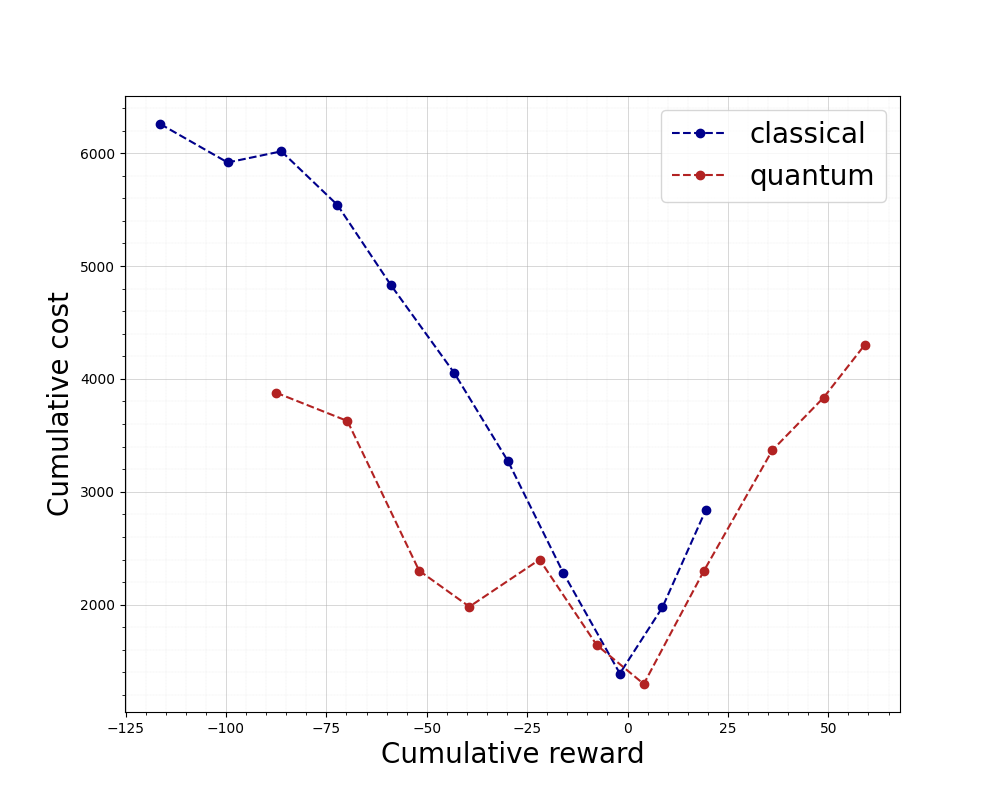}
  \caption{Tiger problem.}
  \label{fig:cost_vs_reward_tiger}
\end{subfigure}\hfill
\begin{subfigure}[t]{.47\textwidth}
  \centering
  \includegraphics[width=\linewidth]{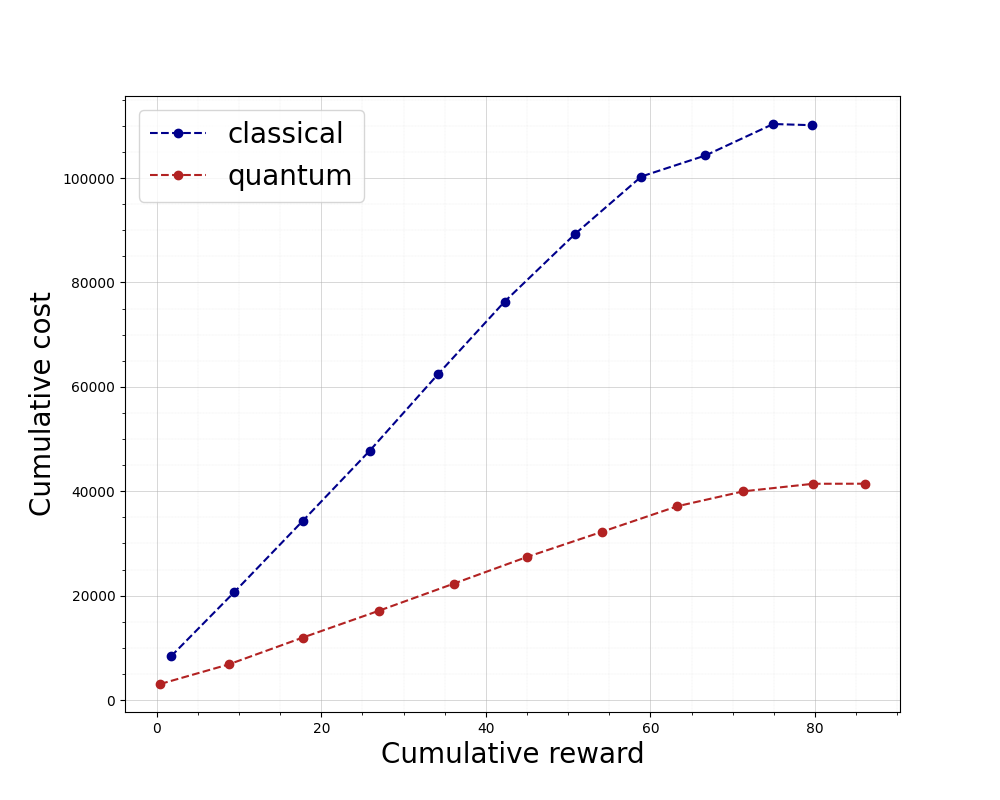}
  \caption{Robot problem.}
  \label{fig:cost_vs_reward_robot}
\end{subfigure}
\caption{Relationship between the achieved reward and the cost (number of  queries) for the classical and quantum algorithms, for the tiger (figure \ref{fig:cost_vs_reward_tiger}) and robot (\ref{fig:cost_vs_reward_robot}) problems. Results were processed by binning and averaging over $100$ runs..}
\label{fig:cost_vs_reward}
\end{figure}

In this case, taking averages is not straightforward, as the rewards take a very large range of possible values,  each of which does not necessarily have a substantial number of data points. As such, we instead bin the data, and average each coordinate independently within bins. This is why the ranges of x-coordinates are different for each algorithm. In the limit of infinitesimal bins, this reduces to taking averages. We empirically found that bins of non-null width still produce sensible results for the cases tested herein. 

Figures \ref{fig:cost_vs_reward_tiger} and \ref{fig:cost_vs_reward_robot} further elucidate the performance differences between the tiger and robot examples. In the tiger problem, there is a larger spread of rewards, namely over negative values. The classical algorithm has runs with more negative cumulative rewards, which are associated with large costs. This is not unexpected, as large absolute values of reward are associated with larger costs regardless of sign: accumulating them through time steps requires more updates.  

Furthermore, the classical algorithm never achieves rewards as good as the best performing quantum runs, even when expending more resources. This is because the agent's poor performance tends to contribute negatively towards the reward as the process advances; positive cumulative rewards are statistical deviations from the expected behavior, and have low costs because those deviations are more likely to occur early.

In the robot problem, figure \ref{fig:cost_vs_reward_robot}, both quantum and classical algorithms display more regular behavior. Higher costs are associated with higher cumulative rewards within each strategy, but the quantum algorithm spends less resources even while performing better. These results, along with those of figure \ref{fig:reward}, suggest that the greater advantage in the tiger problem is due to especially poor performance of the classical algorithm. The classical results are shifted towards more negative rewards, due to the small number of samples. The robot problem shows a smaller discrepancy in final rewards, but a larger discrepancy in cost. This occurs because the number of samples, which is larger for this case, is a multiplicative factor in the cost difference.  

Another salient point for this discussion are cost considerations beyond the concept of a query. In section \ref{sec:complexity-analysis}, we pointed out that the quantum algorithm is exponentially slower in the maximum number of parents of any random variable $M$. This cost is absorbed into each query, since the operator encoding the Bayesian network, which has an exponential dependence on $M$, is necessary to produce the quantum samples. Thus, when zooming in further than queries and looking at elementary operations, this cost may cancel out the quantum advantage for large enough $M$. This consideration harmonizes the analytical and numerical discussions. We note that the complexity analysis considers the worst case scenario; in practice, the full structure of the tree, and not just the largest among all parent counts, will determine this cost. 
\section{Conclusion}
\label{sec:conclusion}

Reinforcement learning underpins numerous achievements in AI, but its application to partially observable environments remains a major challenge. Due to the added complexity brought by the extra probabilistic layer, such applications often struggle with computational bottlenecks. Our work demonstrates how this problem can benefit from quantum resources, showcasing the sources and limitations of the quantum speedup. Specifically, we introduce a hybrid quantum-classical reinforcement learning algorithm that leverages quantum rejection sampling for belief updates within sparse Bayesian networks, obtaining a quantum speed-up for this key sub-routine.

This speed-up can range between null and quadratic, depending on the problem - more specifically, on the structure of the Bayesian network used to model the Markov decision process (largest in-degree among all nodes). The quantum rejection sampling sub-routine can produce samples at lower cost than its classical counterpart, with a complexity reduction that depends on the problem setting. We demonstrate this through a rigorous, oracle-free computational complexity analysis, producing bounds on the classical and hybrid algorithms. We stress that our algorithm is advantageous only when the decision making occurs in a partially observable environment which can be modeled as a sparse Bayesian network. 

To facilitate adoption, we provide all code necessary to reproduce and extend our results, as well as tutorials. In practice, the gains of the quantum approach also depend on the impact of extra samples on performance. Our algorithm is likely to stand out in cases where decisions must be made with limited resources,  the performance is most sensitive to variations in the number of samples used for the belief update, and/or the magnitude of the quantum speed-up is highest. We demonstrate this via numerical simulations for different problems, where the quantum resources improved performance to varying degrees. Our simulations illustrate the practical benefits of the quantum algorithm, which stand out in sample-constrained regimes and for environments where belief accuracy significantly affects the quality of the decisions. 

An interesting direction for future work is the analysis of model-free approaches to reinforcement learning, where the environment description must itself be learned by observation. Such an approach could use learning algorithms to construct a dynamic decision network based on data, in which case the algorithm we presented could be utilized. 

Another possible line of research is the development of a quantum algorithm for choosing an action based on the look ahead tree. This could be done by encoding paths as quantum states, and expected rewards as their amplitudes. This idea is similar to \cite{de2021quantum}, where a quantum unitary operator that weights probability amplitudes according to their utility is developed. However, this operator assumes a static utility, and would here need to be generalized for the case of sequential decision making where the utility is a sequence of rewards. 

Finally, a natural next step is testing our algorithm in real world scenarios, rather than the toy problems we studied. While these serve as illustrative test cases and a starting point for discussion, it would be interesting to assess the practical benefits for realistic applications. Such demonstrations will become both more relevant and more feasible as quantum devices are scaled up.

\bibliography{QBRL.bib}

\appendix 

\renewcommand{\thesection}{\Alph{section}}

\section{Quantum belief update}
\label{app:q_rej_sampling}

Here we detail how to use quantum resources for the quantum belief update, getting the quadratic quantum speed-up for rejection sampling. 

\subsection{Quantum Bayesian networks}

Encoding a BN representing a probability distribution $P \left( X_1, X_2, \dots, X_N \right)$ in a quantum circuit means being able to construct a quantum state $\ket{\Psi} = \mathcal{B} \ket{0}^{\otimes N}$ whose qubits represent the RVs. 

\begin{equation}
    \left| \braket{x_1 x_2 \dots x_N}{\Psi} \right|^2 = P \left( X_1 = x_1, X_2 = x_2, \dots, X_N = x_N \right)
    \label{eq:qbn-encoding}
\end{equation}

Once the circuit is properly encoded, sampling from the BN joint distribution in the quantum setting is equivalent to performing measurements in the qubits of the circuit. For simplicity, we assume binary RVs. For a description of this process for arbitrary discrete RVs, refer to \cite{borujeni2021quantum}. 

Consider a quantum circuit with $N$ qubits $q_i$, each representing a binary RV $X_i$. Each entry of the CPT encodes a probability $P \left( X_i = x_i \middle| \text{parents} \left( X_i \right) = x_p \right)$ of the occurrence of a value $x_i$ in a RV $X_i$, given a value $x_p$ for the RV's parents, $\text{parents}\left( X_i \right)$. In the quantum setting, for each RV $X_i$, a controlled $R_Y(\theta)$ rotation gate is applied for each possible value $x_p$ its parent RVs can take, where the control qubits of this gate are the qubits $q_i$ representing RVs $\text{parents} \left( X_i \right)$. This ensures that any RVs that share an edge in the BN are entangled in the quantum circuit, expressing the information they share.

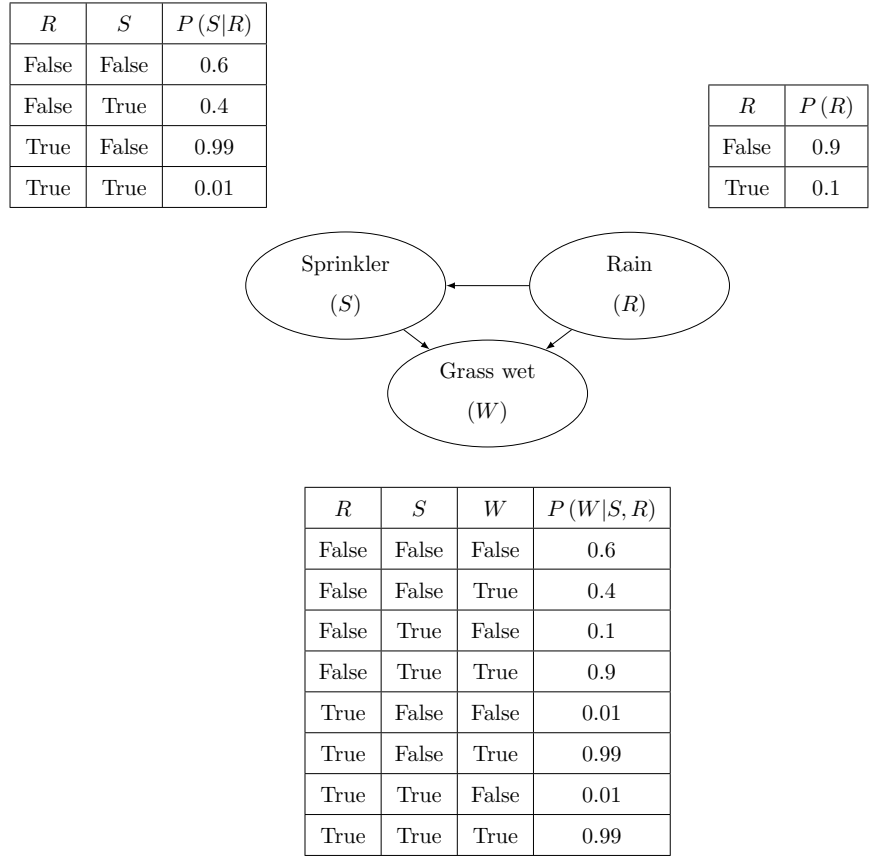
\begin{figure}[ht]
\centering
\resizebox{0.7\linewidth}{!}{\begin{tikzpicture}[
  node distance=0.5cm and 0cm,
  mynode/.style={draw,ellipse,text width=2cm,align=center}
]
\node[mynode] (sp) {Sprinkler \\ ($S$)};
\node[mynode,below right=of sp] (gw) {Grass wet \\ ($W$)};
\node[mynode,above right=of gw] (ra) {Rain \\ ($R$)};
\path (ra) edge[-latex] (sp)
(sp) edge[-latex] (gw) 
(gw) edge[latex-] (ra);
\node[above left=of sp]
{
\begin{tabular}{|c|c|c|}
\hline
$R$ & $S$ & $P \left( S \middle| R \right)$ \\ \hline
False & False & $0.6$ \\ \hline
False & True & $0.4$ \\ \hline
True & False & $0.99$ \\ \hline
True & True & $0.01$ \\ \hline
\end{tabular}
};
\node[above right=of ra]
{
\begin{tabular}{|c|c|}
\hline
$R$ & $P \left( R \right)$ \\ \hline
False & $0.9$ \\ \hline
True & $0.1$ \\ \hline
\end{tabular}
};
\node[below=of gw]
{
\begin{tabular}{|c|c|c|c|}
\hline
$R$ & $S$ & $W$ & $P \left( W \middle| S, R \right)$ \\ \hline
False & False & False & $0.6$ \\ \hline
False & False & True & $0.4$ \\ \hline
False & True & False & $0.1$ \\ \hline
False & True & True & $0.9$ \\ \hline
True & False & False & $0.01$ \\ \hline
True & False & True & $0.99$ \\ \hline
True & True & False & $0.01$ \\ \hline
True & True & True & $0.99$ \\ \hline
\end{tabular}
};
\end{tikzpicture}}
\caption{A BN over three binary random variables}
\label{fig:Bayesian-Network}
\end{figure}

Consider the BN from Figure \ref{fig:Bayesian-Network} as an example. To encode it in a quantum circuit, the following gates have to be applied (see Figure \ref{fig:qbn-example}):
\begin{itemize}
    \item For the variable Rain $R$, with no parents, an uncontrolled rotation gate is applied.
    
    \item For Sprinkler $S$, with $R$ as a parent, two rotation gates are applied, controlled by the qubit representing $R$. The first one represents the case when $\ket{R} = \ket{0}$ \footnote{For $\ket{R} = \ket{0}$, an $X$ gate must be applied \textit{before} and \textit{after} the control operation on qubit $\ket{R}$. This is because a controlled gate, by definition, only changes the target state when its control qubit is in state $\ket{1}$, but in this case, the opposite effect is wanted. Therefore, the first $X$ gate flips state $\ket{R}$ before the control operation, and the second $X$ gate returns the $\ket{R}$ state to the way it was before the first $X$ gate was applied, leaving it unchanged.}, and the second one when $\ket{R} = \ket{1}$.
    
    \item For WetGrass $W$, which has two parents, four rotation gates, controlled by both $\ket{R}$ and $\ket{S}$, are needed. The first one for when $\ket{RS} = \ket{00}$, the second for when $\ket{RS} = \ket{01}$, the third for $\ket{RS} = \ket{10}$ and a final one for $\ket{RS} = \ket{11}$.
\end{itemize}

\usetikzlibrary{quotes}
\begin{figure}[ht]
    \centering
    \resizebox{\linewidth}{!}{
    \begin{tikzpicture}
      \begin{yquant}
        qubit {$R: \; \ket{q_0}$} q0;
        qubit {$S: \; \ket{q_1}$} q1;
        qubit {$W: \; \ket{q_2}$} q2;
        
        box {$R_Y(\theta_1)$} q0;
        
        x q0;
        box {$R_Y(\theta_2)$} q1 | q0;
        x q0;
       
       box {$R_Y(\theta_3)$} q1 | q0;
        
       x q0;
       x q1;
       box {$R_Y(\theta_4)$} q2 | q0, q1;
       x q0;
       
       box {$R_Y(\theta_5)$} q2 | q0, q1;
       x q1;
       
       x q0;
       box {$R_Y(\theta_6)$} q2 | q0, q1;
       x q0;
       
       box {$R_Y(\theta_7)$} q2 | q0, q1;
      \end{yquant}
    \end{tikzpicture}}
    \caption{Quantum circuit encoding the BN of Figure \ref{fig:Bayesian-Network}.}
    \label{fig:qbn-example}
\end{figure}
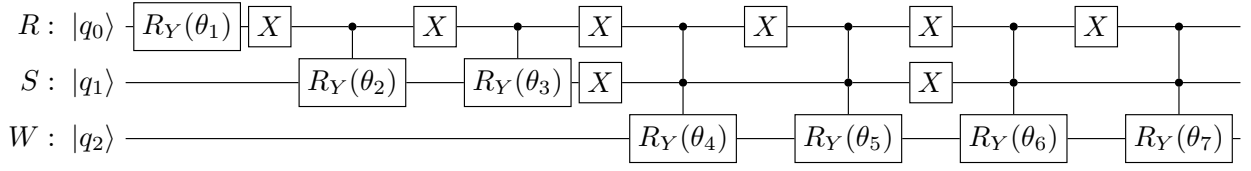 

Finally, the only missing step for the encoding is knowing the angle for each rotation gate. Consider the controlled rotation gate for a binary RV $X_i$ with value $x_p$ for its parents. The angle of rotation $\theta$ is given by:
\begin{equation}
    \theta = 2 \arctan \left( \sqrt{ \frac{ P\left( X_i = 1 \middle| \text{parents}\left(X_i\right) = x_p \right) }{ P\left( X_i = 0 \middle| \text{parents}\left(X_i\right) = x_p \right) } } \right)
    \tag{\ref{eq:qbn-rotation-angle}}
\end{equation}

As \cite{low2014} details, this encoding can be done with a complexity $\mathcal{O} \left( N 2^M \right)$, where $N$ represents the number of RVs in the BN and $M$ represents the greatest number of parents of any RV. Given that this encoding allows the sampling from the joint probability distribution, it can be directly compared to direct sampling, which has a complexity of $\mathcal{O} \left( N M \right)$. The quantum encoding is therefore generally slower  than direct sampling, due to the need to include at most $M$ qubit uniformly controlled rotations, each resulting in a decomposition of $\mathcal{O} \left( 2^M \right)$ elementary gates \cite{2005}.

\subsection{The evidence phase flip operator}
\label{subsec:QC-bn-oracle}

Having explained the state preparation circuit $\mathcal{B}$, the missing part of the algorithm is the oracle operator, which, in this particular example, is the evidence phase flip operator $\mathcal{S}_e$. In quantum Bayesian inference, the good states are states where the evidence qubits $\mathcal{E}$ match the evidence value $e$ of the conditional distribution $P \left( \mathcal{Q} \middle| \mathcal{E} = e \right)$ that is to be inferred.

In this sense, the amplitude amplification technique turns the joint probability distribution of the BN into the aforementioned conditional distribution by decreasing the amplitudes of states with incorrect evidence $\mathcal{E} \neq e$ and amplifying the amplitudes of states with the right evidence $\mathcal{E} = e$. Formally, the resulting state $\ket{\Psi} = \mathcal{B} \ket{0}^{\otimes N}$ of the BN quantum encoding can be expressed as:
\begin{equation}
    \ket{\Psi} = \sqrt{ P (e) } \ket{\mathcal{Q}, e} + \sqrt{1 - P(e)} \ket{\mathcal{Q}, \overline{e}}
    \label{eq:qbn-quantum-state}
\end{equation}

\noindent where $\ket{\mathcal{Q}, e}$ represents the quantum states with the correct evidence $e$, $\ket{\mathcal{Q}, \overline{e}}$ the quantum states with the wrong evidence and $P (e)$ the probability of occurrence of evidence $e$. Therefore, the effect of $\mathcal{S}_e$ on this quantum state is given by:
\begin{equation}
    \mathcal{S}_e \ket{\Psi} = - \sqrt{ P (e) } \ket{\mathcal{Q}, e} + \sqrt{1 - P(e)} \ket{\mathcal{Q}, \overline{e}}
    \label{eq:qbn-evidence-flip}
\end{equation}

Let $e = e_1 e_2 \dots e_k$ represent the evidence bit string, with $e_i \in \mathbb{B}$. Constructing this operator requires the help of two other operators:
\begin{itemize}
    \item The operator $\mathcal{X}_i = \mathbb{I} \otimes \dots \otimes X \otimes \dots \otimes \mathbb{I}$ (where $\mathbb I$ is the identity), which flips the $i$-th evidence qubit by applying an $X$ gate to it and leaving every other qubit unchanged,
    \item The controlled phase operator $P_{1 \dots k}$, which flips the phase of a quantum state if all evidence qubits are in state $\ket{1}$.
\end{itemize}

With the above two operators defined, the evidence phase flip operator $\mathcal{S}_e$ can be defined as follows:
\begin{equation}
    \mathcal{S}_e = \mathcal{F} P_{1 \dots k} \mathcal{F}, \; \mathcal{F} = \prod_{i = 1}^{k} \mathcal{X}_i^{1 - e_i}
    \label{eq:qbn-evidence-flip-operator}
\end{equation}

Consider once more the example of Figure \ref{fig:Bayesian-Network} with evidence qubits $R$ and $W$ and values $1$ and $0$, respectively, to infer the probability distribution $P\left( S \middle| R = 1, W = 0 \right)$. The evidence phase flip operator for this specific case is represented in Figure \ref{fig:qbn-phase-flip-operator}

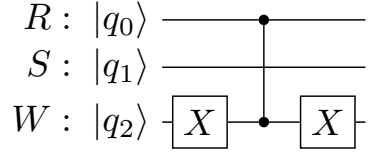
\begin{figure}[ht]
    \centering
    \resizebox{0.3 \linewidth}{!}{
    \begin{tikzpicture}
      \begin{yquant}
        qubit {$R: \; \ket{q_0}$} q0;
        qubit {$S: \; \ket{q_1}$} q1;
        qubit {$W: \; \ket{q_2}$} q2;
        
        x q2;
        zz (q0, q2);
        x q2;
      \end{yquant}
    \end{tikzpicture}}
    \caption{Quantum circuit for the evidence phase flip operator $\mathcal{S}_e$ for the BN of Figure \ref{fig:Bayesian-Network} with evidences $R=1$ and $W=0$.}
    \label{fig:qbn-phase-flip-operator}
\end{figure} 

Notice from Equation \eqref{eq:qbn-evidence-flip-operator} that the quantum operator $\mathcal{F}$ only flips the evidence qubits with evidence $e_i = 0$, otherwise it leaves them unchanged. This allows for the operator $\mathcal{S}_e$ to behave as expressed by Equation \eqref{eq:qbn-evidence-flip}. To show this, let $\ket{\mathcal{Q}, \tilde{e}} = \mathcal{F} \ket{\mathcal{Q}, \overline{e}}$, where $\ket{\tilde{e}} \neq \ket{1}^{\otimes k}$, represent the application of the quantum operator $\mathcal{F}$ to a superposition of quantum states whose evidence does not match $e$. Then: 
\begin{align}
    \begin{split}
        \mathcal{S}_e \ket{\Psi} &= \mathcal{F} P_{1 \dots k} \mathcal{F} \left( \sqrt{ P (e) } \ket{\mathcal{Q}, e} + \sqrt{1 - P(e)} \ket{\mathcal{Q}, \overline{e}} \right) \\
        &= \mathcal{F} P_{1 \dots k} \left( \sqrt{ P (e) } \ket{\mathcal{Q}, 1} + \sqrt{1 - P(e)} \ket{\mathcal{Q}, \tilde{e}} \right) \\
        &= \mathcal{F} \left( - \sqrt{ P (e) } \ket{\mathcal{Q} 1} + \sqrt{1 - P(e)} \ket{\mathcal{Q}, \tilde{e}} \right) \\
        &= - \sqrt{ P (e) } \ket{\mathcal{Q}, e} + \sqrt{1 - P(e)} \ket{\mathcal{Q}, \overline{e}}
    \end{split}
    \label{eq:qbn-evidence-flip-verification}
\end{align}

This operator, together with the BN encoding, leads to an amplitude amplification operator

\begin{equation}
    G = \mathcal{B} \mathcal{S}_0 \mathcal{B}^{\dagger} \mathcal{S}_e
\end{equation}

\noindent making the quantum inference in BNs quadratically faster than in the classical case, provided that the BN has a small maximum number of parents $M$:

\begin{table}[ht]
\centering
{\tabulinesep=1.2mm
\begin{tabu}{l|c|c|}
\cline{2-3}
 & Classical & Quantum \\ \hline
\multicolumn{1}{|l|}{Complexity} & $\mathcal{O} \left( N M P(e)^{-1} \right)$ & $\mathcal{O} \left( N 2^M P(e)^{- \frac{1}{2}} \right)$ \\ \hline
\end{tabu}}
\caption{Comparison between classical and quantum complexity for rejection sampling inference in BNs, taken from \cite{low2014}.}
\label{tab:inference-comparison}
\end{table}

Contrary to many other speedups in quantum algorithms, this particular one does not require any oracle queries. As the authors state, the complexity of this algorithm is completely determined by counting primitive quantum operations. Another important mention is that this algorithm can also be applied to DBNs, as \cite{borujeni2021modeling} shows, by also encoding observation nodes in the quantum circuit and considering a finite number of time-slices of the network at a time, performing inference over those finite slices. The only difference concerning performing inference in a regular BN is their relative size. Moreover, the same concept can be applied to DDNs by also encoding action and reward nodes into qubits, allowing this inference to be performed in the context of solving RL problems.

\subsection{Quantum belief update}

In this subsection, we break down the Bayesian network encoding operator $\mathcal{B}$ introduced in section \ref{subsec:qbn} for the particular case of a dynamic decision network encoding a POMDP.

Classically, a belief update is is performed according to:

\begin{equation}
    \tau_c \left( b, a, o \right) (s^{\prime}) \equiv P \left( s^{\prime} \middle| o, a, b \right) 
     \propto P \left( o \middle| s, a \right) \sum_{s \in \mathcal{S}} P \left( s^{\prime} \middle| a, s \right) b(s)
    \label{eq:qbu-classical}
\end{equation}

A  belief update can be performed in a quantum setting by encoding the DDN in Figure \ref{fig:qbn-belief-update} with the right action $A_t = a$, a superposition on RV $S_t$ according to its belief-state $b$ and performing amplitude amplification on observation $O_{t+1} = o$. If $\tau_q(b, a, o)$ represents the quantum belief update, then the following reasoning seeks to show that the following equation holds:
\begin{equation}
    \tau_q(b, a, o)(s^{\prime}) = \tau_c(b, a, o)(s^{\prime}), \; \forall s^{\prime} \in \mathcal{S}
    \label{eq:belief-update-equality}
\end{equation}

To prove this claim, first consider the quantum circuit corresponding to a belief update of Figure \ref{fig:qbn-belief-update}. The various $\mathcal{U}$ gates perform the encoding of the CPTs of the DDN, while the operator $G^k(o)$ performs the amplitude amplification of the quantum states with the observation $o$. To formally define these operators, a number of definitions are required.

To encode the belief state of the DDN, the operator $\mathcal{U}(b)$ is applied as in Definition \ref{def:quantum-belief-encoding}:
\begin{definition}
	The quantum operator $\mathcal{U}(b)$ encodes belief state $b$ into the qubits of $S_t$:
    \begin{equation*}
          \mathcal{U}(b) \ket{0}^{\otimes k_1} = \sum_{s \in \mathcal{S}} \sqrt{b(s)} \ket{s}
    \end{equation*}
	\label{def:quantum-belief-encoding}
\end{definition}

The quantum operator $\mathcal{U}(a)$ is used to encode action $a$ into the quantum circuit of the DDN, as in Definition \ref{def:quantum-action-encoding}:
\begin{definition}
	The quantum operator $\mathcal{U}(a)$ encodes action $a$ into qubits representing $A_t$, such that:
    \begin{equation*}
    	\mathcal{U}(a) \ket{0}^{\otimes k_2} = \ket{a}
    \end{equation*}
	\label{def:quantum-action-encoding}
\end{definition}

To encode the transition dynamics of the DDN, a controlled operator, $\mathcal{U}_1$ is applied to the quantum circuit, as in Definition \ref{def:quantum-transition-encoding}:
\begin{definition}
	The operator $\mathcal{U}_1$ is a controlled quantum operator that encodes the transition dynamics $P\left( S_{t+1} \middle| s_t, a_t \right)$ on the qubits of RV $S_{t+1}$, based on the values $s_t$ and $a_t$ of RVs $S_t$ and $A_t$, respectively:
	\begin{equation*}
		\mathcal{U}_1 \ket{s a} \ket{0}^{\otimes k_3} = \sum_{s^{\prime} \in \mathcal{S}} \sqrt{ P \left( s^{\prime} \middle| s, a \right) } \ket{s a s^{\prime}}
	\end{equation*}
	\label{def:quantum-transition-encoding}
\end{definition}

For the encoding of the sensor model of the DDN, the quantum operator $\mathcal{U}_2$ is used as in Definition \ref{def:quantum-sensor-encoding}:
\begin{definition}
	$\mathcal{U}_2$ is a controlled quantum operator encoding the sensor model $P \left( O_{t+1} \middle| s_{t+1}, a_t \right)$ on the qubits of RV $O_{t+1}$ given the values $s_{t+1}$ and $a_t$ of RVs $S_{t+1}$ and $A_t$, respectively:
	\begin{equation*}
		\mathcal{U}_2 \ket{a s^{\prime}} \ket{0}^{\otimes k_4} = \sum_{o \in \Omega} \sqrt{ P \left( o \middle| s^{\prime}, a \right) } \ket{a s^{\prime} o}
	\end{equation*}
	\label{def:quantum-sensor-encoding}
\end{definition}

The final encoding operator is $\mathcal{U}_3$, which encodes the reward dynamics of the DDN into the quantum circuit, as expressed by Definition \ref{def:quantum-reward-encoding}:
\begin{definition}
	The quantum operator $\mathcal{U}_3$ encodes the reward dynamics $P \left( R_{t+1} \middle| s_t, a_t \right)$ into the qubits of RV $R_{t+1}$ given the values $s_{t+1}$ and $a_t$ of RVs $S_t$ and $A_t$, respectively:
	\begin{equation*}
		\mathcal{U}_3 \ket{s a} \ket{0}^{\otimes k_5} = \sum_{r \in \mathcal{R}} \sqrt{ P \left( r \middle| s, a \right) } \ket{s a r}
	\end{equation*}
	\label{def:quantum-reward-encoding}
\end{definition}

Applying all these encoding operators, in the order they were presented in the definitions above, is the same as encoding the DDN. They can also be grouped into another quantum operator $\mathcal{B}$, the DDN encoding operator, as in Lemma \ref{lemma:ddn-encoding-operator}. This joint operator can be implemented by the encoding scheme discussed in section \ref{subsec:qbn}; the rotation operators are responsible for encoding the CPTs of a belief update DDN into the quantum circuit. 

\begin{lemma}
	The DDN encoding operator $\mathcal{B}$ is composed of the sequential application of quantum operators $\mathcal{U}(b)$, $\mathcal{U}(b)$, $\mathcal{U}_1$, $\mathcal{U}_2$ and $\mathcal{U}_3$ fom Definitions \ref{def:quantum-belief-encoding} up to \ref{def:quantum-reward-encoding}, in this order. Its application results in the following quantum state:
	\begin{equation*}
		\mathcal{B} \ket{0}^{\otimes N} = \sum_{s \in \mathcal{S}} \sqrt{ b(s) } \sum_{s^{\prime} \in \mathcal{S}} \sqrt{ P(s^{\prime}|s, a) } \sum_{o \in \Omega} \sqrt{ P(o|s^{\prime}, a) } \sum_{r \in \mathcal{R}} \sqrt{ P(r|s^{\prime}, a)} \ket{sas^{\prime} o r}
	\end{equation*}
	\label{lemma:ddn-encoding-operator}
\end{lemma}
\begin{proof}
	First, the operators $\mathcal{U} (b)$ and $\mathcal{U}(a)$ are applied to the initial quantum state $\ket{0}^{\otimes n}$, according to Definitions \ref{def:quantum-belief-encoding} and \ref{def:quantum-action-encoding}:
    \begin{equation*}
        \ket{\psi_1 (b, a)} = \left( \mathcal{U}(b) \otimes \mathcal{U}(a) \otimes \mathbb{I} \right) \ket{0}^{\otimes n} = \sum_{s \in \mathcal{S}} \sqrt{ b(s) } \ket{s a} \ket{0}^{\otimes \left(k_3 + k_4 + k_5 \right) }
    \end{equation*}
    
    Then, the transition dynamics encoding operator $\mathcal{U}_1$ is applied to the quantum state $\ket{\psi_1 (b, a)}$ as in Definition \ref{def:quantum-transition-encoding}: 
    \begin{equation*}
        \ket{\psi_2 (b, a)} = \left( \mathcal{U}_1 \otimes \mathbb{I} \right) \ket{\psi_1 (b, a)} \stackrel{(\ref{def:quantum-transition-encoding})}{=} \sum_{s \in \mathcal{S}} \sqrt{ b(s) } \sum_{s^{\prime} \in \mathcal{S}} \sqrt{ P(s^{\prime}|s,a) } \ket{s a s^{\prime}} \ket{0}^{\otimes \left(k_4 + k_5\right)}
    \end{equation*}
    
    Afterwards, the sensor model encoding operator $\mathcal{U}_2$ is applied to $\ket{\psi_2 (b, a)}$, according to Definition \ref{def:quantum-sensor-encoding}, yielding:
    \begin{align*}
        \begin{split}
        \ket{\psi_3 (b, a)} &= \left( \mathbb{I} \otimes \mathcal{U}_2 \otimes \mathbb{I} \right) \ket{\psi_2 (b, a)} \\ &\stackrel{(\ref{def:quantum-sensor-encoding})}{=} \sum_{s \in \mathcal{S}} \sqrt{ b(s) } \sum_{s^{\prime} \in \mathcal{S}} \sqrt{ P(s^{\prime}|s, a) } \sum_{o \in \Omega} \sqrt{ P(o|s^{\prime}) } \ket{s a s^{\prime} o} \ket{0}^{\otimes k5}
        \end{split}
    \end{align*}
    
    Finally, using Definition \ref{def:quantum-reward-encoding}, the reward dynamics encoding operator $\mathcal{U}_3$ is applied to $\ket{\psi_3 (b, a)}$. The output quantum state is given by:
    \begin{align*}
        \begin{split}
        \ket{\psi_4 (b, a)} &= \left( \mathbb{I} \otimes \mathcal{U}_3 \right) \ket{\psi_3 (b, a)} \\ &\stackrel{(\ref{def:quantum-reward-encoding})}{=} \sum_{s \in \mathcal{S}} \sqrt{ b(s) } \sum_{s^{\prime} \in \mathcal{S}} \sqrt{ P(s^{\prime}|s, a) } \sum_{o \in \Omega} \sqrt{ P(o|s^{\prime}, a) } \sum_{r \in \mathcal{R}} \sqrt{ P(r|s^{\prime}, a)} \ket{sas^{\prime} o r}
        \end{split}
    \end{align*}
\end{proof}

The DDN encoding operator is therefore used to create the amplitude amplification operator $G(o) = \mathcal{B} \mathcal{S}_0 \mathcal{B}^{\dagger} \mathcal{S}_e$. In this section, it is assumed that the number $k$ always gives a solution state with $100\%$ probability, as it still proves that the quantum circuit of Figure \ref{fig:qbn-belief-update} is correct. 

Thus, assuming a perfect parameter $k$ for the number of applications of the amplitude amplification operator, this operator will completely zero all the amplitudes of the quantum states with the wrong observation $O_{t+1} \neq o$, amplifying the amplitudes of quantum states with the observation $O_{t+1} = o$, as follows:
\begin{assertion}
	Suppose the amplitude amplification operator $G(o)$ is applied to a quantum state $\ket{\psi} = \sum_{x \in \mathcal{X}} \sum_{o^{\prime} \in \Omega} \alpha_{x, o^{\prime}} \ket{x, o^{\prime}}$ (a superposition over values of the sets $\mathcal{X}$ and observations $\Omega$) a perfect number $k$ of times. In this scenario, this operator amplifies the amplitudes of quantum states with observations $o$, while leaving all other states with zero amplitude:
	\begin{equation*}
		G^k(o) \ket{\psi} = \frac{1}{\sqrt{\eta}} \sum_{x \in \mathcal{X}} \alpha_{x, o} \ket{x, o}
	\end{equation*}
	\label{asser:grover-qbu}
	
	\noindent where $\eta = \sum_{x \in \mathcal{X}} \left| \alpha_{x, o} \right|^2$ is a normalization constant such that the resulting quantum state $\ket{\psi^{\prime}} = G^k(o) \ket{\psi}$ obeys the relationship $\braket{\psi^{\prime}}{\psi^{\prime}} = 1$.
\end{assertion}
\vspace{1em}

The amplitude amplification operator can be implemented using the phase flip operator discussed in Subsection \ref{subsec:QC-bn-oracle}. Now, with every operator defined, Lemma \ref{lemma:qbu-quantum-circ} derives the final quantum state after applying the encoding circuit depicted in Figure \ref{fig:qbn-belief-update}.
\begin{lemma}
    The final quantum state $\ket{\psi_{\text{final}} (b, a, o)}$ of the quantum circuit depicted in Figure \ref{fig:qbn-belief-update} is given by the following equation:
    \begin{equation*}
        \ket{\psi_{\text{final}} (b, a, o)} = \frac{1}{\sqrt{\eta}} \sum_{s \in \mathcal{S}} \sqrt{ b(s) } \sum_{s^{\prime} \in \mathcal{S}} \sqrt{ P(s^{\prime}|s,a) } \sqrt{ P(o|s^{\prime}, a) } \sum_{r \in \mathcal{R}} \sqrt{P(r|s^{\prime}, a)} \ket{sas^{\prime} o r}
    \end{equation*}
    \label{lemma:qbu-quantum-circ}
    
    \noindent where $\eta = \sum_{s^{\prime} \in \mathcal{S}} P(o|s^{\prime}, a) \sum_{s \in \mathcal{S}} P(s^{\prime}|s,a) b(s)$ is a normalization constant.
\end{lemma}
\begin{proof}
    The final quantum state is the application of the amplitude amplification operator to the output quantum state of the DDN encoding operator of Lemma \ref{lemma:ddn-encoding-operator}:
    \begin{align*}
        \begin{split}
        \ket{\psi_{\text{final}} (b, a, o)} &= G^k(o) \sum_{s \in \mathcal{S}} \sqrt{ b(s) } \sum_{s^{\prime} \in \mathcal{S}} \sqrt{ P(s^{\prime}|s, a) } \sum_{o \in \Omega} \sqrt{ P(o|s^{\prime}, a) } \sum_{r \in \mathcal{R}} \sqrt{ P(r|s^{\prime}, a)} \ket{sas^{\prime} o r} \\ 
        &\stackrel{(\ref{asser:grover-qbu})}{=} \frac{1}{\sqrt{\eta}} \sum_{s \in \mathcal{S}} \sqrt{ b(s) } \sum_{s^{\prime} \in \mathcal{S}} \sqrt{ P(s^{\prime}|s,a) } \sqrt{ P(o|s^{\prime}, a) } \sum_{r \in \mathcal{R}} \sqrt{P(r|s^{\prime}, a)} \ket{sas^{\prime} o r}
        \end{split}
    \end{align*}
\end{proof}

The quantum state $\ket{\psi_{\text{final}}}$ is the result of the application of every quantum operator in the circuit of Figure \ref{fig:qbn-belief-update}. The measurement, the final operation of this circuit, is performed only on the qubit of $S_{t+1}$. The likelihood of an outcome of this measurement can be derived from the partial density matrix $\rho (b, a, o)$ only representing qubit $S_{t+1}$. In turn, this matrix can be calculated by executing a partial trace (over every other RV except $S_{t+1}$) in the density matrix $\ket{\psi_{\text{final}} (b, a, o)} \bra{\psi_{\text{final}} (b, a, o)}$ representing the final quantum state. Using this approach, Lemma \ref{lemma:qbn-measurement-probs_repeat} calculates the measurement probabilities of the quantum circuit of Figure \ref{fig:qbn-belief-update}:

\begin{lemma}
    The probability of measuring $S_{t+1} = s^{\prime}$ in the quantum circuit of Figure \ref{fig:qbn-belief-update} is given by the following expression:
    \begin{equation}
        \expval{\rho (b, a, o)}{s^{\prime}} = \frac{1}{\eta} P(o|s^{\prime}, a) \sum_{s \in \mathcal{S}} P(s^{\prime}|s, a) b(s)
        \tag{\ref{eq:qbu-prob}}
    \end{equation}
    \label{lemma:qbn-measurement-probs_repeat}
\end{lemma}
\begin{proof}
    \begingroup
    \allowdisplaybreaks
    First, let us compute the partial density matrix:
    \begin{align*}
        \rho (b, a, o) &= \sum_{s, a^{\prime}, o^{\prime}, r} \braket{s a^{\prime} o^{\prime} r}{\psi_{\text{final}} (b, a, o)} \braket{\psi_{\text{final}} (b, a, o)}{s a^{\prime} o^{\prime} r} \\
        &= \frac{1}{\eta} \sum_{s, a^{\prime}, o^{\prime}, r} \left( \sum_{s_{\star} \in \mathcal{S}} \sqrt{ b(s_{\star}) } \sum_{s^{\prime} \in \mathcal{S}} \sqrt{ P(s^{\prime}|s_{\star},a) } \sqrt{ P(o|s^{\prime}, a) } \sum_{r^{\prime} \in \mathcal{R}} \sqrt{P(r^{\prime}|s_{\star}, a)} \braket{sa^{\prime}o^{\prime}r}{s_{\star}as^{\prime} o r^{\prime}} \right) 
        \\ 
        &\left( \sum_{s_{\star} \in \mathcal{S}} \sqrt{ b(s_{\star}) } \sum_{s^{\prime} \in \mathcal{S}} \sqrt{ P(s^{\prime}|s_{\star},a) } \sqrt{ P(o|s^{\prime}, a) } \sum_{r^{\prime} \in \mathcal{R}} \sqrt{P(r^{\prime}|s_{\star}, a)} \braket{s_{\star}as^{\prime} o r^{\prime}}{sa^{\prime}o^{\prime}r} \right)
        \\
        &= \frac{1}{\eta} \sum_{s, a^{\prime}, o^{\prime}, r} \left( \sum_{s_{\star} \in \mathcal{S}} \sqrt{ b(s_{\star}) } \sum_{s^{\prime} \in \mathcal{S}} \sqrt{ P(s^{\prime}|s_{\star},a) } \sqrt{ P(o|s^{\prime}, a) } \sum_{r^{\prime} \in \mathcal{R}} \sqrt{P(r^{\prime}|s_{\star}, a)} \delta_{s s_{\star}} \delta_{a a^{\prime}} \delta_{o o^{\prime}} \delta_{r r^{\prime}} \ket{s^{\prime}} \right) 
        \\ 
        &\left( \sum_{s_{\star} \in \mathcal{S}} \sqrt{ b(s_{\star}) } \sum_{s^{\prime} \in \mathcal{S}} \sqrt{ P(s^{\prime}|s_{\star},a) } \sqrt{ P(o|s^{\prime}, a) } \sum_{r^{\prime} \in \mathcal{R}} \sqrt{P(r^{\prime}|s_{\star}, a)} \delta_{s s_{\star}} \delta_{a a^{\prime}} \delta_{o o^{\prime}} \delta_{r r^{\prime}} \bra{s^{\prime}} \right)
        \\
        &= \frac{1}{\eta} \sum_{s \in \mathcal{S}} b(s) \sum_{r \in \mathcal{R}} \left( \sum_{s^{\prime} \in \mathcal{S}} \sqrt{ P(s^{\prime}|s,a) } \sqrt{ P(o|s^{\prime}, a) } \sqrt{ P(r|s, a) } \ket{s^{\prime}} \right) 
        \\
        &\left( \sum_{s^{\prime} \in \mathcal{S}} \sqrt{ P(s^{\prime}|s,a) } \sqrt{ P(o|s^{\prime}, a) } \sqrt{ P(r|s, a) } \bra{s^{\prime}} \right)
    \end{align*}
    \endgroup
    
    Then, the probability of measuring state $S_{t+1}$ with value $s^{\prime}$ is computed as follows:
    \begin{align*}
        \begin{split}
            \expval{\rho (b, a, o)}{s^{\prime}} &= \frac{1}{\eta} \sum_{s \in \mathcal{S}} b(s) \sum_{r \in \mathcal{R}} \left( \sum_{s^{\star} \in \mathcal{S}} \sqrt{ P(s^{\star}|s,a) } \sqrt{ P(o|s^{\star}, a) } \sqrt{ P(r|s, a) } \braket{s^{\prime}}{s^{\star}} \right) 
            \\
            &\left( \sum_{s^{\star} \in \mathcal{S}} \sqrt{ P(s^{\star}|s,a) } \sqrt{ P(o|s^{\star}, a) } \sqrt{ P(r|s, a) } \braket{s^{\star}}{s^{\prime}} \right) 
            \\
            &= \frac{1}{\eta} \sum_{s \in \mathcal{S}} b(s) \sum_{r \in \mathcal{R}} \left( \sum_{s^{\star} \in \mathcal{S}} \sqrt{ P(s^{\star}|s,a) } \sqrt{ P(o|s^{\star}, a) } \sqrt{ P(r|s, a) } \delta_{s^{\prime} s^{\star}} \right) 
            \\
            &\left( \sum_{s^{\star} \in \mathcal{S}} \sqrt{ P(s^{\star}|s,a) } \sqrt{ P(o|s^{\star}, a) } \sqrt{ P(r|s, a) } \delta_{s^{\star} s^{\prime}} \right) 
            \\
            &= \frac{1}{\eta} \sum_{s \in \mathcal{S}} b(s) \sum_{r \in \mathcal{R}} P(s^{\prime}|s,a) P(o|s^{\prime}, a) P(r|s, a) \\
            &= \frac{1}{\eta} P(o|s^{\prime}, a) \sum_{s \in \mathcal{S}} P(s^{\prime}|s, a) b(s) \sum_{r \in \mathcal{R}} P(r|s, a)\\
            &= \frac{1}{\eta} P(o|s^{\prime}, a) \sum_{s \in \mathcal{S}} P(s^{\prime}|s, a) b(s)
        \end{split}
    \end{align*}
\end{proof}

Once the probabilities of measurement of the final quantum state of the belief update quantum circuit are known, let us rename $\tau_q (b, a, o) (s^{\prime}) = \expval{\rho (b, a, o)}{s^{\prime}}$ as the quantum belief update rule. Thus, Theorem \ref{theorem:quantum-belief-update_repeat} proves the initial claim, that the quantum and classical belief update rules, as in Equation \eqref{eq:belief-update-equality}, are equivalent:

\begin{theorem}
    The quantum and classical belief update rules $\tau_q (b, a, o)$ and $\tau_c (b, a, o)$, for DDNs using rejection sampling, are equivalent:
    \begin{equation*}
        \tau_q (b, a, o)(s^{\prime}) = \tau_c (b, a, o)(s^{\prime}), \; \forall s^{\prime} \in \mathcal{S}
    \end{equation*}
    \label{theorem:quantum-belief-update_repeat}
\end{theorem}
\begin{proof}
    Equation \eqref{eq:qbu-prob} can be re-written as:
    \begin{equation*}
        \tau_q (b, a, o) (s^{\prime}) = \frac{1}{\eta} P(o|s^{\prime}, a) \sum_{s \in \mathcal{S}} P(s^{\prime}|s, a) b(s)
    \end{equation*}
    
    \noindent It is also known from the classical case that the classical belief update can be re-written to incorporate a proportionality constant:
    \begin{equation*}
        \tau_c (b, a, o) (s^{\prime}) \stackrel{(\eqref{eq:qbu-classical})}{=} \frac{1}{\eta^{\prime}} P(o|s^{\prime}, a) \sum_{s \in \mathcal{S}} P(s^{\prime}|s, a) b(s)
        \label{eq:qbu-classical-prop}
    \end{equation*}
    
    Since both $\tau_q(b, a, o)(s^{\prime})$ and $\tau_c(b, a, o)(s^{\prime})$ are probability distributions, their sum over $s^{\prime}$ is the same, which is equal to one, and therefore:
    \begin{align*}
        \begin{split}
        \sum_{s^{\prime}} \tau_q (b, a, o) (s^{\prime}) &= \sum_{s^{\prime}} \tau_c (b, a, o) (s^{\prime})
        \\ 
        \Leftrightarrow \frac{1}{\eta} \sum_{s^{\prime} \in \mathcal{S}} P(o | s^{\prime}, a) \sum_{s \in \mathcal{S}} P(s^{\prime}|s,a) b(s) &= \frac{1}{\eta^{\prime}} \sum_{s^{\prime} \in \mathcal{S}} P(o | s^{\prime}, a) \sum_{s \in \mathcal{S}} P(s^{\prime}|s,a) b(s) \\
        \Leftrightarrow \eta &= \eta^{\prime}
        \end{split}
    \end{align*}
    
    Therefore, the quantum and the classical belief update are equivalent:
    \begin{equation*}
        \tau_q(b, a, o)(s^{\prime}) = \tau_c(b, a, o)(s^{\prime}), \; \forall s^{\prime} \in \mathcal{S} 
    \end{equation*}
\end{proof}

As a consequence of Theorem \ref{theorem:quantum-belief-update}, the belief updating procedure on any algorithm resorting to it can be performed in a quantum setting, leveraging the quadratic speedup of quantum rejection sampling. 

Naturally, the way to perform this quantum belief update operation is not to just run the circuit once and retrieve a measurement. What is shown in Theorem \ref{theorem:quantum-belief-update} is that the probabilities of obtaining a certain measurement match the probabilities of the classical belief update rule. Nonetheless, obtaining an approximation of the next belief state using this circuit should be done with a certain number of samples $n$.

The simplest way to proceed is to run the quantum circuit of Figure \ref{fig:qbn-belief-update} $n$ times ($n$ being the desired number of samples), extracting the samples $s^{\prime}_i \sim \tau_q(b, a, o), \; i = \{1,\dots,n\}$ for the next state, and approximating the next belief state $b^{\prime}(s^{\prime})$ by the following equation:
\begin{equation}
    b^{\prime}_n(s^{\prime}) = \frac{1}{n} \sum_{i = 1}^{n} \delta_{s^{\prime} s^{\prime}_i}
    \label{eq:belief-update-approx}
\end{equation}

\noindent where $b^{\prime}_n$ is an approximation of the belief state, computed resorting to $n$ samples. Equation \eqref{eq:belief-update-approx} simply states that the next belief state value for state $s^{\prime}$ is the fraction of extracted samples that match state $s^{\prime}$.


\section{Complexity Analysis}
\label{app:complexity}

In this Appendix we perform an execution time complexity analysis of both the classical and quantum-classical lookahead algorithms in order to compare the two and check for potential speedups in the quantum approach.

Both of these algorithms are probabilistic. To execute them, samples have to be extracted and used to approximate probability distributions, which are then used to calculate a near-optimal action. As such, to determine their computational complexity, we start by determining their \textit{sample complexity}: the total number of samples needed to execute the algorithm. 

This appendix is therefore divided into two sections. \ref{appsub:sample_comp} discusses the sample complexity of the lookahead algorithm, giving conditions under which there are formal guarantees for the algorithm to run with an error of at most $\epsilon$, and a confidence interval of at least $1 - \delta$. This analysis applies to both the classical and quantum algorithms, as the number of samples that need to be extracted is the same in both cases.

These results are then used in \ref{appsub:time_comp} to determine the computational complexity.

\ref{appsub:aux_lemmas} lays out four lemmas used in these analyses, which we present separately for ease of reading. 

\subsection{Sample Complexity}
\label{appsub:sample_comp}


The sample complexity of the quantum-classical algorithm counts the total number of samples \footnote{which corresponds to the total amount of times that the dynamic decision network (DDN) has to be called} required to achieve a certain goal. In this analysis, the chosen goal is to reach an $\epsilon$ approximation of the action value function of the optimal action, with a confidence interval of $1 - \delta$. 

To exemplify, consider that the value $r$ is approximated by the output $r^{\mathcal{L}}$ of a probabilistic algorithm $\mathcal{L}$. For a given precision $\epsilon$, the output of a single execution of $\mathcal{L}$ is said to be $\epsilon$-approximate if:
\begin{equation}
	\left| r - r^{\mathcal{L}} \right| \le \epsilon
	\label{eq:epsilon-approximation}
\end{equation}

However, if algorithm $\mathcal{L}$ is probabilistic, the error of the approximation varies from one execution to another. As a consequence, if this algorithm is executed an extremely high number of times, the approximation error of some of these executions is likely to be higher than the defined precision. Therefore, most probabilistic algorithms can not be called $\epsilon$-approximate, as not every execution obeys Equation \eqref{eq:epsilon-approximation}.

A better alternative is not to force the approximation of every execution to have a precision of $\epsilon$, but to guarantee instead that each execution has at least some probability $1 - \delta$ (confidence interval) of having an error no greater than $\epsilon$:

\begin{equation}
	P \left( \left| r - r^{\mathcal{L}} \right| \le \epsilon \right) \ge 1 - \delta
	\label{eq:pac-learnability}
\end{equation}

Consider the definition of the optimal action value function $Q^{\star}(\mathbf{b}_t, a_t)$\footnote{Vector quantities are represented in bold.}:

\begin{gather*}
    Q^{\star}(\mathbf{b}_t, a_t) = \mathbb{E} \left( r_{t+1} \middle| \mathbf{b}_t, a_t \right) + \gamma \sum_{o_{t+1}} P(o_{t+1}|\mathbf{b}_t,a_t) \max_{a_{t+1} \in \mathcal{A}} Q^{\star}(\mathbf{b}_{t+1},a_{t+1})\\
    \mathbf{b}_{t+1}(s_{t+1}) = P(s_{t+1}|o_{t+1},a_t,\mathbf{b}_t)
\end{gather*}

The analysis of the sample complexity of the algorithm entails many manipulations of these expressions. So, for the sake of compactness and comprehensibility of the proofs of the results ahead, it is useful to resort to vectorial notation \footnote{Although vectorial notation in preferred in this section, explicitly writing out the probability distributions is sometimes clearer and occasionally used.}. As such,  let $\mathbf{p}_{r_t} \in \Delta_{\mathcal{R}}$ \footnote{In this notation, $\mathbf{p} \in \Delta_{\mathcal{V}}$ means that $\mathbf{p}$ is a probability distribution vector over the space $\mathcal{V}$.} represent the probability distribution $P(r_{t+1}|\mathbf{b}_t, a_t), \forall r_{t+1} \in \mathcal{R}$, $\mathbf{p}_{o_t} \in \Delta_{\Omega}$ represent the probability distribution $P(o_{t+1}|\mathbf{b}_t,a_t), \forall o_{t+1} \in \Omega$, $\mathbf{r} \in \mathbb{R}^{\left|\mathcal{R}\right|}$ be a vector whose entries enumerate every possible reward and $\mathbf{V}^{\star}_{t+1} \in \mathbb{R}^{\left|\Omega\right|}$ be a vector whose entries are the values $\max_{a_{t+1} \in \mathcal{A}} Q^{\star}(\mathbf{b}_{t+1}, a_{t+1})$ for every belief update of $\mathbf{b}_t$ after taking action $a_t$ \footnote{Since the belief update rule is $\mathbf{b}_{t+1}(s_{t+1}) = P(s_{t+1}|o_{t+1},a_t,\mathbf{b}_t)$ and both $a_t$ and $\mathbf{b}_t$ are fixed, there are $\left|\Omega\right|$ different possible $\mathbf{b}_{t+1}$, one for each observation $o_{t+1} \in \Omega$.}. The previous equation can be re-written, using this notation, as\footnote{In vector notation, $\mathbf{v}^{\top}$ is the transpose of $\mathbf{v}$, such that the inner-product of two vectors $\mathbf{u}$ and $\mathbf{v}$ can be written as $\mathbf{u}^{\top} \mathbf{v}$.}:

\begin{equation}
    Q^{\star}(\mathbf{b}_t, a_t) = \mathbf{p}_{r_t}^{\top} \mathbf{r} + \gamma \mathbf{p}_{o_t}^{\top} \mathbf{V}_{t+1}^{\star}
    \label{eq:qfunction-vector}
\end{equation}

The lookahead algorithm $\mathcal{L}$ calculates a near-optimal value function $Q^{\mathcal{L}}$ similarly to that of Equation \eqref{eq:qfunction-vector}, but replacing $\mathbf{p}_{r_t}$ with its empirical estimation $\mathbf{p}_{r_t,n_t}$ using $n_t$ samples, $\mathbf{p}_{o_t}$ with its empirical estimation $\mathbf{p}_{o_t,m_t}$ using $m_t$ samples and $\mathbf{V}_{t+1}^{\star}$ with its empirical estimation $\mathbf{V}^{\mathcal{L}}_{t+1}$ calculated recursively. Thus, the lookahead algorithm, starting with root belief node $\mathbf{b}_t$ at time-step $t$ obeys the following recursive Equation \eqref{eq:qfunc-approx-vector-rec} with a terminal case described by Equation \eqref{eq:qfunc-approx-vector-stop} (according to the lookahead horizon $H$):
\begin{equation}
    Q^{\mathcal{L}}(\mathbf{b}_t, a_t) = \mathbf{p}_{r_t,n_t}^{\top} \mathbf{r} + \gamma \mathbf{p}_{o_t,m_t}^{\top} \mathbf{V}^{\mathcal{L}}_{t+1}
    \label{eq:qfunc-approx-vector-rec}
\end{equation}
\begin{equation}
    Q^{\mathcal{L}}(\mathbf{b}_{t+H}, a_{t+H}) = 0
    \label{eq:qfunc-approx-vector-stop}
\end{equation}

So far, three sources of error for the lookahead algorithm have been mentioned: 
\begin{itemize}
    \item Approximating the reward distribution $\mathbf{p}_{r_t}$ by its empirical estimation $\mathbf{p}_{r_t, n_t}$, using $n_t$ samples.
    \item Approximating the observation distribution $\mathbf{p}_{o_t}$ by its empirical estimation $\mathbf{p}_{o_t, m_t}$, using $m_t$ samples.
    \item The stop case for the lookahead algorithms in the horizon $H$ limit.
\end{itemize}

There is yet another source of error, coming from the estimation of the belief state. This estimation, however, differs slightly from the estimations of the two distributions mentioned above. Suppose we have an empirical estimation $\mathbf{b}_{l_t}$ of the belief state $\mathbf{b}_t$, using $l_t$ samples. When performing the belief update to get the next belief state $\mathbf{b}_{l_{t+1}}(s_{t+1}) = P_{l_{t+1}}(s_{t+1}|o_{t+1},a_t,\mathbf{b}_{l_t})$, this quantity estimates $P(s_{t+1}|o_{t+1},a_t,\mathbf{b}_{l_t})$ rather than the belief-update $P(s_{t+1}|o_{t+1},a_t,\mathbf{b}_t)$ of the true belief-state $\mathbf{b}_t$. As a consequence, $\mathbf{b}_{l_{t+1}}$ has two sources of error: the fact that the belief update is performed on a previous estimator and the sampling process that comes after it.

Definition \ref{def:lookahead-error} introduces the lookahead error using the notion of \textit{regret} \cite{liu2021regret}, a scalar quantity generally used in RL to define an algorithm's error. It is defined as the absolute difference between the utility of the best action that could be taken and the utility of the action that the agent actually took. The remainder of this section serves to provide and analyze an upper bound to this error in order to give theoretical guarantees concerning the performance of the algorithm. These theoretical guarantees are the conditions under which the algorithm can be considered $\epsilon$-approximate with a confidence interval of $1-\delta$, as per Equation \eqref{eq:pac-learnability}.

\begin{definition}
    (Lookahead error). Let $a_t^{\mathcal{L}} = \argmax_{a_t \in \mathcal{A}} Q^{\mathcal{L}}(\mathbf{b}_{l_t}, a_t)$ be the near optimal action chosen by the algorithm at time $t$. The lookahead error $\epsilon_t$ at time $t$ is defined as the regret:
    \begin{equation}
        \epsilon_t = \left| \max_{a_t \in \mathcal{A}} Q^{\star}(\mathbf{b}_t, a_t) - Q^{\star}(\mathbf{b}_{l_t}, a_t^{\mathcal{L}}) \right|
    \end{equation}
    \label{def:lookahead-error}
\end{definition}

Using Lemmas \ref{lemma:max-diff} and \ref{lemma:triangle-ineq}, Lemma \ref{lemma:error-split} decomposes the lookahead error into three distinct errors, each one easier to bound than the original one:
\begin{lemma}
    The lookahead error $\epsilon_t$ can be bounded by the following expression:
    \begin{equation*}
        \epsilon_t \le \max_{a_t \in \mathcal{A}} \left( \left| Q^{\star}(\mathbf{b}_t, a_t) - Q^{\mathcal{L}}(\mathbf{b}_t, a_t) \right| + \left| Q^{\mathcal{L}}(\mathbf{b}_t, a_t) - Q^{\mathcal{L}}(\mathbf{b}_{l_t}, a_t) \right| \right) + \left| Q^{\star}(\mathbf{b}_{l_t}, a_t^{\mathcal{L}}) - Q^{\mathcal{L}}(\mathbf{b}_{l_t}, a_t^{\mathcal{L}}) \right|
    \end{equation*}
    \label{lemma:error-split}
\end{lemma}
\begin{proof}
    \begin{align*}
        \begin{split}
            \epsilon_t \stackrel{(\ref{def:lookahead-error})}{=}& \; \left| \max_{a_t \in \mathcal{A}} Q^{\star}(\mathbf{b}_t, a_t) - Q^{\star}(\mathbf{b}_{l_t}, a_t^{\mathcal{L}}) \right| \\
            \stackrel{(\ref{lemma:triangle-ineq})}{\le}& \; \left| \max_{a_t \in \mathcal{A}} Q^{\star}(\mathbf{b}_t, a_t) - \max_{a_t \in \mathcal{A}} Q^{\mathcal{L}}(\mathbf{b}_{l_t}, a_t) \right| + \left| Q^{\star}(\mathbf{b}_{l_t}, a_t^{\mathcal{L}}) - \max_{a_t \in \mathcal{A}} Q^{\mathcal{L}}(\mathbf{b}_{l_t}, a_t) \right| \\
            \stackrel{(\ref{lemma:max-diff})}{\le}& \; \max_{a_t \in \mathcal{A}} \left| Q^{\star}(\mathbf{b}_t, a_t) - Q^{\mathcal{L}}(\mathbf{b}_{l_t}, a_t) \right| + \left| Q^{\star}(\mathbf{b}_{l_t}, a_t^{\mathcal{L}}) - Q^{\mathcal{L}}(\mathbf{b}_{l_t}, a_t^{\mathcal{L}}) \right|\\
            =& \; \max_{a_t \in \mathcal{A}} \left| \left( Q^{\star}(\mathbf{b}_t, a_t) - Q^{\mathcal{L}}(\mathbf{b}_t, a_t) \right) + \left( Q^{\mathcal{L}}(\mathbf{b}_t, a_t) - Q^{\mathcal{L}}(\mathbf{b}_{l_t}, a_t) \right) \right| + \left| Q^{\star}(\mathbf{b}_{l_t}, a_t^{\mathcal{L}}) - Q^{\mathcal{L}}(\mathbf{b}_{l_t}, a_t^{\mathcal{L}}) \right|\\
            \stackrel{(\ref{lemma:triangle-ineq})}{\le}& \; \max_{a_t \in \mathcal{A}} \left( \left| Q^{\star}(\mathbf{b}_t, a_t) - Q^{\mathcal{L}}(\mathbf{b}_t, a_t) \right| + \left| Q^{\mathcal{L}}(\mathbf{b}_t, a_t) - Q^{\mathcal{L}}(\mathbf{b}_{l_t}, a_t) \right| \right) + \left| Q^{\star}(\mathbf{b}_{l_t}, a_t^{\mathcal{L}}) - Q^{\mathcal{L}}(\mathbf{b}_{l_t}, a_t^{\mathcal{L}}) \right|
        \end{split}
    \end{align*}
\end{proof}

Lemma \ref{lemma:error-split} states that the total lookahead error can be decomposed into three distinct errors:
\begin{itemize}
    \item $\left| Q^{\star}(\mathbf{b}_t,a_t) - Q^{\mathcal{L}}(\mathbf{b}_t, a_t) \right|$: the error between the true action value function and the approximate one, both evaluated at belief state $\mathbf{b}_t$.
    
    \item $\left| Q^{\mathcal{L}}(\mathbf{b}_t, a_t) - Q^{\mathcal{L}}\left(\mathbf{b}_{l_t}, a_t\right) \right|$: the error between the approximate action value function evaluated at the true belief state $\mathbf{b}_t$ and at the approximate belief-state $\mathbf{b}_{l_t}$.
    
    \item $\left| Q^{\star}(\mathbf{b}_{l_t}, a_t^{\mathcal{L}}) - Q^{\mathcal{L}}(\mathbf{b}_{l_t}, a_t^{\mathcal{L}}) \right|$: the error between the true and approximate action value functions, both evaluated at the approximate belief $\mathbf{b}_{l_t}$ and the near-optimal action $a^{\star}_t$.
\end{itemize}

The remainder of this section determines upper bounds to all error terms in order to define an upper bound for the lookahead error. The main result that allows defining these bounds is Hoeffding's inequality, stated in Theorem \ref{theorem:hoeffding}\footnote{$\lVert \cdot \rVert_{\infty}$ denotes the $l_{\infty}$ norm of a vector, which extracts its maximum element. If $u = \lVert \mathbf{U} \rVert_{\infty}$, then $u \ge u^{\prime}, \forall u^{\prime} \in \mathbf{U}$.}.
\begin{theorem}
    (Hoeffding's Inequality \cite{sidford2018near}). Let $\mathbf{p} \in \Delta_{\mathcal{V}}$ be a probability vector and $\mathbf{V} \in \mathbb{R}^{\left| \mathcal{V} \right|}$ a vector. Let $\mathbf{p}_m \in \Delta_{\mathcal{V}}$ be an empirical estimation of $\mathbf{p}$ using $m$ i.i.d. \footnote{Independent and identically distributed} samples and $\delta \in (0, 1)$ a parameter. Then, with probability at least $1-\delta$:
    \begin{equation*}
        \left| \mathbf{p}^{\top} \mathbf{V} - \mathbf{p}_m^{\top} \mathbf{V} \right| \le \lVert V \rVert_{\infty} \mathcal{H}(m, \delta)
    \end{equation*}
    Where $\mathcal{H}(m, \delta) = \sqrt{ \frac{1}{2m} \log \left( \frac{2}{\delta}  \right) }$.
    \label{theorem:hoeffding}
\end{theorem}

\vspace{1em}

The results above already provide the necessary tools to bound two of the error sources that compose the lookahead error as in Lemma \ref{lemma:error-split}. Using these results, Lemma \ref{lemma:true-error} bounds the error between the true value function and the approximate one, both evaluated at the same belief state and action:
\begin{lemma}
    Let $\Gamma = \left(1-\gamma\right)^{-1}$ be the effective horizon, $r_{\text{max}} = \lVert \mathbf{r} \rVert_{\infty}$ be the maximum reward and $\sigma_t \in (0, 1)$ a parameter defining the confidence interval $1-\sigma_t$ for sampling both $\mathbf{p}_{r_t}$ and $\mathbf{p}_{o_t}$ at time-step $t$. The error $\left| Q^{\star}(\mathbf{b}_t, a_t) - Q^{\mathcal{L}}(\mathbf{b}_t, a_t) \right|$ has the following upper bound:
    \begin{equation*}
        \left| Q^{\star}(\mathbf{b}_t, a_t) - Q^{\mathcal{L}}(\mathbf{b}_t, a_t) \right| \le r_{\text{max}} \sum_{k=0}^{H-1} \gamma^k \left( \mathcal{H}(n_{t+k}, \sigma_{t+k}) + \gamma \Gamma \mathcal{H}(m_{t+k}, \sigma_{t+k}) \right) + \gamma^{H} \Gamma r_{\text{max}}
    \end{equation*}
    \label{lemma:true-error}
\end{lemma}
\begin{proof}
    Using the definitions for $Q^{\star}$ and $Q^{\mathcal{L}}$ from Equations \eqref{eq:qfunction-vector} and \eqref{eq:qfunc-approx-vector-rec}, respectively:
    \begin{align*}
        \begin{split}
            \left| Q^{\star}(\mathbf{b}_t, a_t) - Q^{\mathcal{L}}(\mathbf{b}_t, a_t) \right| \stackrel{(\ref{lemma:triangle-ineq})}{\le}& \; \left| \mathbf{r}^{\top} \mathbf{p}_{r_t} - \mathbf{r}^{\top} \mathbf{p}_{r_t,n_t} \right| + \gamma \left| \mathbf{p}_{o_t}^{\top} \mathbf{V}_{t+1}^{\star} - \mathbf{p}_{o_t,m_t}^{\top} \mathbf{V}^{\mathcal{L}}_{t+1} \right|
        \end{split}
    \end{align*}
    
    Identity $\mathbf{p}_{o_t} - \mathbf{p}_{o_t, m_t} \le \left| \mathbf{p}_{o_t} - \mathbf{p}_{o_t, m_t} \right| \Leftrightarrow - \mathbf{p}_{o_t, m_t} \le \left| \mathbf{p}_{o_t} - \mathbf{p}_{o_t, m_t} \right| - \mathbf{p}_{o_t}$ yields:
    \begin{equation*}
        \left| Q^{\star}(\mathbf{b}_t, a_t) - Q^{\mathcal{L}}(\mathbf{b}_t, a_t) \right| \stackrel{(\ref{lemma:triangle-ineq})}{\le} \; \left| \mathbf{r}^{\top} \mathbf{p}_{r_t} - \mathbf{r}^{\top} \mathbf{p}_{r_t,n_t} \right| + \gamma \left| \mathbf{p}_{o_t}^{\top} \mathbf{V}^{\mathcal{L}}_{t+1} - \mathbf{p}_{o_t, m_t}^{\top} \mathbf{V}^{\mathcal{L}}_{t+1} \right| + \gamma \left| \mathbf{p}_{o_t}^{\top} \mathbf{V}_{t+1}^{\star} - \mathbf{p}_{o_t}^{\top} \mathbf{V}^{\mathcal{L}}_{t+1} \right|
    \end{equation*}
    
    The terms $\left| \mathbf{r}^{\top} \mathbf{p}_{r_t} - \mathbf{r}^{\top} \mathbf{p}_{r_t,n_t} \right|$ and $\left| \mathbf{p}_{o_t}^{\top} \mathbf{V}^{\mathcal{L}}_{t+1} - \mathbf{p}_{o_t, m_t}^{\top} \mathbf{V}^{\mathcal{L}}_{t+1} \right|$ can be bounded through Hoeffding's inequality in Theorem \ref{theorem:hoeffding} and the fact that, by definition, $\lVert \mathbf{V}^{\mathcal{L}}_{t+1} \rVert_{\infty} \le Q_{\text{max}} = \Gamma r_{\text{max}}$. The other term can be decomposed into a sum:
    \begin{align*}
        \begin{split}
            \left| Q^{\star}(\mathbf{b}_t, a_t) - Q^{\mathcal{L}}(\mathbf{b}_t, a_t) \right| \stackrel{(\ref{theorem:hoeffding})}{\le}& \; r_{\text{max}} \left( \mathcal{H}(n_{t}, \sigma_{t}) + \gamma \Gamma \mathcal{H}(m_{t}, \sigma_{t}) \right)\\
        &+ \gamma \max_{a_{t+1} \in \mathcal{A}} \sum_{o_{t+1} \in \Omega} P(o_{t+1}|\mathbf{b}_t, a_t) \left| Q^{\star}(\mathbf{b}_{t+1}, a_{t+1}) - Q^{\mathcal{L}}(\mathbf{b}_{t+1}, a_{t+1}) \right|
        \end{split}
    \end{align*}
    
    Let $\mathbf{\mu}_{t+1}$ be a belief-state and $a_{t+1}^{\prime}$ an action such that 
    $$\left| Q^{\star}(\mathbf{\mu}_{t+1}, a_{t+1}^{\prime}) - Q^{\mathcal{L}}(\mathbf{\mu}_{t+1}, a_{t+1}^{\prime}) \right| = \max_{a_{t+1}, \mathbf{b}_{t+1}} \left| Q^{\star}(\mathbf{b}_{t+1}, a_{t+1}) - Q^{\mathcal{L}}(\mathbf{b}_{t+1}, a_{t+1}) \right|$$
    Then:
    \begin{align*}
        \begin{split}
            \left| Q^{\star}(\mathbf{b}_t, a_t) - Q^{\mathcal{L}}(\mathbf{b}_t, a_t) \right| \le& \; r_{\text{max}} \left( \mathcal{H}(n_{t}, \sigma_{t}) + \gamma \Gamma \mathcal{H}(m_{t}, \sigma_{t}) \right)\\
            &+ \gamma \left| Q^{\star}(\mathbf{\mu}_{t+1}, a_{t+1}^{\prime}) - Q^{\mathcal{L}}(\mathbf{\mu}_{t+1}, a_{t+1}^{\prime}) \right| \sum_{o_{t+1} \in \Omega} P(o_{t+1}|\mathbf{b}_t,a_t)\\
            =& \; r_{\text{max}} \left( \mathcal{H}(n_{t}, \sigma_{t}) + \gamma \Gamma \mathcal{H}(m_{t}, \sigma_{t}) \right)\\
            &+ \gamma \left| Q^{\star}(\mathbf{\mu}_{t+1}, a_{t+1}^{\prime}) - Q^{\mathcal{L}}(\mathbf{\mu}_{t+1}, a_{t+1}^{\prime}) \right|
        \end{split}
    \end{align*}

   Finally, with the recursive application of this inequality to $\left| Q^{\star}(\mathbf{\mu}_{t+1}, a_{t+1}^{\prime}) - Q^{\mathcal{L}}(\mathbf{\mu}_{t+1}, a_{t+1}^{\prime}) \right|$, with Equation \eqref{eq:qfunc-approx-vector-stop} as a stop case to the approximate action value function, the result of this lemma is attained:
    \begin{align*}
        \begin{split}
            \left| Q^{\star}(\mathbf{b}_t, a_t) - Q^{\mathcal{L}}(\mathbf{b}_t, a_t) \right| \le& \; r_{\text{max}} \sum_{k=0}^{H-1} \gamma^k \left( \mathcal{H}(n_{t+k}, \sigma_{t+k}) + \gamma \Gamma \mathcal{H}(m_{t+k}, \sigma_{t+k}) \right)\\
            &+ \gamma^H \left| Q^{\star}(\mathbf{\mu}_{t+H},a_{t+H}^{\prime}) - Q^{\mathcal{L}}(\mathbf{\mu}_{t+H},a_{t+H}^{\prime}) \right|\\
            \stackrel{(\eqref{eq:qfunc-approx-vector-stop})}{=}& \; r_{\text{max}} \sum_{k=0}^{H-1} \gamma^k \left( \mathcal{H}(n_{t+k}, \sigma_{t+k}) + \gamma \Gamma \mathcal{H}(m_{t+k}, \sigma_{t+k}) \right)\\
            &+ \gamma^H \left| Q^{\star}(\mathbf{\mu}_{t+H},a_{t+H}^{\star})\right|\\
            \le& \; r_{\text{max}} \sum_{k=0}^{H-1} \gamma^k \left( \mathcal{H}(n_{t+k}, \sigma_{t+k}) + \gamma \Gamma \mathcal{H}(m_{t+k}, \sigma_{t+k}) \right) + \gamma^{H} \Gamma r_{\text{max}}
        \end{split}
    \end{align*}
\end{proof}

Lemma \ref{lemma:true-error} gives a bound to the difference between the true value function and the value function calculated by the lookahead algorithm, both evaluated at the same action and belief state. Notice, however, that the bound holds for any choice of action and belief state, because it is independent of both these quantities. Thus, both $\left| Q^{\star}(\mathbf{b}_t,a_t) - Q^{\mathcal{L}}(\mathbf{b}_t, a_t) \right|$ and $\left| Q^{\star}(\mathbf{b}_{l_t},a^{\mathcal{L}}_t) - Q^{\mathcal{L}}(\mathbf{b}_{l_t},a^{\mathcal{L}}_t) \right|$ can be bounded by this lemma. 

Before moving on to the next lemma, let us first count the total number of Hoeffding inequalities that were used in the proof of Lemma \ref{lemma:true-error}. This count will be useful later on to calculate the confidence interval of the lookahead algorithm's bounds. The bound in Lemma \ref{lemma:true-error} was attained by splitting the error source $\left| Q^{\star}(\mathbf{b}_{t}, a_{t}) - Q^{\mathcal{L}}(\mathbf{b}_{t}, a_{t}) \right|$ into two terms, one where the Hoeffding bounds are applied ($\left| \mathbf{r}^{\top} \mathbf{p}_{r_t} - \mathbf{r}^{\top} \mathbf{p}_{r_t,n_t} \right| + \gamma \left| \mathbf{p}_{o_t}^{\top} \mathbf{V}^{\mathcal{L}}_{t+1} - \mathbf{p}_{o_t, m_t}^{\top} \mathbf{V}^{\mathcal{L}}_{t+1} \right|$) and a recursive term ($\gamma \left| \mathbf{p}_{o_t}^{\top} \mathbf{V}_{t+1}^{\star} - \mathbf{p}_{o_t}^{\top} \mathbf{V}^{\mathcal{L}}_{t+1} \right|$) where the same procedure is applied until a stop case is found. Let's call these two terms \textit{immediate} and \textit{recursive}, respectively. To bound the immediate term, $2$ Hoeffding inequalities were used. For the first time-step in the recursive term, $2\mathcal{A}\Omega$ Hoeffding inequalities are needed ($2$ Hoeffding inequalities as per the immediate term, but for every action and every observation). As such, a total of $2\sum_{k=0}^{H-1}\mathcal{A}^k\Omega^k$ Hoeffding inequalities are used to derive the bound in Lemma \ref{lemma:true-error}.

To bound the full lookahead error, only the error of approximating the belief state remains to be bounded. As mentioned above, the belief states for the lookahead algorithm are approximated in a different manner from the reward and observation distributions. The only true belief state that is known to the lookahead algorithm is the prior $\mathbf{b}_0$, and therefore, the belief estimations $\mathbf{b}_{l_t}$ are calculated by performing a belief update on the previous estimation $\mathbf{b}_{l_{t-1}}$, yielding the distribution $\mathbf{p}_t \in \Delta_{\mathcal{S}}$ with entries $P(s_t|o_t,a_{t-1},\mathbf{b}_{l_{t-1}})$, and extracting $l_t$ samples from it, creating an empirical estimation distribution $\mathbf{p}_{t,l_t} = \mathbf{b}_{l_t} \in \Delta_{\mathcal{S}}$ with entries $P_{l_t}(s_t|o_t,a_{t-1},\mathbf{b}_{l_{t-1}})$. As such, Hoeffding's inequality can not be applied to the distributions $\mathbf{b}_t$ and $\mathbf{b}_{l_t}$, because they are not direct estimations of one another. It can be applied, however, to $\mathbf{p}_t$ and $\mathbf{p}_{t,l_t}$, since the latter approximates the former distribution via sampling, as stated above. 

Lemma \ref{lemma:belief-error} uses an indirect approach to applying the Hoeffding inequalities and still be able to bound the error between the true belief state $\mathbf{b}_t$ and the approximate belief state $\mathbf{b}_{l_t}$ of the algorithm at time $t$:
\begin{lemma}
    Given a distribution $\mathbf{p}_{s_t}$ from a family of distributions $P(\cdot|s_t,\cdot)$ conditioned on the state $s_t \in \mathcal{S}$ \footnote{In the results of this lemma and the lemmas and theorems that follow, the notation $\mathcal{X}$ for a set $\mathcal{X}$ can have two meanings, that depend on the context of use: it can refer to the set itself, or to the size of the set $\left| \mathcal{X} \right|$.}, the following inequality bounds the error of approximating the belief-state at time $t$, using $l_k$ samples and a confidence interval of $1 - \sigma_k$ to estimate it at time-step $k$:
    
    \begin{equation*}
        \left| \mathbf{p}_{s_t}^{\top} \mathbf{b}_t - \mathbf{p}_{s_t}^{\top} \mathbf{b}_{l_t} \right| \le \mathcal{S}^t 
        \sum_{k=0}^{t} \frac{1}{\mathcal{S}^k} \mathcal{H}(l_{k}, \sigma_{k})
    \end{equation*}
    \label{lemma:belief-error}
\end{lemma}
\begin{proof}
    \begin{equation*}
        \left| \mathbf{p}_{s_t}^{\top} \mathbf{b}_t - \mathbf{p}_{s_t}^{\top} \mathbf{b}_{l_t} \right| = \left| \sum_{s_t \in \mathcal{S}} P(\cdot|s_t,\cdot) \left( P(s_t|o_t, a_{t-1}, \mathbf{b}_{t-1}) - P_{l_t}(s_t|o_t, a_{t-1}, \mathbf{b}_{l_{t-1}}) \right) \right|
    \end{equation*}
    
	Notice that the difference $P(s_t|o_t, a_{t-1}, \mathbf{b}_{t-1}) - P_{l_t}(s_t|o_t, a_{t-1}, \mathbf{b}_{l_{t-1}})$ in the equation above is between two distributions conditioned on different belief states. As such, Hoeffding's inequality can't be directly applied. To overcome this issue, the triangle inequality 
    
    $$P(s_t|o_t, a_{t-1}, \mathbf{b}_{t-1}) \le P(s_t|o_t, a_{t-1}, \mathbf{b}_{l_{t-1}}) + \left| P(s_t|o_t, a_{t-1}, \mathbf{b}_{t-1}) - P(s_t|o_t, a_{t-1}, \mathbf{b}_{l_{t-1}}) \right| $$
    $$ = P(s_t|o_t, a_{t-1}, \mathbf{b}_{l_{t-1}}) + \left| \mathbf{p}_{s_{t-1}}^{\top} \mathbf{b}_{t-1} - \mathbf{p}_{s_{t-1}}^{\top} \mathbf{b}_{l_{t-1}} \right|$$ 
    
    \noindent is used, allowing the expression above to be re-written as:
    \begin{align*}
        \begin{split}
            \left| \mathbf{p}_{s_t}^{\top} \mathbf{b}_t - \mathbf{p}_{s_t}^{\top} \mathbf{b}_{l_t} \right| \le& \; \sum_{s_t \in \mathcal{S}} P(\cdot|s_t, \cdot) \left( P(s_t|o_t,a_{t-1},\mathbf{b}_{l_{t-1}}) - P_{l_t}(s_t|o_t,a_{t-1},\mathbf{b}_{l_{t-1}}) + \left| \mathbf{p}_{s_{t-1}}^{\top} \mathbf{b}_{t-1} - \mathbf{p}_{s_{t-1}}^{\top} \mathbf{b}_{l_{t-1}} \right| \right)\\
            \le& \; \left| \mathbf{p}_{s_t}^{\top} \mathbf{p}_t - \mathbf{p}_{s_t}^{\top} \mathbf{p}_{t, l_t} \right| + \mathcal{S} \left| \mathbf{p}_{s_{t-1}}^{\top} \mathbf{b}_{t-1} - \mathbf{p}_{s_{t-1}}^{\top} \mathbf{b}_{l_{t-1}} \right|
        \end{split}
    \end{align*}
    
    Recursively applying the inequality above and using Hoeffding's inequality:
    \begin{align*}
        \begin{split}
            \left| \mathbf{p}_{s_t}^{\top} \mathbf{b}_t - \mathbf{p}_{s_t}^{\top} \mathbf{b}_{l_t} \right| \le& \; \sum_{k=0}^t \mathcal{S}^{t-k} \left| \mathbf{p}_{s_k}^{\top} \mathbf{p}_k - \mathbf{p}_{s_k}^{\top} \mathbf{p}_{k,l_k} \right|\\ \stackrel{(\ref{theorem:hoeffding})}{\le}& \; \mathcal{S}^t
            \sum_{k=0}^{t} \frac{1}{\mathcal{S}^k} \mathcal{H}(l_{k}, \sigma_{k})
        \end{split}
    \end{align*}
\end{proof}

This lemma states that the belief-approximation error bound is cumulative over time. This is intuitive, given that the approximation of a belief state at time-step $t$ is attained by estimating over the previous estimation. Therefore, estimation errors pile up as time progresses.

A few more bounds are useful to derive the full upper bound of the belief error. Due to the fact that an inaccurate belief state is used throughout every time step of the algorithm, not only does this cause errors in the algorithm's belief states themselves, but also in calculating other probability distributions that are conditioned on the belief state (because they should be conditioned on the true belief-state rather than an estimation of it). This implies that both reward sampling and observation sampling also suffer from this belief-state approximation error, and so these error contributions should be accounted for. Lemma \ref{lemma:belief-reward-error} bounds the reward sampling error caused by this belief-approximation:
\begin{lemma}
    Let $\mathbf{p} \in \Delta_{\mathcal{R}}$ represent the probability distribution $P(r_{t+1}|\mathbf{b}_t)$ and $\hat{\mathbf{p}} \in \Delta_{\mathcal{R}}$ represent distribution $P(r_{t+1}|\mathbf{b}_{l_t})$. Also, let $\mathbf{p}_{n_t}$ and $\hat{\mathbf{p}}_{n_t}$ be their respective estimations using $n_t$ samples. Then, the following inequality holds:
    \begin{equation*}
        \left| \mathbf{p}^{\top}_{n_t} \mathbf{r} - \hat{\mathbf{p}}^{\top}_{n_t} \mathbf{r} \right| \le r_{\text{max}} \left(2 \mathcal{H}(n_{t}, \sigma_{t}) + \mathcal{R} \mathcal{S}^t \sum_{k=0}^t \frac{1}{\mathcal{S}^k} \mathcal{H}(l_{k}, \sigma_{k}) \right)
    \end{equation*}
    \label{lemma:belief-reward-error}
\end{lemma}
\begin{proof}
    \begin{align*}
        \begin{split}
            \left| \mathbf{p}^{\top}_{n_t} \mathbf{r} - \hat{\mathbf{p}}^{\top}_{n_t} \mathbf{r} \right| \le& \; \left| \left( \mathbf{p}^{\top} + \left| \mathbf{p}^{\top} - \mathbf{p}^{\top}_{n_t} \right| \right) \mathbf{r} + \left( - \hat{\mathbf{p}}^{\top} + \left| \hat{\mathbf{p}}^{\top} - \hat{\mathbf{p}}^{\top}_{n_t} \right| \right) \mathbf{r} \right|\\
            \stackrel{(\ref{lemma:triangle-ineq})}{\le}& \left| \mathbf{p}^{\top} \mathbf{r} - \mathbf{p}^{\top}_{n_t} \mathbf{r} \right| + \left| \hat{\mathbf{p}}^{\top} \mathbf{r} - \hat{\mathbf{p}}^{\top}_{n_t} \mathbf{r} \right| + \left| \mathbf{p}^{\top} \mathbf{r} - \hat{\mathbf{p}}^{\top} \mathbf{r} \right|\\ \stackrel{(\ref{theorem:hoeffding})}{\le}& \; 2 r_{\text{max}} \mathcal{H}(n_{t}, \sigma_{t}) + \sum_{r_t \in \mathcal{R}} r_t \left| \mathbf{p}_{s_t}^{\top} \mathbf{b}_t - \mathbf{p}_{s_t}^{\top} \mathbf{b}_{l_t} \right|\\ \stackrel{(\ref{lemma:belief-error})}{\le}& \; r_{\text{max}} \left(2 \mathcal{H}(n_{t}, \sigma_{t}) + \mathcal{R} \mathcal{S}^t \sum_{k=0}^t \frac{1}{\mathcal{S}^k} \mathcal{H}(l_{k}, \sigma_{k}) \right)
        \end{split}
    \end{align*}
\end{proof}

Analogously to the Lemma \ref{lemma:belief-reward-error}, Lemma \ref{lemma:belief-observation-error} bounds the observation sampling error due to the belief-state approximation:
\begin{lemma}
    Let $\mathbf{p}_{o_t} \in \Delta_{\Omega}$ be a probability distribution $P(o_{t+1}|\mathbf{b}_t,a_t)$ and $\hat{\mathbf{p}}_{o_t} \in \Delta_{\Omega}$ be $P(o_{t+1}|\mathbf{b}_{l_t}, a_t)$, such that $\mathbf{p}_{o_t,m_t}$ and $\hat{\mathbf{p}}_{o_t,m_t}$ are their empirical estimations with $m_t$ samples. Also, let $\mathbf{V} \in \mathbb{R}^{\left|\Omega\right|}$ be a vector of value functions with entries $V(\mathbf{b}_t)$ and $\mathbf{V}^{\mathcal{L}}$ a similar vector with entries $V^{\mathcal{L}}(\mathbf{b}_{l_t})$. Then, the following inequality holds:
    \begin{equation*}
        \left| \mathbf{p}^{\top}_{o_t,m_t} \mathbf{V} - \hat{\mathbf{p}}^{\top}_{o_t,m_t} \mathbf{V}^{\mathcal{L}} \right| \le \Gamma r_{\text{max}} \left( 2 \mathcal{H}(m_{t}, \sigma_{t}) + \Omega\mathcal{S}^t\sum_{k=0}^t \frac{1}{\mathcal{S}^k} \mathcal{H}(l_{k}, \sigma_{k}) \right) + \left| \hat{\mathbf{p}}^{\top}_{o_t,m_t} \mathbf{V} - \hat{\mathbf{p}}^{\top}_{o_t,m_t} \mathbf{V}^{\mathcal{L}} \right|
    \end{equation*}
    \label{lemma:belief-observation-error}
\end{lemma}
\begin{proof}
    \begin{align*}
        \begin{split}
            \left| \mathbf{p}^{\top}_{o_t,m_t} \mathbf{V} - \hat{\mathbf{p}}^{\top}_{o_t,m_t} \mathbf{V}^{\mathcal{L}} \right| \le& \; \left| \left( \mathbf{p}_{o_t}^{\top} + \left| \mathbf{p}_{o_t}^{\top} - \mathbf{p}_{o_t,m_t}^{\top} \right| \right) \mathbf{V} + \left( - \hat{\mathbf{p}}_{o_t}^{\top} + \left| \hat{\mathbf{p}}_{o_t}^{\top} - \hat{\mathbf{p}}_{o_t,m_t}^{\top} \right| \right) \mathbf{V}^{\mathcal{L}} \right|\\
            \stackrel{(\ref{lemma:triangle-ineq})}{\le}& \; \left| \mathbf{p}_{o_t}^{\top} \mathbf{V} - \mathbf{p}_{o_t,m_t}^{\top} \mathbf{V} \right| + \left| \hat{\mathbf{p}}_{o_t}^{\top} \mathbf{V}^{\mathcal{L}} - \hat{\mathbf{p}}_{o_t,m_t}^{\top} \mathbf{V}^{\mathcal{L}} \right| + \left| \mathbf{p}_{o_t}^{\top} \mathbf{V} - \hat{\mathbf{p}}_{o_t}^{\top} \mathbf{V}^{\mathcal{L}} \right|\\
            \stackrel{(\ref{theorem:hoeffding})}{\le}& \; 2 \Gamma r_{\text{max}} \mathcal{H}(m_{t}, \sigma_{t}) + \left| \mathbf{p}_{o_t}^{\top} \mathbf{V} - \hat{\mathbf{p}}_{o_t}^{\top} \mathbf{V} \right| + \left| \hat{\mathbf{p}}_{o_t}^{\top} \mathbf{V} - \hat{\mathbf{p}}_{o_t}^{\top} \mathbf{V}^{\mathcal{L}} \right|\\
            =& \; 2 \Gamma r_{\text{max}} \mathcal{H}(m_{t}, \sigma_{t}) + \sum_{o_{t+1} \in \Omega} V^{\mathcal{L}}(\mathbf{b}_{t+1}) \left| \mathbf{p}_{s_t}^{\top} \mathbf{b}_{t} - \mathbf{p}_{s_t}^{\top} \mathbf{b}_{l_t} \right| + \left| \hat{\mathbf{p}}_{o_t}^{\top} \mathbf{V} - \hat{\mathbf{p}}_{o_t}^{\top} \mathbf{V}^{\mathcal{L}} \right|\\
            \le& \; 2 \Gamma r_{\text{max}} \mathcal{H}(m_{t}, \sigma_{t}) + \Gamma r_{\text{max}} \Omega \left| \mathbf{p}_{s_t}^{\top} \mathbf{b}_{t} - \mathbf{p}_{s_t}^{\top} \mathbf{b}_{l_t} \right| + \left| \hat{\mathbf{p}}_{o_t}^{\top} \mathbf{V} - \hat{\mathbf{p}}_{o_t}^{\top} \mathbf{V}^{\mathcal{L}} \right|\\
            \stackrel{(\ref{lemma:belief-error})}{\le}& \; \Gamma r_{\text{max}} \left( 2\mathcal{H}(m_{t}, \sigma_{t}) + \Omega\mathcal{S}^t\sum_{k=0}^t \frac{1}{\mathcal{S}^k} \mathcal{H}(l_{k}, \sigma_{k})\right) + \left| \hat{\mathbf{p}}^{\top}_{o_t,m_t} \mathbf{V} - \hat{\mathbf{p}}^{\top}_{o_t,m_t} \mathbf{V}^{\mathcal{L}} \right|
        \end{split}
    \end{align*}
    
    The final recursive term $\left| \hat{\mathbf{p}}^{\top}_{o_t,m_t} \mathbf{V} - \hat{\mathbf{p}}^{\top}_{o_t,m_t} \mathbf{V}^{\mathcal{L}} \right|$ is expanded and used to give an upper bound to the full belief error in Lemma \ref{lemma:lookahead-belief-error}.
\end{proof}

Finally, using the results of Lemmas \ref{lemma:belief-reward-error} and \ref{lemma:belief-observation-error}, Lemma \ref{lemma:lookahead-belief-error} provides an upper bound to the belief approximation error:
\begin{lemma}
    The belief error $\left| Q^{\mathcal{L}}(\mathbf{b}_t, a_t) - Q^{\mathcal{L}}(\mathbf{b}_{l_t},a_t) \right|$ is bounded by the following expression \footnote{In this bound, and throughout the remaining of this section, it is assumed that every sampling operation on the same time-step $t$ (be it reward, observation or belief-state sampling) has the same confidence interval $1 - \sigma_t$.}:
    \begin{multline*}
        \left| Q^{\mathcal{L}}(\mathbf{b}_t, a_t) - Q^{\mathcal{L}}(\mathbf{b}_{l_t},a_t) \right| \le r_{\text{max}}\\
        \sum_{k=0}^{H-1} \gamma^k \left(2 \mathcal{H}(n_{t+k}, \sigma_{t+k}) + 2\gamma\Gamma \mathcal{H}(m_{t+k}, \sigma_{t+k}) + \left( \mathcal{R} + \gamma \Gamma \Omega \right) \mathcal{S}^{t+k} \sum_{j=0}^{t+k} \frac{1}{\mathcal{S}^j} \mathcal{H}(l_{j}, \sigma_{j}) \right)
    \end{multline*}
    \label{lemma:lookahead-belief-error}
\end{lemma}
\begin{proof}
    \begin{equation*}
        \left| Q^{\mathcal{L}}(\mathbf{b}_t, a_t) - Q^{\mathcal{L}}(\mathbf{b}_{l_t},a_t) \right| \le \left| \mathbf{p}^{\top}_{n_t} \mathbf{r} - \hat{\mathbf{p}}^{\top}_{n_t} \mathbf{r} \right| + \gamma \left| \mathbf{p}_{o_t,m_t}^{\top} \mathbf{V} - \hat{\mathbf{p}}_{o_t,m_t}^{\top} \mathbf{V}^{\mathcal{L}} \right|
    \end{equation*}
    
    \noindent using Lemmas \ref{lemma:belief-reward-error} and \ref{lemma:belief-observation-error}:
    \begin{align*}
        \begin{split}
            \left| Q^{\mathcal{L}}(\mathbf{b}_t, a_t) - Q^{\mathcal{L}}(\mathbf{b}_{l_t},a_t) \right| \le& \; r_{\text{max}} \left( 2 \mathcal{H}(n_t,\sigma_t) + 2\gamma\Gamma\mathcal{H}(m_t,\sigma_t) + \left( \mathcal{R} + \gamma \Gamma \Omega \right) \mathcal{S}^t \sum_{j=0}^{t} \frac{1}{\mathcal{S}^j} \mathcal{H}(l_j,\sigma_j) \right)\\
            &+ \gamma \left| \hat{\mathbf{p}}^{\top}_{o_t,m_t} \mathbf{V} - \hat{\mathbf{p}}^{\top}_{o_t,m_t} \mathbf{V}^{\mathcal{L}} \right|\\
            \le& \; r_{\text{max}} \left( 2 \mathcal{H}(n_t,\sigma_t) + 2\gamma\Gamma\mathcal{H}(m_t,\sigma_t) + \left( \mathcal{R} + \gamma \Gamma \Omega \right) \mathcal{S}^t \sum_{j=0}^{t} \frac{1}{\mathcal{S}^j} \mathcal{H}(l_j,\sigma_j) \right)\\
            &+ \gamma \sum_{o_{t+1} \in \Omega} P_{m_t}(o_{t+1}|\mathbf{b}_{l_t},a_t) \max_{a_{t+1}} \left| Q^{\mathcal{L}}(\mathbf{b}_{t+1}, a_{t+1}) - Q^{\mathcal{L}}(\mathbf{b}_{l_{t+1}}, a_{t+1}) \right|
        \end{split}
    \end{align*}
    
    \noindent let $\mathbf{\mu}_{t+1}$ and $a^{\prime}_{t+1}$ be a belief-state and an action, respectively, such that\\ $\left| Q^{\mathcal{L}}(\mathbf{\mu}_{t+1}, a^{\prime}_{t+1}) - Q^{\mathcal{L}}(\mathbf{\mu}_{l_{t+1}}, a^{\prime}_{t+1}) \right| \ge \left| Q^{\mathcal{L}}(\mathbf{b}_{t+1}, a_{t+1}) - Q^{\mathcal{L}}(\mathbf{b}_{l_{t+1}}, a_{t+1}) \right|, \forall \mathbf{b}_{t+1}, a_{t+1}$. Then:
    \begin{multline*}
    		\left| Q^{\mathcal{L}}(\mathbf{b}_t, a_t) - Q^{\mathcal{L}}(\mathbf{b}_{l_t},a_t) \right| \le \; \gamma \left| Q^{\mathcal{L}}(\mathbf{\mu}_{t+1}, a^{\prime}_{t+1}) - Q^{\mathcal{L}}(\mathbf{\mu}_{l_{t+1}}, a^{\prime}_{t+1}) \right| \\
    		+ r_{\text{max}} \left( 2 \mathcal{H}(n_t,\sigma_t) + 2\gamma\Gamma\mathcal{H}(m_t,\sigma_t) + \left( \mathcal{R} + \gamma \Gamma \Omega \right) \mathcal{S}^t \sum_{j=0}^{t} \frac{1}{\mathcal{S}^j} \mathcal{H}(l_j,\sigma_j) \right)
    \end{multline*}
    
    Using the inequality above recursively, and the stop case from Equation \eqref{eq:qfunc-approx-vector-stop}:
    \begin{multline*}
        \left| Q^{\mathcal{L}}(\mathbf{b}_t, a_t) - Q^{\mathcal{L}}(\mathbf{b}_{l_t},a_t) \right| \le r_{\text{max}}\\
        \sum_{k=0}^{H-1} \gamma^k \left(2 \mathcal{H}(n_{t+k}, \sigma_{t+k}) + 2\gamma\Gamma \mathcal{H}(m_{t+k}, \sigma_{t+k}) + \left( \mathcal{R} + \gamma \Gamma \Omega \right) \mathcal{S}^{t+k} \sum_{j=0}^{t+k} \frac{1}{\mathcal{S}^j} \mathcal{H}(l_{j}, \sigma_{j}) \right)
    \end{multline*}
\end{proof}

Once again, let us count the number of Hoeffding inequalities used to prove this result. Analogously to the count done for Lemma \ref{lemma:true-error}, consider the \textit{immediate} and \textit{recursive} components. Notice that, for the immediate term, $(t+1)+2$ Hoeffding inequalities are used. The first time-step in the recursive term needs $\mathcal{A}\Omega \left((t+2)+2\right)$ Hoeffding inequalities (it should hold for every action and observation, and the $t+1$ Hoeffding inequalities from the belief-error bound now become $t+2$ because the time-step has increased by one). In total, $\sum_{k=0}^{H-1}\mathcal{A}^k\Omega^k \left(t+k+2\right)$ Hoeffding inequalities are needed to prove this bound.

Finally, with the two error terms of the lookahead already bounded, the full lookahead error can be derived:

\begin{theorem}
    The lookahead error $\epsilon_t$ has the following upper bound:
    \begin{multline*}
        \epsilon_t \le r_{\text{max}} \sum_{k=0}^{H-1} \gamma^k \left( 4 \mathcal{H}(n_{t+k},\sigma_{t+k}) + 4\gamma\Gamma\mathcal{H}(m_{t+k},\sigma_{t+k}) + \left( \mathcal{R} + \gamma \Gamma \Omega \right) \mathcal{S}^{t+k} \sum_{j=0}^{t+k} \frac{1}{\mathcal{S}^j} \mathcal{H}(l_j,\sigma_{j}) \right)\\
         + 2 \gamma^H \Gamma r_{\text{max}}
    \end{multline*}
    \label{lemma:lookahead-error-bound}
\end{theorem}
\begin{proof}
    Lemma \ref{lemma:error-split} yields:
    \begin{equation*}
        \epsilon_t \le \max_{a_t \in \mathcal{A}} \left( \left| Q^{\star}(\mathbf{b}_t, a_t) - Q^{\mathcal{L}}(\mathbf{b}_t, a_t) \right| + \left| Q^{\mathcal{L}}(\mathbf{b}_t, a_t) - Q^{\mathcal{L}}(\mathbf{b}_{l_t}, a_t) \right| \right) + \left| Q^{\star}(\mathbf{b}_{l_t}, a_t^{\mathcal{L}}) - Q^{\mathcal{L}}(\mathbf{b}_{l_t}, a_t^{\mathcal{L}}) \right|
    \end{equation*}
    
    Applying Lemmas \ref{lemma:true-error} and \ref{lemma:lookahead-belief-error} to the previous expression concludes the proof.
\end{proof}
    
The error bound estimated in Theorem \ref{lemma:lookahead-error-bound} encompasses all four errors of the quantum-classical lookahead algorithm. They can be easily identified from the expression in Theorem \ref{lemma:lookahead-error-bound}:
\begin{itemize}
    \item \textit{Approximating the reward distribution}: this error term is captured by the $n_{t+k}$ term in the sum.
    
    \item \textit{Approximating the observation distribution}: this error term is captured by the $m_{t+k}$ term in the sum.
    
    \item \textit{Approximating the belief states}: this error is captured by the $l_j$ term of the bound and increases over time due to repeated approximation of the belief state.
    
    \item \textit{The finite horizon of the lookahead}: the lookahead tree has a fixed depth, given by the finite horizon $H$, introducing a systematic error that can't be compensated by increasing the number of samples. It is captured by the $2\gamma^H \Gamma R_{\text{max}}$ term in the bound.
\end{itemize}

Theorem \ref{lemma:lookahead-error-bound} provides an upper bound to the lookahead error of Definition \ref{def:lookahead-error}. This upper bound takes into account the use of multiple sampling operations throughout the lookahead tree: reward sampling, observation sampling, and belief-state sampling. Given that the bounds used for each sampling operator fail with a probability of $\sigma_t$, the bound for the whole lookahead error also has a non-zero chance of failure, which depends on the confidence interval of each sampling operation and is yet to be defined. This is because the confidence interval $1 - \sigma_t$ only holds locally for each sampling operation, and not for the whole expression presented in Theorem \ref{lemma:lookahead-error-bound}. Naturally, if none of the sampling operation bounds fail, the lookahead error bound holds. As such, the probability of failure $\delta_t$ of the lookahead error bound is at least as large as the probability that any one of the sampling bounds fails. 

One way to relate the confidence interval $1 - \delta_t$ of the bound in Theorem \ref{lemma:lookahead-error-bound} to the confidence interval $1 - \sigma_t$ of each sampling operation is by using Theorem \ref{theorem:union-bound}. It states that the probability of the union of multiple events is at least as small as the sum of the probabilities of each individual event:
\begin{theorem}
    (Boole's Inequality \cite{seneta1992history}). Let $Z_1, Z_2, \dots, Z_n$ be probabilistic events that occur with probability $P\left(Z_1\right), P\left(Z_2\right), \dots, P\left(Z_n\right)$. Then:
    \begin{equation*}
        P \left( \bigcup_{i=1}^n Z_i \right) \le \sum_{i=1}^n P\left(Z_i\right)
    \end{equation*}
    \label{theorem:union-bound}
\end{theorem}

In the formulation of Theorem \ref{theorem:union-bound}, the probabilistic events $Z_i$ could have a higher sampling error than the one provided in Theorem \ref{theorem:hoeffding}. By attributing this meaning to the probabilistic events, Boole's inequality can be used to relate the sampling bound's confidence interval $1 - \sigma_t$ to the confidence interval of the algorithm's upper bound presented in Theorem \ref{lemma:lookahead-error-bound}.

\begin{lemma}
    The confidence interval $1 - \delta_t$ of the algorithm's upper bound given in Theorem \ref{lemma:lookahead-error-bound} and the confidence interval $1 - \sigma_t$ of the Hoeffding's bound for each sampling operation obeys the following inequality:
    \begin{equation*}
        \delta_t \le \sum_{i=0}^{H-1} \mathcal{A}^i \Omega^i \left( 2 + \mathcal{A} \left( t + i + 4 \right) \right) \sigma_{t+i}
    \end{equation*}
    \label{lemma:confidence-bound}
\end{lemma}
\begin{proof}
    Recall, from Lemma \ref{lemma:error-split}, that the lookahead error is bounded by the following expression:
    \begin{equation*}
        \epsilon_t \le \max_{a_t \in \mathcal{A}} \left( \left| Q^{\star}(\mathbf{b}_t, a_t) - Q^{\mathcal{L}}(\mathbf{b}_t, a_t) \right| + \left| Q^{\mathcal{L}}(\mathbf{b}_t, a_t) - Q^{\mathcal{L}}(\mathbf{b}_{l_t}, a_t) \right| \right) + \left| Q^{\star}(\mathbf{b}_{l_t}, a_t^{\mathcal{L}}) - Q^{\mathcal{L}}(\mathbf{b}_{l_t}, a_t^{\mathcal{L}}) \right|
    \end{equation*}
    
    The terms $\left| Q^{\star}(\mathbf{b}_t, a_t) - Q^{\mathcal{L}}(\mathbf{b}_t, a_t) \right|$ and $\left| Q^{\star}(\mathbf{b}_{l_t}, a_t^{\mathcal{L}}) - Q^{\mathcal{L}}(\mathbf{b}_{l_t}, a_t^{\mathcal{L}}) \right|$ can both be bounded using Lemma \ref{lemma:true-error}, which needs $2 \sum_{i=0}^{H-1} \mathcal{A}^i \Omega^i$ Hoeffding inequalities to be derived. The term $\left| Q^{\mathcal{L}}(\mathbf{b}_t, a_t) - Q^{\mathcal{L}}(\mathbf{b}_{l_t}, a_t) \right|$ can be bounded using Lemma \ref{lemma:lookahead-belief-error}, which requires $\sum_{i=0}^{H-1} \mathcal{A}^i \Omega^i (t+i+2)$ Hoeffding inequalities. However, two of the terms are under a $\text{max}_{a_t \in \mathcal{A}}$ operation in the lookahead error bound above, so their number of Hoeffding inequalities should be multiplied by a factor $\mathcal{A}$, since these bounds should hold for any possible action. As such, the total number of Hoeffding inequalities to bound the lookahead error is given by:
    \begin{equation*}
        2 \sum_{i=0}^{H-1} \mathcal{A}^i \Omega^i + 2 \mathcal{A} \sum_{i=0}^{H-1} \mathcal{A}^i \Omega^i + \mathcal{A} \sum_{i=0}^{H-1} \mathcal{A}^i \Omega^i (t + i + 2) = \sum_{i=0}^{H-1} \mathcal{A}^i \Omega^i \left( 2 + \mathcal{A} \left( t + i + 4 \right) \right)
    \end{equation*}
    
    Each term of the previous sum represents the total number of Hoeffding inequalities at that particular depth in the lookahead tree (if $i=2$, the corresponding term represents the number of Hoeffding inequalities at depth $2$). Given that, for a depth $i$ at time-step $t$, the probability of a sampling operation failing is $\sigma_{t+1}$, Boole's inequality in Theorem \ref{theorem:union-bound} yields:
    \begin{equation*}
        \delta_t \le \sum_{i=0}^{H-1} \mathcal{A}^i \Omega^i \left( 2 + \mathcal{A} \left( t + i + 4 \right) \right) \sigma_{t+i}
    \end{equation*}
\end{proof}

The bounds derived in Lemmas \ref{lemma:lookahead-error-bound} and \ref{lemma:confidence-bound} depend on the definition of four successions:
\begin{itemize}
    \item $n_t$, $m_t$, and $l_t$: the successions that define how the number of samples for each sampling operation varies over time.
    \item $\sigma_t$: the succession defining how the confidence interval of each sampling operation varies over time.
\end{itemize}

Unfortunately, due to the cumulative nature of the belief-state error, the lookahead error grows exponentially as time progresses (it grows with $\mathcal{S}^t$, and the state size $\mathcal{S}$ of the POMDP can be extremely large by itself!). This effect could be counteracted by an appropriate choice for both $l_t$ and $\sigma_t$, but this work does not try to answer this question. Instead, a much simpler case is considered: both $n_t, m_t, l_t$, and $\sigma_t$ are fixed, such that they do not change over time, as stated in Definition \ref{def:fixed-successions}.

\begin{definition}
    The lookahead algorithm's samples $n_t, m_t$ and $l_t$ and the confidence interval $1-\sigma_t$ are considered to be time-invariant:
    \begin{equation*}
        n_t = n, \quad m_t = m, \quad l_t = l, \quad \sigma_t = \sigma
    \end{equation*}
    \label{def:fixed-successions}
\end{definition}

Moreover, a further simplification is assumed to reduce the degrees of freedom in the choice of the number of samples: the contribution of each sampling operation to the whole lookahead error at time-step $t=0$ is assumed to be the same, such that the number of samples $n$, $m$, and $l$ can be related to each other according to Lemma \ref{lemma:sample-relation}:
\begin{lemma}
    If the contribution to the error bound in Theorem \ref{lemma:lookahead-error-bound} of each sampling operation is considered to be the same at time-step $t=0$, the number of samples can be related to each other according to the following expressions:
    \begin{equation*}
        m = \left( \gamma \Gamma \right)^2 n, \quad l = \left( \frac{1}{4} \frac{\mathcal{R} + \gamma \Gamma \Omega}{\mathcal{S} - 1} \left( \frac{\mathcal{S}}{\Gamma \left(1 - \gamma^H\right)} \frac{\left(\gamma\mathcal{S}\right)^H - 1}{\gamma \mathcal{S} - 1} - 1 \right) \right)^2 n
    \end{equation*}
    \label{lemma:sample-relation}
\end{lemma}
\begin{proof}
    According to Theorem \ref{lemma:lookahead-error-bound} and Definition \ref{def:fixed-successions}, the lookahead error at time $t=0$ has the following upper-bound:
    \begin{multline*}
            \epsilon_0 \le \; r_{\text{max}} \sum_{k=0}^{H-1} \gamma^k \left( 4 \mathcal{H}(n,\sigma) + 4\gamma\Gamma\mathcal{H}(m,\sigma) + \left( \mathcal{R} + \gamma \Gamma \Omega \right) \mathcal{S}^{k} \mathcal{H}(l,\sigma) \sum_{j=0}^{k} \frac{1}{\mathcal{S}^j} \right)\\
            + 2 \gamma^H \Gamma r_{\text{max}} = S_n + S_m + S_l + 2 \gamma^H \Gamma r_{\text{max}}
    \end{multline*}
    
    \noindent where $S_n, S_m$ and $S_l$ represent the sums over the number of samples $n, m$ and $l$, respectively. For their contributions to be the same, as stated in Lemma \ref{lemma:sample-relation}, all of these sums must be equal ($S_n = S_m = S_l$).
    
    Letting $S_n = S_m$\footnote{Recall that $\mathcal{H}(n,\sigma) = \sqrt{ \frac{1}{2 n} \log\left( \frac{2}{\sigma} \right) }$} yields:
    \begin{align*}
        \begin{split}
            4 r_{\text{max}} \mathcal{H}(n,\sigma) \sum_{k=0}^{H-1} \gamma^k &= 4 \gamma \Gamma r_{\text{max}} \mathcal{H}(m,\sigma) \sum_{k=0}^{H-1} \gamma^k\\
            \sqrt{n^{-1}} &= \gamma \Gamma \sqrt{m^{-1}}\\
            m &= \left( \gamma \Gamma \right)^2 n
        \end{split}
    \end{align*}
    
    Finally, making $S_n = S_l$ and resorting to Lemma \ref{lemma:geometric-series}:
    \begin{align*}
        \begin{split}
            4 r_{\text{max}} \mathcal{H}(n,\sigma) \sum_{k=0}^{H-1} \gamma^k =& r_{\text{max}} \mathcal{H}(l,\sigma) \left( \mathcal{R} + \gamma \Gamma \Omega \right) \sum_{k=0}^{H-1} \gamma^k \mathcal{S}^k \sum_{j=0}^{k} \frac{1}{\mathcal{S}^j}\\
            4 \sqrt{n^{-1}} \sum_{k=0}^{H-1} \gamma^k \stackrel{(\ref{lemma:geometric-series})}{=}& \sqrt{l^{-1}} \left( \mathcal{R} + \gamma \Gamma \Omega \right) \sum_{k=0}^{H-1} \gamma^k \frac{\mathcal{S}^{k+1} - 1}{\mathcal{S} - 1}\\
            4 \sqrt{n^{-1}} \Gamma \left( 1- \gamma^H \right) =& \sqrt{l^{-1}} \left( \mathcal{R} + \gamma \Gamma \Omega \right) \frac{1}{\mathcal{S}-1} \left( \mathcal{S} \sum_{k=0}^{H-1} \left( \gamma \mathcal{S} \right)^k - \sum_{k=0}^{H-1} \gamma^k \right)\\
            4 \sqrt{n^{-1}} \Gamma \left( 1- \gamma^H \right) \stackrel{(\ref{lemma:geometric-series})}{=}& \sqrt{l^{-1}} \left( \mathcal{R} + \gamma \Gamma \Omega \right) \frac{1}{\mathcal{S}-1} \left( \mathcal{S} \frac{\left(\gamma\mathcal{S}\right)^H - 1}{\gamma \mathcal{S} - 1} - \Gamma \left( 1 - \gamma^H \right) \right)\\
            l =& \left( \frac{1}{4} \frac{\mathcal{R} + \gamma \Gamma \Omega}{\mathcal{S} - 1} \left( \frac{\mathcal{S}}{\Gamma \left(1 - \gamma^H\right)} \frac{\left(\gamma\mathcal{S}\right)^H - 1}{\gamma \mathcal{S} - 1} - 1 \right) \right)^2 n
        \end{split}
    \end{align*}
\end{proof}

Given Definition \ref{def:fixed-successions} and Lemma \ref{lemma:sample-relation}, the bounds in Lemmas \ref{lemma:lookahead-error-bound} and \ref{lemma:confidence-bound} can be re-written as follows:

\begin{theorem}
    The algorithm's error $\epsilon_t$ has the following upper bound, with a confidence interval of $1 - \delta_t$:
    \begin{align*}
        \begin{split}
            \epsilon_t \le& \; r_{\text{max}} \sqrt{ \frac{1}{2 n} \log \left( \frac{2}{\sigma} \right) } \left( 8 \Gamma + 4 \mathcal{S}^t \right) + 2 \gamma \Gamma r_{\text{max}}\\
            \delta_t \le& \; \sigma \left( \mathcal{A} \Omega \right)^{H} \left( t + \mathcal{A} \left( 8 + \Omega \frac{H-1}{\left( \mathcal{A} \Omega - 1 \right)^2} \right) \right)
        \end{split}
    \end{align*}
    \label{theorem:algorithm-bounds}
\end{theorem}
\begin{proof}
    The bound on error $\epsilon_t$ follows directly from Theorem \ref{lemma:lookahead-error-bound}, by substituting $m$ and $l$ using the expressions in Lemma \ref{lemma:sample-relation}.
    
    As for the bound for the confidence $\delta_t$, it follows from Lemma \ref{lemma:confidence-bound} that:
    \begin{align*}
        \begin{split}
            \delta_t \le& \; \sum_{i=0}^{H-1} \mathcal{A}^i \Omega^i \left( 2 + \mathcal{A} \left(t + i + 4\right) \right) \sigma_{t+i}\\
            =& \; \sigma \sum_{i=0}^{H-1} \mathcal{A}^i \Omega^i \left( 2 + \mathcal{A} \left(t + 4\right) + i \mathcal{A} \right)
        \end{split}
    \end{align*}
    
    By using the geometric series of Lemma \ref{lemma:geometric-series}:
    \begin{align*}
        \begin{split}
            \delta_t \stackrel{(\ref{lemma:geometric-series})}{\le}& \; \sigma \left( \mathcal{A} \sum_{i=0}^{H-1} i \mathcal{A}^i \Omega^i + \left( 2 + \mathcal{A} \left(t + 4\right) \right) \frac{\left(\mathcal{A}\Omega\right)^H - 1}{\mathcal{A}\Omega - 1} \right) \\
            \le& \; \sigma \left( \mathcal{A} \sum_{i=0}^{H-1} i \mathcal{A}^i \Omega^i + \left( 2 + \mathcal{A} \left(t + 4\right) \right) \left(\mathcal{A}\Omega\right)^H \right)
        \end{split}
    \end{align*}
    
    Analogously applying a variation of the finite geometric series given by Lemma \ref{lemma:other-geometric-series} to the sum $\sum_{i=0}^{H-1}i\mathcal{A}^i\Omega^i$, the bound becomes:
    \begin{align*}
        \begin{split}
            \delta_t \stackrel{(\ref{lemma:other-geometric-series})}{\le}& \; \sigma \left( \mathcal{A} \frac{\left(H-1\right)\left(\mathcal{A}\Omega\right)^{H+1} - H \left( \mathcal{A}\Omega \right)^H + \mathcal{A}\Omega}{\left(\mathcal{A}\Omega-1\right)^2} + \left( 2 + \mathcal{A} \left(t + 4\right) \right) \left(\mathcal{A}\Omega\right)^H \right)\\
            \le& \; \sigma \left( \mathcal{A} \frac{\left(H-1\right)\left(\mathcal{A}\Omega\right)^{H+1}}{\left(\mathcal{A}\Omega-1\right)^2} + \left( 2 + \mathcal{A} \left(t + 4\right) \right) \left(\mathcal{A}\Omega\right)^H \right)\\
            \le& \; \sigma \left(\mathcal{A}^2\Omega\right)^H \left( \mathcal{A}\Omega \frac{\left(H-1\right)}{\left(\mathcal{A}\Omega-1\right)^2} + \mathcal{A} \left(t + 4\right) + 2 \right)\\
            \le& \; \sigma \left( \mathcal{A} \Omega \right)^{H} \left( 2 + \mathcal{A} \left( t + 4 + \mathcal{A} \Omega \frac{H-1}{\left( \mathcal{A} \Omega - 1 \right)^2} \right) \right)
        \end{split}
    \end{align*}
\end{proof}

As can be seen in Theorem \ref{theorem:algorithm-bounds}, both the bounds for $\epsilon_t$ and $\delta_t$ grow over time. Suppose that the designed algorithm accepts three arguments as input: the maximum error $\epsilon$, the minimum confidence interval $1 - \delta$, and the stopping time $T$ for the decision-making algorithm. In this scenario, both $\epsilon_t$ and $\delta_t$ will have a maximum bound for $t=T$. To ensure the error $\epsilon$ and the confidence interval $1-\delta$, Theorem \ref{theorem:algorithm-bounds} can be re-expressed to derive a lower-bound for the number of samples and an upper bound for $\sigma$ (the probability of failure for the sampling bounds):
\begin{align}
    \begin{split}
        n \ge& \; \frac{1}{2} \log\left(\frac{2}{\sigma}\right) \left( \frac{r_{\text{max}}}{\epsilon - 2\gamma^H\Gamma r_{\text{max}}} \left( 8\Gamma + 4 \mathcal{S}^T \right) \right)^2\\
        \sigma \le& \; \delta \left( \mathcal{A} \Omega \right)^{-H} \left( 2 + \mathcal{A} \left( T + 4 + \mathcal{A} \Omega \frac{H-1}{\left( \mathcal{A} \Omega - 1 \right)^2} \right) \right)^{-1}
    \end{split}
    \label{eq:sample-confidence-bounds}
\end{align}

The above conditions, however, are not enough to ensure the error $\epsilon$. This is because no matter how many samples one extracts for approximating the probability distributions, the number of samples can't compensate for the error of performing a finite lookahead if the chosen horizon $H$ is too shallow \footnote{Using a finite lookahead introduces an error of $2\gamma\Gamma r_{\text{max}}$, as in Theorem \ref{theorem:algorithm-bounds}.}. As such, it should also be assured that the lookahead horizon $H$ is large enough to make the desired error $\epsilon$ attainable. This condition is met when $2\gamma^H\Gamma r_{\text{max}} < \epsilon$, such that:
\begin{equation}
    H > \log_{\gamma} \left( \frac{\epsilon}{2 \Gamma r_{\text{max}}} \right)
    \label{eq:horizon-bound}
\end{equation}

The three bounds given by Equations \eqref{eq:sample-confidence-bounds} and \eqref{eq:horizon-bound} sum up the whole sample complexity analysis of the lookahead algorithm. They are captured in Theorem \ref{theorem:sample-complexity}, the main theorem of this analysis:

\begin{theorem}
    Consider a stopping time $T$, horizon $H$, and Definition \ref{def:fixed-successions} and Lemma \ref{lemma:sample-relation} for the lookahead algorithm. A near-optimal action is guaranteed to be taken at each time step of the decision-making process with a regret of at most $\epsilon$ and confidence interval of at least $1-\delta$ if all of the following conditions are met:
    \begin{align*}
        \begin{split}
            n \ge& \; \frac{1}{2} \log\left(\frac{2}{\sigma}\right) \left( \frac{r_{\text{max}}}{\epsilon - 2\gamma^H\Gamma r_{\text{max}}} \left( 8\Gamma + 4 \mathcal{S}^T \right) \right)^2\\
            \sigma \le& \; \delta \left( \mathcal{A} \Omega \right)^{-H} \left( 2 + \mathcal{A} \left( T + 4 + \mathcal{A} \Omega \frac{H-1}{\left( \mathcal{A} \Omega - 1 \right)^2} \right) \right)^{-1}\\
            H >& \; \log_{\gamma} \left( \frac{\epsilon}{2 \Gamma r_{\text{max}}} \right)
        \end{split}
    \end{align*}
    \label{theorem:sample-complexity}
\end{theorem}

\subsection{Time Complexity}
\label{appsub:time_comp}

To calculate the computational complexity of the lookahead algorithm, in both the classical and quantum-classical versions, one must multiply the total number of samples required by the computational complexity necessary to compute each sample. Given that the computational complexity for each sample has already been defined, the only task left is to define the total number of samples that need to be extracted.

In order to do so, it is important to recall that there are three types of sampling operations that the algorithm performs:
\begin{itemize}
    \item \textit{Reward sampling}: this sampling operation is performed at every observation node and requires $n$ samples.
    \item \textit{Observation sampling}: performed at every observation node and requires $m$ samples.
    \item \textit{Belief sampling}: performed at every belief node and requires $l$ samples.
\end{itemize}

As such, the total number of samples $N_s$ can be defined in terms of the number of belief nodes $N_b$ and the number of observation nodes $N_o$ as:
\begin{equation*}
    N_s = l N_b + \left(n + m\right) N_o
\end{equation*}

\noindent where the number of belief nodes is computed as
\begin{equation*}
    N_b = \sum_{i=0}^{H-1} \left( \mathcal{A}\Omega \right)^i
    = \frac{\left( \mathcal{A}\Omega \right)^H - 1}{ \mathcal{A}\Omega - 1}
\end{equation*}

\noindent and the number of observation nodes as
\begin{equation*}
    N_o = \mathcal{A} \sum_{i=0}^{H-1} \left( \mathcal{A}\Omega \right)^i
    = \mathcal{A} \frac{\left( \mathcal{A}\Omega \right)^H - 1}{ \mathcal{A}\Omega - 1}
\end{equation*}

\noindent such that the total number of samples can be re-expressed as
\begin{equation*}
    N_s = \frac{\left( \mathcal{A} \Omega \right)^H - 1}{\mathcal{A} \Omega - 1} \left( l + \mathcal{A} \left( n + m \right) \right)
\end{equation*}

Let $C_n$, $C_m$, and $C_l$ represent the computational complexity of extracting each of the three types of samples considered. In asymptotic terms, where $\mathcal{S}$, $\mathcal{A}$, $\mathcal{R}$, $\Omega$, $H$, and $T$ are all assumed to be very large, the computational complexity of the algorithm can be expressed as:
\begin{equation*}
    \mathcal{O} \left( \mathcal{A}^{H-1} \Omega^{H-1} \left( l C_l + \mathcal{A} \left( n C_n + m C_m \right) \right) \right)
\end{equation*}

Recalling the relation between the number of samples $n$, $m$ and $l$ established in Lemma \ref{lemma:sample-relation}, this complexity becomes:
\begin{align}
    \begin{split}
        &\mathcal{O} \left( n \mathcal{A}^{H-1} \Omega^{H-1} \left( \left( \frac{1}{4} \frac{\mathcal{R} + \gamma \Gamma \Omega}{\mathcal{S} - 1} \left( \frac{\mathcal{S}}{\Gamma \left(1 - \gamma^H\right)} \frac{\left(\gamma\mathcal{S}\right)^H - 1}{\gamma \mathcal{S} - 1} - 1 \right) \right)^2 C_l + \mathcal{A} \left( C_n + \gamma^2 \Gamma^2 C_m \right) \right) \right)\\
        =& \; \mathcal{O} \left( n \mathcal{A}^{H-1} \Omega^{H-1} \left( \mathcal{A} \left( C_n + C_m \right) + \left( \left(\mathcal{R} + \Omega \right)\mathcal{S}^{H-1} \right)^2 C_l \right) \right)
    \end{split}
    \label{eq:comp-C}
\end{align}

To substitute the computational complexities $C_n$, $C_m$, and $C_l$, a distinction has to be made for the classical and quantum-classical algorithms, since they use different inference algorithms to extract probability distributions.

In the classical case, direct sampling is used for performing inference on the rewards and observations, while belief updating uses rejection sampling. The computational complexities are of order $\mathcal{O}\left(NM\right)$ for direct sampling and $\mathcal{O}\left(NMP(e)^{-1}\right)$ for rejection sampling.

What does the evidence stand for in the case of rejection sampling? As this operation is performed at every observation node in the lookahead tree, to account for all these sampling operations, it is the mean inverse probability of evidence for belief updating. Since belief updating only uses observation RVs as evidence, then it is given by $c_l = \sum_{o_{t+1}} P(o_{t+1}|\mathbf{b}_t,a_t)^{-1}$, where the sum is over all observations in the lookahead tree, conditioned on the previous belief state $\mathbf{b}_t$ and previous action $a_t$.

For the quantum-classical case, direct sampling is still used for reward and observation sampling, having the same complexity as the classical case. Quantum rejection sampling is, however, applied for belief updating. Therefore, in the quantum-classical case, the complexity for $C_l$ is of the order $\mathcal{O}\left(N 2^M q_l\right)$, where $q_l = P(e)^{-\frac{1}{2}} = \sum_{o_{t+1}} P(o_{t+1}, \mathbf{b}_t, a_t)^{-\frac{1}{2}}$. In this scenario, $c_l$ is no longer adequate to represent $P(e)^{-\frac{1}{2}}$, because the latter is an average of the inverse square rooted probabilities $P\left(o_{t+1}|\mathbf{b}_t,a_t\right)^{-\frac{1}{2}}$ and not of the inverse probabilities $P\left(o_{t+1}|\mathbf{b}_t,a_t\right)^{-1}$.

It is also important to mention that the complexity of the algorithm can be written both in terms of the number of samples $n$ or in terms of the precision $\epsilon$ and confidence interval $1-\delta$, since they can be related according to Theorem \ref{theorem:sample-complexity}:
\begin{equation}
	\mathcal{O}\left(n\right) \le \mathcal{O}\left(\left(\frac{\mathcal{S}^T}{\epsilon}\right)^2\log\left(\frac{1}{\sigma}\right)\right) \le \mathcal{O}\left( \left( \frac{\mathcal{S}^T}{\epsilon} \right)^2 \log \left( \left( T + \frac{H}{\mathcal{A}\Omega} \right) \frac{\mathcal{A}^{H+1} \Omega^H}{\delta} \right) \right)
	\label{eq:sample-precision-relation}
\end{equation}

In this context, Theorem \ref{theorem:classical-complexity} characterizes the computational complexity of the classical lookahead algorithm:
\begin{theorem}
    The computational complexity of the classical lookahead algorithm is given by any of the following expressions:
    \begin{gather*}
        \mathcal{O} \left( n N M \mathcal{A}^{H-1} \Omega^{H-1} \left( \mathcal{A} + \left( \left(\mathcal{R} + \Omega \right)\mathcal{S}^{H-1} \right)^2 c_l \right) \right)\\
        \mathcal{O} \left( N M \mathcal{A}^{H-1} \Omega^{H-1} \left( \mathcal{A} + \left( \left(\mathcal{R} + \Omega \right)\mathcal{S}^{H-1} \right)^2 c_l \right) \left( \frac{\mathcal{S}^T}{\epsilon} \right)^2 \log \left( \left( T + \frac{H}{\mathcal{A}\Omega} \right) \frac{\mathcal{A}^{H+1} \Omega^H}{\delta} \right) \right)
    \end{gather*}
    \label{theorem:classical-complexity}
\end{theorem}
\begin{proof}
    This result is directly attained with Equations \eqref{eq:comp-C} and \eqref{eq:sample-precision-relation} and using $C_n = C_m = \mathcal{O}(NM)$ and $C_l = \mathcal{O}\left(N M c_l\right)$.
\end{proof}

Theorem \ref{theorem:quantum-complexity} is analogous to the previous theorem, but characterizes the computational complexity of the quantum-classical lookahead algorithm:
\begin{theorem}
    The computational complexity of the quantum-classical lookahead algorithm is given by any of the following expressions, whenever $2^M \approx M$:
    \begin{gather*}
        \mathcal{O} \left( n N M \mathcal{A}^{H-1} \Omega^{H-1} \left( \mathcal{A} + \left( \left(\mathcal{R} + \Omega \right)\mathcal{S}^{H-1} \right)^2 q_l \right) \right)\\
        \mathcal{O} \left( N M \mathcal{A}^{H-1} \Omega^{H-1} \left( \mathcal{A} + \left( \left(\mathcal{R} + \Omega \right)\mathcal{S}^{H-1} \right)^2 q_l \right) \left( \frac{\mathcal{S}^T}{\epsilon} \right)^2 \log \left( \left( T + \frac{H}{\mathcal{A}\Omega} \right) \frac{\mathcal{A}^{H+1} \Omega^H}{\delta} \right) \right)
    \end{gather*}
    \label{theorem:quantum-complexity}
\end{theorem}
\begin{proof}
    This result is directly attained with Equations \eqref{eq:comp-C} and \eqref{eq:sample-precision-relation} and using $C_n = C_m = \mathcal{O}(NM)$ and $C_l = \mathcal{O}\left(N 2^M q_l\right)$.
\end{proof}

As can be seen from Theorems \ref{theorem:classical-complexity} and \ref{theorem:quantum-complexity}, a partial speedup is obtained in the quantum case for the belief-state sampling operations, while the computational complexity of the rest of the algorithm is similar across both classical and quantum algorithms. This is so whenever the dynamic decision networks at hand are sparse enough (such that $2^M \approx M$), just as in the case of quantum rejection sampling. However, for some specific situations, the portion of the algorithm that uses a quantum subroutine dominates the complexity of the algorithm so much that the contribution of the classical subroutines to the computational complexity of the algorithm is negligible. In these scenarios, detailed in Corollaries \ref{cor:classical-complexity} and \ref{cor:quantum-complexity}, the quantum-classical lookahead has the greatest possible speedup over its classical counterpart:

\begin{corollary}
    Under the assumption that $\frac{1}{q_l} \ll \frac{\left( \left(\mathcal{R} + \Omega\right) \mathcal{S}^{H-1} \right)^2}{\mathcal{A}}$, the expressions for the computational complexity of the classical lookahead algorithm of Theorem \ref{theorem:classical-complexity} can be re-written as:
    \begin{gather*}
        \mathcal{O} \left( n N M \mathcal{A}^{H-1} \Omega^{H-1} \left( \left(\mathcal{R} + \Omega \right) \mathcal{S}^{H-1} \right)^2 c_l \right)\\
        \mathcal{O} \left( N M \mathcal{A}^{H-1} \Omega^{H-1} \left( \left(\mathcal{R} + \Omega \right) \frac{\mathcal{S}^{T+H-1}}{\epsilon} \right)^2 c_l \log \left( \left( T + \frac{H}{\mathcal{A}\Omega} \right) \frac{\mathcal{A}^{H+1} \Omega^H}{\delta} \right) \right)
    \end{gather*}
    \label{cor:classical-complexity}
\end{corollary}
\begin{proof}
    Given that $q_l \le c_l$, it follows that $\frac{1}{c_l} \le \frac{1}{q_l}$, and therefore $\frac{1}{q_l} \ll \frac{\left( \left(\mathcal{R} + \Omega\right) \mathcal{S}^{H-1} \right)^2}{\mathcal{A}} \Rightarrow \frac{c_n}{c_l} \ll \frac{\left( \left(\mathcal{R} + \Omega\right) \mathcal{S}^{H-1} \right)^2}{\mathcal{A}}$. Thus, the result of this corollary follows directly from Theorem \ref{theorem:classical-complexity}.
\end{proof}

\begin{corollary}
    Under the assumption that $\frac{1}{q_l} \ll \frac{\left( \left(\mathcal{R} + \Omega\right) \mathcal{S}^{H-1} \right)^2}{\mathcal{A}}$, the expressions for the computational complexity of the quantum-classical lookahead algorithm of Theorem \ref{theorem:quantum-complexity} can be re-written as:
    \begin{gather*}
        \mathcal{O} \left( n N 2^M \mathcal{A}^{H-1} \Omega^{H-1} \left( \left(\mathcal{R} + \Omega \right) \mathcal{S}^{H-1} \right)^2 q_l \right)\\
        \mathcal{O} \left( N 2^M \mathcal{A}^{H-1} \Omega^{H-1} \left( \left(\mathcal{R} + \Omega \right) \frac{\mathcal{S}^{T+H-1}}{\epsilon} \right)^2 q_l \log \left( \left( T + \frac{H}{\mathcal{A}\Omega} \right) \frac{\mathcal{A}^{H+1} \Omega^H}{\delta} \right) \right)
    \end{gather*}
    \label{cor:quantum-complexity}
\end{corollary}

Corollaries \ref{cor:classical-complexity} and \ref{cor:quantum-complexity} compare the classical and quantum algorithm's complexity in the best-case scenario for the quantum algorithm, where $\frac{c_n}{q_l} \ll \frac{\left( \left(\mathcal{R} + \Omega\right) \mathcal{S}^{H-1} \right)^2}{\mathcal{A}}$. It remains to be answered as to what RL problems this assumption can be applied to in practice, as this work does not address this question. Nonetheless, it gives a benchmark as to what speedup one can expect from the quantum algorithm in a best-case scenario.

What speedups can therefore be expected? Although one might initially think that using quantum rejection sampling would provide a quadratic speedup, this is not always the case. Dividing the classical complexity by the quantum complexity (and assuming that $M$ and $2^M$ have about the same order of magnitude), this fraction reduces to $\frac{c_l}{q_l} = \frac{\sum_{x} P^{-1}(x)}{\sum_{x} P^{-\frac{1}{2}} (x)}$, where the sum over $x$ represents the sum over the evidence. From inequality $\sqrt{\sum_x a_x} \le \sum_{x} \sqrt{a_x} \le \sum_{x} a_x$, it is easy to derive that:

\begin{equation}
    1 \le \frac{c_l}{q_l} \le q_l
    \label{eq:speedup-inequality}
\end{equation}

The inequality in Equation \eqref{eq:speedup-inequality} entails that the complexity of the quantum algorithm relative to the classical one can range between no speedup to a quadratic one, but it is most likely to fall somewhere in between these two extremes. This comes from the notion that, although using quantum rejection sampling yields a quadratic advantage over classical rejection sampling, the same cannot be said about using it sequentially, due to the mathematical property that the sum of square roots can be (and usually is) larger than the square root of the sum.

\subsection{Auxiliary lemmas}
\label{appsub:aux_lemmas}

To derive an upper bound on the lookahead algorithm's error of Definition \ref{def:lookahead-error}, some lemmas are of use. One example is Lemma \ref{lemma:max-diff}, an inequality that relates the difference of the maxima of two vectors and the maximum of their difference:
\begin{lemma}
    Let $\mathbf{U}, \mathbf{V} \in \mathcal{W}$ be vectors in some vector space $\mathcal{W}$. Then:
    \begin{equation*}
        \lVert \mathbf{U} \rVert_{\infty} - \lVert \mathbf{V} \rVert_{\infty} \le \lVert \mathbf{U} - \mathbf{V} \rVert_{\infty}
    \end{equation*}
    \label{lemma:max-diff}
\end{lemma} 
\begin{proof}
    Let $U_i$ and $V_i$ represent entries of $\mathbf{U}$ and $\mathbf{V}$, respectively. Then:
    \begin{align*}
        \begin{split}
            U_i &\le \left| U_i - V_i \right| + V_i\\
            \max_i U_i &\le \max_i \left( \left| U_i - V_i \right| + V_i \right)\\
            \max_i U_i &\le \max_i \left| U_i - V_i \right| + \max_i V_i\\
            \max_i U_i - \max_i V_i &\le \max_i \left| U_i - V_i \right|\\
            \lVert \mathbf{U} \rVert_{\infty} - \lVert \mathbf{V} \rVert_{\infty} &\le \lVert \mathbf{U} - \mathbf{V} \rVert_{\infty}
        \end{split}
    \end{align*}
\end{proof}

Another result that is used repeatedly throughout the proofs is a generalization of the triangle's inequality, stated in Lemma \ref{lemma:triangle-ineq}:
\begin{lemma}
    Let $\mathbf{U}, \mathbf{V} \in \mathcal{W}$ be vectors in an Euclidean space $\mathcal{W}$. Then:
    \begin{equation*}
        \lVert \mathbf{U} + \mathbf{V} \rVert \le \lVert \mathbf{U} \rVert + \lVert \mathbf{V} \rVert
    \end{equation*}
    \label{lemma:triangle-ineq}
\end{lemma}
\begin{proof}
    \begin{equation*}
        \lVert \mathbf{U} + \mathbf{V} \rVert = \sqrt{ \lVert \mathbf{U} \rVert^2 + \lVert \mathbf{V} \rVert^2 + 2 \mathbf{U}^{\top} \mathbf{V} }
    \end{equation*}
    
    Using the Cauchy-Schwarz inequality, stating that $\left| \mathbf{U}^{\top}\mathbf{V} \right| \le \lVert \mathbf{U} \rVert \lVert \mathbf{V} \rVert$, we get:
    \begin{align*}
        \begin{split}
            \lVert \mathbf{U} + \mathbf{V} \rVert &\le \sqrt{ \lVert \mathbf{U} \rVert^2 + \lVert \mathbf{V} \rVert^2 + 2 \lVert \mathbf{U} \rVert \lVert \mathbf{V} \rVert}\\
            \lVert \mathbf{U} + \mathbf{V} \rVert &\le \lVert \mathbf{U} \rVert + \lVert \mathbf{V} \rVert
        \end{split}
    \end{align*}
\end{proof}

Lemma \ref{lemma:geometric-series} simplifies a  finite geometric series into a fraction:
\begin{lemma}
	The finite geometric series can be simplified according to the following expression \cite{herman2016calculus}:
	\begin{equation*}
		\sum_{i = 0}^N z^i = \frac{z^N - 1}{z - 1}
	\end{equation*}	
	\label{lemma:geometric-series}
\end{lemma}

Finally, a variation of the finite geometric series is also simplified in Lemma \ref{lemma:other-geometric-series}:
\begin{lemma}
	A variation of the finite geometric series of Lemma \ref{lemma:geometric-series} can be simplified according to the following expression \cite{herman2016calculus}:
	\begin{equation*}
		\sum_{i = 0}^N i z^i = \frac{(N-1)z^{N+1} - Nz^N + z}{(z - 1)^2}
	\end{equation*}	
	\label{lemma:other-geometric-series}
\end{lemma}

\end{document}